\documentclass[twoside, runningheads, oribibl, envcountsame]{llncs}
\newcommand{\qedhere}{\qed}

\usepackage{silence}
\usepackage{lineno}
\usepackage{amsmath}
\usepackage{amssymb}
\usepackage{mathtools}

\usepackage{cite} % to compress citation ranges such as [4-7,9]

\usepackage[hyperfootnotes=false]{hyperref}
\usepackage[dvipsnames]{xcolor}

\usepackage{xspace}

\usepackage{nicefrac}
\usepackage{stmaryrd} % probably lightning
\let\merge\undefined
\usepackage[misc]{ifsym} % letter symbol
\usepackage{cancel}

\usepackage{tikz}
\usetikzlibrary{matrix}
\usetikzlibrary{calc}
\usetikzlibrary{positioning}
\usetikzlibrary{scopes}
\usetikzlibrary{chains}
\usetikzlibrary{matrix}
\usetikzlibrary{decorations.pathmorphing, decorations.pathreplacing, decorations.shapes} 

\usepackage{graphicx}
\usepackage{caption}
\usepackage{subcaption}
\usepackage{wrapfig}

\usepackage[inline]{enumitem}

\usepackage{prettyref}
\newcommand{\rref}[2][]{\prettyref{#2}}
\newrefformat{sec}{Section\,\ref{#1}}
\newrefformat{def}{Def.\,\ref{#1}}
\newrefformat{thm}{Theorem\,\ref{#1}}
\newrefformat{prop}{Proposition\,\ref{#1}}
\newrefformat{lem}{Lemma\,\ref{#1}}
\newrefformat{cor}{Corollary\,\ref{#1}}
\newrefformat{ex}{Example\,\ref{#1}}
\newrefformat{tab}{Table\,\ref{#1}}
\newrefformat{fig}{Fig.\,\ref{#1}}
\newrefformat{app}{Appendix\,\ref{#1}}
\newrefformat{itm}{item \ref{#1}\@\xspace}
\newrefformat{eq}{equation (\ref{#1})}
\newrefformat{rem}{Remark\,\ref{#1}}

\makeatletter
\newcommand\footnoteref[1]{\protected@xdef\@thefnmark{\ref*{#1}}\@footnotemark}
\makeatother

\xspaceaddexceptions{\footnoteref}

\usepackage{ifthen}
\usepackage{xkeyval}
\usepackage{texhelper}

\usepackage{maths}
\usepackage[smallsig]{dlchpsyntax}
\usepackage{semantics}
\usepackage{calculus}
\usepackage[usnames]{misc}
\usepackage{uniform}

\usepackage{xproofs} % this is a improved version of the bussproofs package

\spnewtheorem*{notation}{Notation}{\itshape}{\rmfamily}

\let\oldsubsubsection\subsubsection
\renewcommand*{\subsubsection}[1]{\oldsubsubsection{#1.}}

\newcommand{\citeDLCHP}[1][]{%
	\ifthenelse{\equal{#1}{}}
	{\cite{BriegerMP2023}}{\cite[#1]{BriegerMP2023}}\xspace}
\newcommand{\citeReport}{\cite{BriegerMP2023unisubs}}

\newif\iflongversion\longversionfalse
\newif\ifreport\reporttrue

\newcommand{\reportonly}[1]{\ifreport#1\else\fi{}}

\begin{document}

\title{Uniform Substitution for Dynamic Logic with Communicating Hybrid Programs}
\titlerunning{Uniform Substitution for Communicating Hybrid Programs}

\author{Marvin Brieger\inst{1}%\textsuperscript{\Letter}
\orcidID{0000-0001-9656-2830}, Stefan Mitsch\inst{2}\orcidID{0000-0002-3194-9759}, Andr\'e Platzer\inst{2,3}\orcidID{0000-0001-7238-5710}}

\institute{
  LMU Munich, Germany \\
  \and
  Carnegie Mellon University, Pittsburgh, USA \\
  \and
  Karlsruhe Institute of Technology, Germany \\
  \email{marvin.brieger@sosy.ifi.lmu.de}\quad \email{smitsch@cs.cmu.edu}\quad \email{platzer@kit.edu}
}

\maketitle

\begin{abstract}
	This paper introduces a uniform substitution calculus for \dLCHP, the dynamic logic of communicating hybrid programs.
	Uniform substitution enables parsimonious prover kernels by using axioms instead of axiom schemata.
	Instantiations can be recovered from a single proof rule responsible for soundness-critical instantiation checks rather than being spread across axiom schemata in side conditions.
	Even though communication and parallelism reasoning are notorious for necessitating subtle soundness-critical side conditions, uniform substitution when generalized to \dLCHP manages to limit and isolate their conceptual overhead.
	Since uniform substitution has proven to simplify the implementation of hybrid systems provers substantially,
	uniform substitution for \dLCHP paves the way for a parsimonious implementation of theorem provers for hybrid systems with communication and parallelism.
\end{abstract}


\section{Introduction} \label{sec:intro}

\begin{wrapfigure}[7]{r}{.3\textwidth}
	\begin{center}
		\vspace*{-3em}
		\begin{prooftree*}
			\Axiom{$[ \alpha ] \varphi$}
		
			\Axiom{$[ \beta ] \psi$}
		
			\RightLabel{\sidecondition[color=dark]{($\star\star$)}}
			\BinaryInf{$[ \alpha \parOp \beta ] (\varphi \wedge \psi)$}
		\end{prooftree*}
	\end{center}
	\vspace*{-1em}
	\caption{The proof rule is only sound under subtle side conditions \sidecondition[color=dark]{($\star\star$)}.}
	\label{fig:everyone}
\end{wrapfigure}

Hybrid systems and parallel systems are notoriously subtle to analyze.
Combining both not only culminates these subtleties
but is further complicated because parallel hybrid systems are interlocked by synchronization in a  shared global time.
The \emph{dynamic logic of communicating hybrid programs} \dLCHP \citeDLCHP tames the complexity of parallel hybrid systems providing a compositional proof calculus that disentangles reasoning into purely discrete, continuous, and communication pieces.
However, the calculus is subject to schematic side conditions whose implementation is generally error-prone
causing large soundness-critical code bases \cite{DBLP:journals/jar/Platzer17}.
In particular, compositional reasoning about parallelism as in the idealized proof rule in \rref{fig:everyone}
holds the challenge to exhaustively characterize \emph{all} side conditions required to make \emph{all} instances of this proof rule sound.
Proof systems for discrete parallelism \cite{AcHoare_Zwiers,Xu1997,deRoever2001,AptdeBoerOlderog10,OwickiGries1976,LevinGries1981} already have complicated side conditions, but complexity only increases with continuous interactions in shared global time.

In order to compositionally support compositional reasoning for parallel hybrid systems, this paper generalizes Church's uniform substitution \cite{Church1956} and develops a uniform substitution calculus \cite{DBLP:journals/jar/Platzer17,DBLP:conf/cade/Platzer18,DBLP:conf/cade/Platzer19} for \dLCHP.
Uniform substitution modularizes the calculus itself enabling its parsimonious implementation.
Although applicable to discrete parallelism,
the \dLCHP development resolves the inherent challenge that parallel hybrid systems always synchronize in time.

Uniform substitution adopts a finite list of concrete formulas as axioms instead of an infinite set of formulas via axiom schemata with side conditions. 
This enables theorem provers without the extensive algorithmic checks otherwise required for each schema to sort out unsound instances.
Thanks to the proof rule~\RuleName{US} for uniform substitution,
only sound instances derive from the axioms
such that the parallel composition rule in \dLCHP could be adopted almost literally as above, but with all the soundness-critical checking encapsulated solely in rule~\RuleName{US}.
Thanks to \RuleName{US}'s checking, parallel systems reasoning even reduces to a single parallel injection axiom
\([ \alpha ] \psi \rightarrow [ \alpha \parOp \beta ] \psi\) that merely describes the preservation of property $\psi$ of one parallel component $\alpha$ in the parallel system $\alpha \parOp \beta$.
Proofs about $\alpha \parOp \beta$ reduce to a sequence of property embeddings with this axiom from local abstractions of the subcomponents, which combine soundly due to \RuleName{US}.

Soundness checks in uniform substitution are ultimately determined by the binding structures as identified in the static semantics.
The development of uniform substitution for \dLCHP is, therefore, grounded in the following key observation:
Communication and parallelism both cause additional binding structure that needs attention in the substitution process performed by rule \RuleName{US}:
\begin{enumerate}[label=(B \Roman*), leftmargin=0em,itemindent=!, align=left, itemsep=.5em]
	\item Expressions depend on communication along (co)finite channel sets (besides finitely many free variables),
	which, by the core substitution principle \cite{Church1956}, must not be introduced free into contexts where they are written. \label{itm:channels}%	
	\item Subprograms in a parallel context need to be restricted in the variables and channels written as compositional proof rules for parallelism require local abstractions of subprograms not depending on the internals of the context \cite{deRoever2001}.
	%without knowledge of the context's internals \cite{deRoever2001}. 
	\label{itm:abstraction}
\end{enumerate}

Grounded in the need for abstraction \ref{itm:abstraction},
\([ \alpha ] \psi \rightarrow [ \alpha \parOp \beta ] \psi\) can only be adopted as a sound axiom schema if $\alpha$ and $\beta$ do not share state,
and if program~$\beta$ does not interfere with the contract $\psi$, \iest
\begin{enumerate*}[label={(\roman*)}]
	\item $\psi$ has no free variables bound by~$\beta$ 
	(with exceptions), and
	\item $\psi$ does not depend on communication channels written by $\beta$ (except for channels joint with $\alpha$).
\end{enumerate*}
This extensive side condition would need nontrivial soundness-critical implementations of \dLCHP axiom schemata.
Still, uniform substitution can be lifted with only small changes locally checking for clashes with written channels, and prohibited variables or channels.

The modularity of uniform substitution is the key to the parsimonious implementation \cite{DBLP:series/lncs/MitschP20} of the theorem prover \keymaerax \cite{DBLP:conf/cade/FultonMQVP15} for differential dynamic logic \dL and differential game logic \dGL \cite{DBLP:journals/tocl/Platzer15},
thus paving the way for a straightforward theorem prover implementation of \dLCHP.
Since \dLCHP conservatively generalizes \dL \citeDLCHP,
its uniform substitution calculus inherits the complete \cite{DBLP:journals/jacm/PlatzerT20} axiomatic treatment of differential equation invariants \cite{DBLP:journals/jar/Platzer17}.
\ifreport\else
All proofs are in \citeReport.
\fi{}

\section{Dynamic Logic of Communicating Hybrid Programs} \label{sec:dlCHP}

This section briefly recaps \dLCHP \citeDLCHP, 
the dynamic logic of communicating hybrid programs (CHPs).
It combines hybrid programs \cite{DBLP:journals/jar/Platzer08} with CSP-style communication and parallelism \cite{Hoare1978}.
By assumption-commitment (ac) reasoning \cite{Misra1981,AcHoare_Zwiers,AcSemantics_Zwiers}, \dLCHP allows compositional verification of parallelism in \dL.
For uniform substitution, function and predicate symbols, and program constants are added.

\subsection{Syntax} \label{sec:sytax}

The set of variables $\V = \RVar \cup \NVar \cup \TVar$ has real ($\RVar$), integer ($\NVar$), and trace ($\TVar$) variables.
For each $x \in \RVar$, the differential symbol $x'$ is in $\RVar$, too.
The designated variable $\globalTime \in \RVar$ represents the shared global time.
The set of channel names is~$\Chan$.
By convention $x, y \in \RVar$, $\intVar \in \NVar$, $\historyVar \in \TVar$, $\ch{} \in \Chan$, and $\arbitraryVar \in \V$.
Channel set $\cset \subseteq \Chan$ is (co)finite.
Vectorial expressions are denoted $\exprvec$.
Moreover, $\fsymb[some]$, $\fsymb[alt, some]$ are
$\anysort$-valued function symbols and
$\psymb, \psymb[alt], \psymb[opt]$ are predicate symbols,
where argument sorts are annotated by $\_ : \anysort_1, \ldots, \anysort_k$.
Finally, $\pconst, \pconst[alt]$ are program constants.

\begin{definition}[Terms] \label{def:syntax_terms}
	Terms consist of real $(\Rtrm)$, integer $(\Itrm)$, channel $(\Ctrm)$, and trace $(\Ttrm)$ terms, and are defined by the grammar below, where $\rp, \rp_1, \rp_2 \in \polynoms{\rationals}{\RVar} \subset \Rtrm$ are polynomials in~$\RVar$:
	\begin{align*}
		\Rtrm: &&\re_1, \re_2 & \cceq x \mid \fsymb[real](\cset, \exprvec) \mid \re_1 + \re_2 \mid \re_1 \cdot \re_2 \mid (\rp)' \mid \val{\te} \mid \stamp{\te} \\
		\Itrm: &&\ie_1, \ie_2 & \cceq \intVar \mid \fsymb[int](\cset, \exprvec) \mid \ie_1 + \ie_2 \mid \len{\te} \\
		\Ctrm: &&\ce_1, \ce_2 &  \cceq \fsymb[chan](\cset, \exprvec) \mid \chan{\te} \\
		\Ttrm: &&\te_1, \te_2 & \cceq \historyVar \mid \fsymb[trace](\cset, \exprvec) \mid \comItem{ch, \rp_1, \rp_2} \mid \te_1 \cdot \te_2 \mid \te \downarrow \cset \mid \at{\te}{\ie}
	\end{align*}
\end{definition}

Real terms are polynomials in $\RVar$ enriched with %real-valued 
function symbols~$\fsymb[real](\cset, \exprvec)$ (including constants $\const[rational] \in \rationals$)
only depending on communication along channels~$\cset$ and terms $\exprvec$, differential terms $(\rp)'$, and $\val{\te}$ and $\stamp{\te}$,
which access the value and the timestamp of the last communication in~$\te$, respectively. 
By convention, $\rp \in \polynoms{\rationals}{\RVar}$ denotes a pure polynomial in~$\RVar$ without $(\cdot)'$, $\val{\cdot}$, and $\stamp{\cdot}$ as they occur in programs. 
For simplicity, we do not define 
$\polynoms{\rationals}{\RVar} \subset \Rtrm$
as a fifth term sort but use the convention that function symbols $\fsymb[alt, real]$ can only be replaced with  $\polynoms{\rationals}{\RVar}$-terms.
Integer terms are variables $\intVar$, 
function symbols $\fsymb[int](\cset, \exprvec)$ (including constants $0$, $1$), addition, and length~$\len{\te}$ of trace term~$\te$.%
\footnote{Omitting multiplication results in decidable Presburger arithmetic \cite{Presburger1931}.}
The function symbol $\fsymb[chan](\cset, \exprvec)$ includes constants $\ch{} \in \Chan$, and $\chan{\te}$ is channel access.
Trace terms record the communication history of programs.
They encompass variables $\historyVar$, function symbols~$\fsymb[trace](\cset, \exprvec)$ (including the empty trace~$\epsilon$), communication items $\comItem{\ch{}, \rp_1, \rp_2}$ with value $\rp_1$ and timestamp $\rp_2$, projection $\te \downarrow \cset$ onto channels $\cset$, and access $\at{\te}{\ie}$ of the $\ie$-th item in $\te$.
Where useful, $\fsymb[builtin](\exprvec)$ denotes built-in function symbols of fixed interpretation, \eg~$\usarg + \usarg$.

\dLCHP's context-sensitive program and formula syntax presumes notions of free and bound variables (\rref{sec:static_semantics}) defined on the context-free syntax:

\begin{definition}[Programs] \label{def:syntax_chps}
	\emph{Communicating hybrid programs} are defined by the following grammar,
	where $\rp \in \polynoms{\rationals}{\RVar}$ is a polynomial in $\RVar$ and $\chi \in \FolRA$ is a formula of first-order real-arithmetic.
	In $\alpha \parOp \beta$, the subprograms must not share state but can share time and history, \iest 
	$\SBV(\alpha) \cap \SBV(\beta) \subseteq \{\globalTime,\globalTime'\} \cup \TVar$.%
	\footnote{
		Previous work \citeDLCHP disallows reading of variables bound in parallel as their change is not observable.
		This restriction is conceptually desirable but not soundness-critical.
		Here we drop it for simplicity, but it could be maintained by \RuleName{US} as well.
	}
	\begin{align*}
		\alpha, \beta \cceq \;
		& \rpconst{}{} \mid x \ceq \rp \mid x \ceq * \mid \test{} \mid \evolution{}{} \mid \alpha \seq \beta \mid \alpha \cup \beta \mid \repetition{\alpha} \mid \\
		& \send{}{}{} \mid \receive{}{}{} \mid \alpha \parOp \beta
	\end{align*}
\end{definition}

The program constant $\rpconst{}{}$ restricts the written channels to $\cset \subseteq \Chan$ and the bound variables to $\varvec \subseteq \RVar \cup \TVar$,
where $\cset$ and $\varvec$ are (co)finite.
Instead of $\rpconst[free={\RVar\cup\TVar}]{\cset}{\varvec}$, write $\pconst$ if $\cset$ and $\varvec$ can be arbitrary.
Assignment $x \ceq \rp$ updates~$x$ to~$\rp$,
nondeterministic assignment $x \ceq *$ assigns an arbitrary real value to~$x$,
and the test $\test{}$ does nothing if $\chi$ holds and aborts the computation otherwise.
The continuous evolution $\evolution{}{}$ follows the ODE $x' = \rp$ for any duration as long as formula $\chi$ is not violated.
The global time $\globalTime$ evolves with every continuous evolution according to ODE $\globalTime' = 1$.
Sequential composition $\alpha \seq \beta$ executes~$\beta$ after~$\alpha$, choice $\alpha \cup \beta$ executes $\alpha$ or $\beta$ nondeterministically,~$\repetition{\alpha}$ repeats $\alpha$ zero or more times,
$\send{}{}{}$ sends $\rp$ along channel $\ch{}$,
and $\receive{}{}{}$ receives a value into variable $x$ along channel $\ch{}$.
The trace variable $\historyVar$ records communication.
Finally, $\alpha \parOp \beta$ executes $\alpha$ and $\beta$ in parallel synchronized in global time $\globalTime$. 

\newcommand{\controllerName}{\progtt{ct}}
\newcommand{\vehicleName}{\progtt{ve}}
\newcommand{\controller}{\repetition{\controllerName}}
\newcommand{\vehicle}{\repetition{\vehicleName}}

\newcommand{\velo}[1]{v_\text{#1}}
\newcommand{\vtar}[1]{v^\text{tr}_\text{#1}}
\newcommand{\acceleration}{a_\vehicleName}

\newcommand{\maxvelo}{V}
\newcommand{\maxdiff}{\delta_V}

\newcommand{\vrange}[1]{0 \le #1 \le \maxvelo}
\newcommand{\vsafe}[2]{\vrange{#1} \wedge \abs{#1 - #2} \le \maxdiff}

\begin{example}
	\label{ex:prog}
	The program $\controller \parOp \vehicle$ models a simplified cruise control \cite{DBLP:conf/ifm/MullerMRSP16}.
	The vehicle~$\vehicleName$ repeatedly receives a target velocity $\vtar{\vehicleName}$ from the controller $\controllerName$ along channel $\ch{tar}$.
	The target $\vtar{\controllerName}$ sent by $\controllerName$ is in range $[0,\maxvelo]$. 
	Hence, $\vehicleName$'s velocity~$\velo{\vehicleName}$ stays in range $[0, \maxvelo]$ within the $\epsilon > 0$ time units till the next communication if $\velo{\vehicleName} \in [0,\maxvelo]$ held initially.
	The evolution $\evolution{t' = 1}{non}$ allows passage of time in $\controllerName$.
	\begin{align*}
		\controllerName & \equiv 
			\vtar{\controllerName} \ceq * \seq 
			\test{(\vrange{\vtar{\controllerName}})} \seq 
			\send{\ch{tar}}{}{\vtar{\controllerName}} \seq \evolution{t' = 1}{non} \\
		\vehicleName & \equiv
			\receive{\ch{tar}}{}{\vtar{\vehicleName}} \seq 
			\acceleration \ceq \frac{\vtar{\vehicleName} - \velo{\vehicleName}}{\epsilon} \seq 
			t_0 \ceq \globalTime \seq 
			\evolution{\velo{\vehicleName}' = \acceleration}{\globalTime - t_0 \le \epsilon}
	\end{align*}
\end{example}

\begin{definition}[Formulas] \label{def:syntax_formulas}
	Formulas are defined by the grammar below 
	for relations $\sim$,
	terms $\expr_1, \expr_2 \in \Trm$ of equal sort,
	and $\arbitraryVar \in \V$.
	Moreover, the ac-formulas are unaffected by state change in $\alpha$,
	\iest $(\SFV(\A) \cup \SFV(\C)) \cap \SBV(\alpha) \subseteq \TVar$.
	\begin{align*}
		\varphi, \psi, \A, \C \cceq \; 
			& \expr_1 \sim \expr_2 \mid \psymb(\cset, \exprvec) \mid 
			\neg \varphi \mid \varphi \wedge \psi 
			\mid \fa{\arbitraryVar} \varphi 
			\mid [ \alpha ] \psi \mid [\alpha] \ac \psi
	\end{align*}
\end{definition}

The formulas combine first-order dynamic logic with ac-reasoning.
Predicate symbols $\psymb(\cset, \exprvec)$ depend on channels $\cset$ and terms $\exprvec$.
The ac-box $[ \alpha ] \ac \psi$ expresses that $\C$ holds after each communication event and~$\psi$ in the final state, for all runs of $\alpha$ whose incoming communication satisfies~$\A$.
Other connectives $\vee$, $\rightarrow$, $\leftrightarrow$ and quantifiers $\ex{\arbitraryVar} \varphi \equiv \neg \fa{\arbitraryVar} \neg \varphi$ can be derived.
The relations $\sim$ include~$=$ for all term sorts, $\ge$ on real and integer terms, and prefixing $\preceq$ on trace terms.

By convention, the predicate symbol $\psymb[alt,real]$ can only be replaced with formulas of first-order real arithmetic.
It serves as placeholder for tests $\chi$ in CHPs.

\newcommand{\ccpre}{{\varphi}}
\newcommand{\ccsafe}{{\psi_\text{safe}}}

\begin{example}
	\label{ex:safety}
	The cruise control from \rref{ex:prog} is safe if its velocity stays in range $[0, \maxvelo]$.
	This can be expressed with the formula $\ccpre \rightarrow [ \controller \parOp \vehicle ] \ccsafe$,
	where $\ccsafe \equiv \vrange{\velo{\vehicleName}}$ and $\ccpre \equiv \ccsafe \wedge \epsilon > 0 \wedge \maxvelo > 0$.
\end{example}

\subsection{Semantics} \label{sec:semantics}

A \emph{trace} $\trace = (\trace_1, ..., \trace_k)$ is a finite chronological sequence of communication events $\trace_i = \comItem{\ch{}_i, \semConst_i, \duration_i}$,
where $\ch{}_i \in \Chan$, and $\semConst_i \in \reals$ is the communicated value, and $\duration_i \in \reals$ is a timestamp such that $\duration_i < \duration_j$ for $1 \le i < j \le k$.
A \emph{recorded trace} $\trace = (\trace_1, ..., \trace_k)$ additionally carries a trace variable $\historyVar_i \in \TVar$ with each event, \iest $\trace_i = \comItem{\historyVar_i, \ch{}_i, \semConst_i, \duration_i}$.
For variable $\arbitraryVar \in \V_\anysort$ and $\anysort \in \{\reals, \naturals, \traces\}$, let $\type(\arbitraryVar) = \anysort$.
A \emph{state} $\pstate{v}$ maps each $\arbitraryVar \in \V$ to a value $\pstate{v}(\arbitraryVar) \in \type(\arbitraryVar)$.
The sets of traces, recorded traces, and states are denoted $\traces$, $\recTraces$, and $\states$, respectively.

For $\semConst \in \type(\arbitraryVar)$, the state $\pstate{v} \subs{\arbitraryVar}{\semConst}$ is the modification of $\pstate{v}$ at $\arbitraryVar$ to $\semConst$.
For $\trace \in \recTraces$, the trace $\trace(\historyVar) \in \traces$ is obtained from the subsequence of $\trace$ carrying $\historyVar \in \TVar$ by removing the carried variable.
\emph{State-trace concatenation} $\stconcat{\pstate{v}}{\trace} \in \states$ for $\trace \in \recTraces$,
appends $\trace(\historyVar)$ to $\pstate{v}$ at~$\historyVar$ for all $\historyVar \in \TVar$.
The \emph{projection} $\trace \downarrow \cset$ of (recorded) trace~$\trace$ is the subsequence of all communication events in $\trace$ whose channel is in $\cset \subseteq \Chan$.
The \emph{state projection} $\pstate{v} \downarrow \cset \in \states$ modifies $\pstate{v}$ at $\historyVar$ to $\pstate{v}(\historyVar) \downarrow \cset$ for all $\historyVar \in \TVar$.

An \emph{interpretation} $\inter$ assigns a function 
$\interOf{\fsymb[some] : \anysort_1, \ldots, \anysort_k} : \bigtimes_{i=1}^k \anysort_i \rightarrow \anysort$
to each function symbol $\fsymb[some]$ 
that is smooth in all real-valued arguments if $\anysort = \reals$,
and a relation $\interOf{\psymb : \anysort_1, \ldots, \anysort_k} \subseteq \bigtimes_{i=1}^k \anysort_i$ to each $k$-ary predicate symbol $\psymb$.

\begin{definition}[Term Semantics]
	The \emph{valuation} $\sem{\expr}{\lstate{v}} \in \reals \cup \naturals \cup \Chan \cup \traces$ of term~$\expr$ in interpretation $\inter$ and state $\pstate{v}$ is defined as follows: 
	\begin{align*}
		\sem{\arbitraryVar}{\lstate{v}} 
			& = \pstate{v}(\arbitraryVar) \\
		\sem{\fsymb(\cset, \expr_1, ..., \expr_k)} {\lstate{v}} 
			& = \interOf{\fsymb}(\sem{\expr_1}{\lstate[opt]{v}}, ..., \sem{\expr_k}{\lstate[opt]{v}}) 
			&&\sidecondition[color=black]{where $\pstate[alt]{v} = \pstate{v} \downarrow \cset$} \\
		\sem{\fsymb[builtin](\expr_1, \ldots, \expr_k)}{\lstate{v}} 
			& = \fsymb[builtin](\sem{\expr_1}{\lstate{v}}, \ldots, \sem{\expr_k}{\lstate{v}}) 
			&&\sidecondition[color=black]{for builtin $\fsymb[builtin] \in \{\usarg+\usarg, \usarg\downarrow\cset, \ldots\}$} \\
		\sem{(\rp)'}{\lstate{v}} 
			& = \sum_{\mathclap{x \in \RVar}} \pstate{v}(x') \frac{\partial \sem{\rp}{\lstate{v}}}{\partial x}
	\end{align*}
\end{definition}

The projection $\pstate[alt]{v} = \pstate{v} \downarrow \cset$ ensures that $\fsymb(\cset, \exprvec)$ only depends on $\cset$,
\iest the communication in $\pstate{v}$ along channels $\cset^\complement$ does not matter.
The differentials~$(\rp)'$ have a semantics describing the local rate of change of $\rp$ \cite{DBLP:journals/jar/Platzer17}.

The denotational semantics of CHPs \citeDLCHP combines \dL's Kripke semantics \cite{DBLP:journals/jar/Platzer17} with a 
linear history semantics \cite{AcSemantics_Zwiers} and a global notion of time. 
Denotations are subsets of $\pDomain = \states \times \recTraces \times \botop{\states}$ with $\botop{\states} = \states \cup \{\bot\}$.
Final state $\bot$ marks an unfinished computation,
\iest it still can be continued or was aborted due to a failing test.
If ($\pstate[pre]{w} = \bot$ and $\trace[pre] \preceq \trace$),
where $\preceq$ is the prefix relation on traces,
or $(\trace[pre], \pstate[pre]{w}) = (\trace, \pstate{w})$,
then $\observable[pre]$ is a prefix of $\observable$ written
$\observable[pre] \preceq \observable$.
Since (even empty) communication of unfinished computations is still observable,
denotations~$\denotation \subseteq \pDomain$ of CHPs are prefix-closed and total,
\iest $\computation \in \denotation$ and $\observable[pre] \preceq \observable$ implies $\computation[pre] \in \denotation$, and $\pLeast \subseteq \denotation$ with $\pLeast = \states \times \{\epsilon\} \times \{\bot\}$.
Moreover, all $\computation \in \denotation$ are chronological,
\iest $\statetime{\pstate{v}} \le \statetime{\pstate{w}}$ and when $\trace = (\trace_1, \ldots, \trace_k) \neq \epsilon$ and let $\statetime{\trace_i} = \statetime{(\comItem{\historyVar_i, \ch{}_i, \semConst_i, \duration_i})} = \duration_i$, then $\statetime{\pstate{v}} \le \statetime{\trace_1}$ and if $\pstate{w} \neq \bot$, then $\statetime{\trace_k} \le \statetime{\pstate{w}}$.
Note that $\trace$ is chronological as all traces are.

The interpretation $\interOf{\rpconst{}{}} \subseteq \pDomain$ of a program constant $\rpconst{}{}$ is a prefix-closed and total set of chronological computations that 
\begin{enumerate*}[label=(\roman*)]
	\item only communicate along (write) channels $\cset$ and\label{itm:inter_1}
	\item only bind variables $\varvec$\label{itm:inter_2}.
\end{enumerate*}
More precisely, for all $\computation \in \interOf{\rpconst{}{\varvec}}$, we have 
\ref*{itm:inter_1} $\trace \downarrow \cset^\complement = \epsilon$, and 
\ref*{itm:inter_2} $\pstate{v} = \pstate{w}$ on $\TVar$ and $\stconcat{\pstate{w}}{\trace} = \pstate{v}$ on $\varvec^\complement$.
For $\denotation, \denotation[alt] \subseteq \pDomain$,
we define $\botop{\denotation} = \{ (\pstate{v}, \trace, \bot) \mid \computation \in \denotation \}$,
and $\computation \in \denotation \continuation \denotation[alt]$ if $(\pstate{v}, \trace_1, \pstate{u}) \in \denotation$ and $(\pstate{u}, \trace_2, \pstate{w}) \in \denotation[alt]$ exist with $\trace = \trace_1 \cdot \trace_2$.
For states $\pstate{w}_\alpha, \pstate{w}_\beta$, the merged state $\pstate{w}_\alpha \merge \pstate{w}_\beta$ is $\bot$ if one of the substates $\pstate{w}_\alpha$ or $\pstate{w}_\beta$ is~$\bot$.
Otherwise, $\pstate{w}_\alpha \merge \pstate{w}_\beta = \pstate{w}_\alpha$ on $\SBV(\alpha)$ and $\pstate{w}_\alpha \merge \pstate{w}_\beta = \pstate{w}_\beta$ on $\SBV(\alpha)^\complement$
(or, equivalently by syntactic well-formedness, on $\SBV(\beta)^\complement$ and $\SBV(\beta)$, respectively).
If $\cset$ is the set of all channel names occurring in $\alpha$,
we write $\trace \downarrow \alpha$ for $\trace \downarrow \cset$.

\begingroup
\allowdisplaybreaks
\begin{definition}[Program semantics]\label{def:programsemantics}
	Given an interpretation $\inter$, the \emph{semantics} $\sem{\alpha}{\inter} \subseteq \pDomain$ of a CHP $\alpha$ is defined as follows,
	where $\pLeast = \states \times \{\epsilon\} \times \{\bot\}$ 
	and~$\vDash$ denotes the satisfaction relation (\rref{def:formulaSemantics}): 
	\begingroup
	\begin{align*}%
		%
		% atomic programs from dL
		%
		& \sem{\rpconst{}{}}{\inter}
		= \interOf{\rpconst{}{}} \\
		&\sem{x \ceq \rp}{\inter}
		= \pLeast \cup \{ (\pstate{v}, \epsilon, \pstate{w}) \mid \pstate{w} = \pstate{v} \subs{x}{\semConst} \text{ where } \semConst = \sem{\rp}{\lstate{v}} \} \\
		&\sem{x \ceq *}{\inter}
		= \pLeast \cup \{ (\pstate{v}, \epsilon, \pstate{w}) \mid \pstate{w} = \pstate{v} \subs{x}{\semConst} \text{ where } \semConst \in \reals \} \\
		&\sem{\test{}}{\inter}
		= \pLeast \cup \{ (\pstate{v}, \epsilon, \pstate{v}) \mid \lstate{v} \vDash \chi \} \\
		&\sem{\evolution{}{}}{\inter} 
		= \pLeast \cup \big\{ (\pstate{v}, \epsilon, \odeSolution(\duration)) \mid
			\pstate{v} = \odeSolution(0) \text{ on } \{\globalTime', x'\}^\complement \text{, and } \odeSolution(\zeta) = \odeSolution(0) \\
			&\quad \text{on } \{ x, x', \globalTime, \globalTime'\}^\complement\text{, and }
			\lsolution(\zeta) \vDash \globalTime' = 1 \wedge x' = \rp \wedge \chi
			\text{ for all } \zeta \in [0, \duration] \text{ and} \\
			&\quad 
			\text{a solution } \odeSolution : [0, \duration] \rightarrow \states \text{ with } \odeSolution(\zeta)(\arbitraryVar') = \solutionDerivative{\odeSolution}{\arbitraryVar}(\zeta) \text{ for } \arbitraryVar \in \{x,\globalTime\} \big\} \\
		%
		% communication primitives
		%
		&\sem{\send{}{}{}}{\inter}
		= \{ \computation \mid \observable \preceq \semCom{\semConst}{\pstate{v}} \text{ where } \semConst = \sem{\rp}{\lstate{v}} \} \\
		&\sem{\receive{}{}{}}{\inter} 
		= \{ \computation \mid \observable \preceq \semCom{\semConst}{\pstate{v} \subs{x}{\semConst}} \text{ where } \semConst \in \reals \} \\
		%
		% compound programs
		%
		&\sem{\alpha \cup \beta}{\inter}
		= \sem{\alpha}{\inter} \cup \sem{\beta}{\inter} \\
		&\sem{\alpha \seq \beta}{\inter} 
		= \sem{\alpha}{\inter} \closedComposition \sem{\beta}{\inter} 
		\equalsdef \botop{(\sem{\alpha}{\inter})} \cup (\sem{\alpha}{\inter} \continuation \sem{\beta}{\inter})\\
		&\sem{\repetition{\alpha}}{\inter} 
		= \bigcup_{n \in \naturals} (\sem{\alpha}{\inter})^n 
		= \bigcup_{n \in \naturals} \sem{\alpha^n}{\inter} 
		\sidecondition[color=black]{\quad where $\alpha^0 \equiv \test{\true}$ and $\alpha^{n+1} = \alpha \seq \alpha^n$} \\
		&\sem{\alpha_1 \parOp \alpha_2}{\inter}
		%= \sem{\alpha_1} \inter \parOp \sem{\alpha_2} \inter 
		= \Bigg\{ (\pstate{v}, \trace, \pstate{w}_{\alpha_1} \merge \pstate{w}_{\alpha_2})
		\;\bigg\vert\; 
		\begin{aligned}
			&\computation[proj={\alpha_j}] \in \sem{\alpha_j}{\inter} \text{ for } j = 1, 2 \text{, and } \\
			&\pstate{w}_{\alpha_1} = \pstate{w}_{\alpha_2} \text{ on } \{ \globalTime, \globalTime' \} \text{, and }
			\trace = \trace \downarrow (\alpha_1{\parOp}\alpha_2)
		\end{aligned} 
		\Bigg\}
	\end{align*}
	\endgroup
\end{definition}
\endgroup

The semantics is indeed constructed prefix-closed, total, and chronological.
Communication $\trace$ of $\alpha_1 \parOp \alpha_2$ is implicitly characterized via its subsequences for the subprograms.
By $\trace = \trace \downarrow (\alpha_1 \parOp \alpha_2)$, there is no non-causal communication.
Joint communication and the whole computation are synchronized in global time by the projections and by $\pstate{w}_{\alpha_1} = \pstate{w}_{\alpha_2}$ on $\{ \globalTime, \globalTime' \}$, respectively.
Likewise, by projection, communication is synchronously recorded by trace variables.

\begin{definition}[Formula semantics]\label{def:formulaSemantics}
	The \emph{satisfaction} $\lstate{v} \vDash \phi$ of a \dLCHP formula~$\phi$ in interpretation $\inter$ and state $\pstate{v}$ is inductively defined as follows:
	\begin{enumerate}
		\item $\lstate{v} \vDash \expr_1 {\sim} \expr_2$ if $\sem{\expr_1}{\lstate{v}} \sim \sem{\expr_2}{\lstate{v}}$ \sidecondition[color=black]{\quad where $\sim$ is any relation symbol}
		\item $\lstate{v} \vDash \psymb(\cset, \expr_1, \ldots, \expr_k)$ if $(\sem{\expr_1}{\lstate[opt]{v}}, \ldots, \sem{\expr_k}{\lstate[opt]{v}}) \in \interOf{\psymb}$
		\sidecondition[color=black]{\quad where $\pstate[alt]{v} = \pstate{v} \downarrow \cset$}
		\item $\lstate{v} \vDash \varphi \wedge \psi$ if $\lstate{v} \vDash \varphi$ and $\lstate{v} \vDash \psi$
		\item $\lstate{v} \vDash \neg \varphi$ if $\lstate{v} \nvDash \varphi$, \iest it is not the case that $\lstate{v} \vDash \varphi$
		\item $\lstate{v} \vDash \fa{\arbitraryVar} \varphi$ if $\lstate{v} \subs{\arbitraryVar}{d} \vDash \varphi$ for all $d \in \type(\arbitraryVar)$
		\item \label{itm:dynBoxSem}
		$\lstate{v} \vDash [ \alpha ] \psi$ if $\stconcat{\lstate{w}}{\trace}\vDash \psi$ for all $\computation \in \sem{\alpha}{\inter}$ with $\pstate{w} \neq \bot$ 
		\item \label{itm:acBoxSem}
		$\lstate{v} \vDash [ \alpha ] \ac \psi$ if for all $\computation \in \sem{\alpha}{\inter}$ the following conditions hold:
		\begin{align}
			\vspace{-.7em}
			& \assCommit{\lstate{v}}{\trace} \vDash \A \text{ implies } \stconcat{\lstate{v}}{\trace} \vDash \C \tag{commit} \label{eq:commit} \\
			&\big( \assPost{\lstate{v}}{\trace} \vDash \A \text{ and } \pstate{w} \neq \bot \big) \text{ implies } \stconcat{\lstate{w}}{\trace} \vDash \psi \tag{post} \label{eq:post}
		\end{align}%
		Where $U \vDash \varphi$ for a set of interpretation-state pairs $U$ and any formula $\varphi$ if $\lstate{v} \vDash \varphi$ for all $\lstate{v} \in U$. 
		In particular, $\emptyset \vDash \varphi$.
	\end{enumerate}
\end{definition}

In \rref{itm:dynBoxSem} and \ref{itm:acBoxSem}, reachable worlds are built from states $\pstate{v}$ and $\pstate{w}$, and communication $\trace$, as change of state \emph{and} communication are observable.
The strict prefix $\prec$ for the assumption in case \acCommit in \rref{itm:dynBoxSem} excludes (when $\A \equiv \C$) the circularity that commitment $\C$ can be shown in states where it is assumed.

\subsection{Static Semantics} \label{sec:static_semantics}

In the uniform substitution process, checks of free and bound variables, as well as accessed and written channels, separate sound from unsound axiom instantiations.
As parallelism requires fine-grained control over channels,
the static semantics for \dL \cite{DBLP:journals/jar/Platzer17} is lifted to a communication-aware static semantics for \dLCHP.
It uses accessed channels to characterize the subsequence of a communication trace influencing truth of a formula even more precisely than free variables.

To precisely grasp free and bound variables, and accessed and written channels, 
\rref{def:staticSemantics} gives a semantic characterization.
In this section, formulas are considered truth-valued,
\iest $\sem{\phi}{\lstate{v}} = \mathbf{tt}$ if $\lstate{v} \vDash \phi$ and $\sem{\phi}{\lstate{v}} = \mathbf{ff}$ if $\lstate{v} \nvDash \phi$.

\begin{definition}[Static semantics]
	\label{def:staticSemantics}
	For term or formula $\expr$, and program $\alpha$,
	free variables $\SFV(\expr)$ and $\SFV(\alpha)$, bound variables $\SBV(\alpha)$, accessed channels $\SCN(\expr)$, and written channels $\SCN(\alpha)$ form the static semantics.
	\begin{align*}
		\SFV(\expr) 
			& = \{ \arbitraryVar \in \V \mid \eexists \inter, \pstate{v}, \pstate[alt]{v} \text{ such that } \pstate{v} = \pstate[alt]{v} \text{ on } \{ \arbitraryVar \}^\complement \text{ and } \sem{\expr}{\lstate{v}} \neq \sem{\expr}{\lstate[opt]{v}} \} \\
		\SCN(\expr) 
			& = \{ \ch{} \in \Chan \mid \eexists \inter, \pstate{v}, \pstate[alt]{v} \text{ such that } \pstate{v} \downarrow \{ \ch{} \}^\complement = \pstate[alt]{v} \downarrow \{ \ch{} \}^\complement \text{ and } \sem{\expr}{\lstate{v}} \neq \sem{\expr}{\lstate[opt]{v}} \} \\
		\SFV(\alpha) 
			& = \{ \arbitraryVar \in \V \mid \eexists \inter, \pstate{v}, \pstate[alt]{v}, \trace, \pstate{w} \text{ such that } \pstate{v} = \pstate[alt]{v} \text{ on } \{ \arbitraryVar \}^\complement \text{ and } \computation \in \sem{\alpha}{\inter} \text{,} \\[-3pt]
			& \quad\,\,\,\,\,\, \text{and there is no } \computation[alt] \in \sem{\alpha}{\inter} \text{ such that } \trace[alt] = \trace \text{ and } \pstate{w} = \pstate[alt]{w} \text{ on } \{ \arbitraryVar \}^\complement \} \\
		\SBV(\alpha)
			& = \{ \arbitraryVar \in \V \mid \eexists \inter, \computation \in \sem{\alpha}{\inter} \text{ such that } \pstate{w} \neq \bot \text{ and } (\stconcat{\pstate{w}}{\trace})(\arbitraryVar) \neq \pstate{v}(\arbitraryVar) \} \\
		\SCN(\alpha)
			& = \{ \ch{} \in \Chan \mid \eexists \inter, \computation \in \sem{\alpha}{\inter} \text{ such that } \trace \downarrow \{\ch{}\} \neq \epsilon \}
	\end{align*}
\end{definition}

The already subtle static semantics of hybrid systems \cite{DBLP:journals/jar/Platzer17}
becomes even more subtle with communication and parallelism.
For example, CHPs (silently) synchronize with the global time $\globalTime$,
which is free and bound in ODEs,
and the differential~$\globalTime'$ is bound,
\iest $\globalTime \in \SFV(\evolution{}{})$ and $\globalTime, \globalTime' \in \SBV(\evolution{}{})$ if the evolution has a run of non-zero duration,
regardless of whether $\globalTime$ occurs in~$x$.
Since reachable worlds of CHPs 
consist of communication \emph{and} state,
bound variables $\SBV(\alpha)$ of program $\alpha$ compare $\pstate{v}$ with the state-trace concatenation $\pstate{w} \cdot \trace$ instead of missing $\trace$.
Consequently, $\historyVar \in \SBV(\send{}{}{}) \subseteq \SFV(\send{}{}{})$,
which also reflects that the initial communication never gets lost.
\reportonly{All proofs for this section and computable overapproximations of the static semantics  are in \rref{app:staticSemantics}.}

\begin{lemma}[Bound effect property] \label{lem:boundEffect}
	The sets $\SBV(\alpha)$ and $\SCN(\alpha)$ are the smallest sets with the \emph{bound effect property for program} $\alpha$. 
	That is, $\pstate{v} = \pstate{w}$ on~$\TVar$ and $\pstate{v} = \stconcat{\pstate{w}}{\trace}$ on $\SBV(\alpha)^\complement$ if $\pstate{w} \neq \bot$, and $\trace \downarrow \SCN(\alpha)^\complement = \epsilon$ for all $\computation \in \sem{\alpha}{\inter}$.
\end{lemma}

By the following \emph{communication-aware} coincidence property, 
terms and formulas only depend on their free variables,
which for trace variables can be further refined to the subtraces whose channels are accessed.
This subtrace-level precision is crucial in the soundness proof of the parallel injection axiom as it allows to drop $\beta$ from $[ \alpha \parOp \beta ] \psi$ only if $\beta$ does not write channels of $\psi$ that are not also written by $\alpha$.
The signature $\sigof{\cdot}$ of an expression denotes all occurring symbols.

\begin{lemma}[Coincidence for terms and formulas] 
	\label{lem:termCoincidence}
	The sets $\SFV(\expr)$ and $\SCN(\expr)$ are the smallest sets with the \emph{communication-aware 
	coincidence property for term or formula} $\expr$.
	That is, if $\pstate{v} \downarrow \SCN(\expr) = \pstate[alt]{v} \downarrow \SCN(\expr)$ on $\SFV(\expr)$ and $\inter = \inter[alt]$ on~$\sigof{\expr}$,
	then $\sem{\expr}{\lstate{v}} = \sem{\expr}{\lstate[alt]{v}}$.
	In particular, for formula $\phi$: $\lstate{v} \vDash \phi$ iff $\lstate[alt]{v} \vDash \phi$.
\end{lemma}

Programs communicate but do \emph{not} depend on the recorded history, thus the coincidence property for programs is not communication-aware.
However, programs can produce the same communication starting from coinciding states.

\begin{lemma}[Coincidence for programs] \label{lem:programCoincidence}
	The set $\SFV(\alpha)$ is the smallest set with the 
	\emph{coincidence property for program} $\alpha$.
	That is, if $\pstate{v} = \pstate[alt]{v}$ on $\varset \supseteq \SFV(\alpha)$,
	and $\inter = \inter[alt]$ on $\sigof{\alpha}$, 
	and $\computation \in \sem{\alpha}{\inter}$,
	then $\computation[alt] \in \sem{\alpha}{\inter[alt]}$ exists such that $\pstate{w} = \pstate[alt]{w}$ on $\varset$,
	and $\trace = \trace[alt]$,
	and ($\pstate{w} = \bot$ iff $\pstate[alt]{w} = \bot$).
\end{lemma}

\section{Uniform Substitution for \dLCHP} \label{sec:unisubs}

\newcommand{\smapsto}{\kern-2pt\mapsto\kern-2pt}

In \dLCHP, a uniform substitution \cite{DBLP:journals/jar/Platzer17} $\usubs$ maps function and predicate symbols to terms (of equal sort) and formulas, respectively, while substituting the arguments of the symbol for their placeholders in the replacement, 
and program constants are mapped to CHPs.
For example, $\usubs = \{ \fsymb(\usarg) \smapsto \usarg + 1, \pconst \smapsto \receive{}{}{v} \seq \evolution{x' = v}{non} \}$
replaces all occurrences of function symbol $\fsymb$ with $\usarg + 1$ 
while the reserved $0$-ary function symbol $\usarg$ marks the positions for the parameter of $\fsymb$ in the replacement.
Moreover, $\usubs$ replaces the program constant $\pconst$ with the program $\receive{}{}{v} \seq \evolution{x' = v}{non}$.

The key to sound uniform substitution is that new free variables must not be introduced into a context where they are bound \cite{Church1956}. 
In the presence of communication,
likewise, \emph{new channel access must not be introduced into contexts} where the channel is written \ref{itm:channels}.
For parallelism, substitution \emph{must not reveal internals} of the parallel context to the local abstraction of a subprogram \ref{itm:abstraction},
and must not violate state disjointness.
The one-pass approach \cite{DBLP:conf/cade/Platzer19} used for \dLCHP postpones these checks \emph{and} simply applies the substitution recursively while collecting written variables and channels as taboo set,
thus operates linearly in the input.
Clashes between the taboo, and new free variables and channel access are only checked locally at the replacement site.
Likewise, clashes between the permitted channels and variables of a program constant, and its replacement program are checked locally.
\reportonly{All proofs for this section are in \rref{app:soundness}.}

The substitution operator $\usInOutOp{}{}{\alpha}$ for program $\alpha$ takes an input taboo $\jointTaboo \subseteq \V \cup \Chan$ and a parallel context $\parallelCtx \subseteq \V$,
and returns, if defined, the substitution result and a set of output taboos $\jointOut \subseteq \V \cup \Chan$.
For terms and formulas, the substitution operator $\usTaboo{}$ only takes a taboo $\jointTaboo \subseteq \V \cup \Chan$ as input.
The substitution process clashes,
\iest prevents unsound instantiation,
if it were to introduce a free variable or accessed channel into a context where it is bound \ref{itm:channels} \emph{or} if it were to write variables and channels violating abstraction \ref{itm:abstraction}.
Moreover, substitution preserves well-formedness of programs and formulas,
\iest substitution clashes if replacements were to violate well-formedness. 

\newcommand{\sidetext}[1]{&& \text{#1}}
\newcommand{\HRule}[1]{\par
  \vspace*{\dimexpr-\parskip-\baselineskip+#1}
  \noindent\rule{\linewidth}{.4pt}\par
  \vspace*{\dimexpr-\parskip-1.5\baselineskip+#1}}

\begin{figure}[ht!]
	\vspace*{-2em}
	\begin{minipage}{\textwidth}
		\setlength{\jot}{-1.5pt}
		\begin{align*}
			%
			% Terms
			%
			\usTabooOp{}{\arbitraryVar} 
				& \equiv \arbitraryVar
				\sidetext{for $\arbitraryVar \in \V$} \\
			\usTabooOp{}{\fsymb(\cset, \expr)} 
				& %\equiv (\usTabooOp{}{f})(\usTabooOp{}{\expr}) 
					\equiv \usAuxOp{\expr}{\usubs \fsymb (\usarg)}
					\sidetext{if $(\SFV(\usubs \fsymb(\usarg)) \cup \SCN(\usubs \fsymb(\usarg))) \cap \jointTaboo = \emptyset$} \\
			\usTabooOp{}{\fsymb[builtin](\expr_1, \ldots, \expr_k)}
				& \equiv \fsymb[builtin](\usTabooOp{}{\expr_1}, \ldots, \usTabooOp{}{\expr_k}) 
				\sidetext{for built-in $\fsymb[builtin] \in \{ \usarg+\usarg, \usarg\downarrow\cset, \ldots \}$} \\
			\usTabooOp{}{(\rp)'} 
				& \equiv (\usTabooOp{\V \cup \Chan}{\rp})'
		\end{align*}%
		\HRule{-.3em}%
		\begin{align*}
			%
			% FOL formulas
			%
			\usTabooOp{}{\expr_1 \sim \expr_2} 
				& \equiv \usTabooOp{}{\expr_1} \sim \usTabooOp{}{\expr_2} && \\
			\usTabooOp{}{\psymb(\cset, \expr)}
				& \equiv \usAuxOp{\expr}{\usubs \psymb(\usarg)}
				\sidetext{if $(\SFV(\usubs \psymb(\usarg)) \cup \SCN(\usubs \psymb(\usarg))) \cap \jointTaboo = \emptyset$} \\
			\usTabooOp{}{\neg \varphi}
				& \equiv \neg \usTabooOp{}{\varphi} \\
			\usTabooOp{}{\varphi \wedge \psi}
				& \equiv \usTabooOp{}{\varphi} \wedge \usTabooOp{}{\psi} \\
			\usTabooOp{}{\fa{\arbitraryVar} \varphi}
				& \equiv \fa{\arbitraryVar} \usTabooOp{\jointTaboo\cup\{\arbitraryVar\}}{\varphi} \\
			\usTabooOp{}{[ \alpha ] \psi} 
				& \equiv [ \usInOutOp{\jointTaboo, \emptyset}{}{\alpha} ] \usTabooOp{\jointOut}{\psi} \\
			\usTabooOp{}{[ \alpha ] \ac \psi} 
				& \equiv [ \usInOutOp{\jointTaboo, \emptyset}{}{\alpha} ] \acpair{\usTabooOp{\jointOut}{\A}, \usTabooOp{\jointOut}{\C}} \usTabooOp{\jointOut}{\psi}
		\end{align*}
		\HRule{-.3em}
		\begin{align*}
			\usInOutOp{}{\jointTaboo\cup\SBV(\usubs \pconst)\cup\SCN(\usubs \pconst)}{\rpconst{}{}} 
					& \equiv \usubs \pconst 
					\sidecondition[tag]{if $\SBV(\usubs \pconst) \subseteq \varvec$ and $\SCN(\usubs a) = \cset$} 
					\\[.1em]
			%
			% atomic programs
			%
			\usInOutOp{}{\jointTaboo\cup\{x\}}{x \ceq \rp} 
				& \equiv x \ceq \usTabooOp{\jointTaboo \cup \parallelCtx}{\rp} \\[.1em]
			\usInOutOp{}{\jointTaboo\cup\{x\}}{x \ceq *} 
				& \equiv x \ceq * 
				\\[.1em]
			\usInOutOp{}{\jointTaboo}{\test{\chi}} 
				& \equiv \;\test{\usTabooOp{\jointTaboo \cup \parallelCtx}{\chi}} 
				\\[.1em]
			\usInOutOp{}{\jointOut}{\evolution{}{}} 
				& \equiv \evolution
					{x' = \usTabooOp{\jointTaboo\cup\parallelCtx}{\rp}}
					{\usTabooOp{\jointTaboo\cup \parallelCtx}{\chi}} 
				\sidecondition[tag]{with $\jointOut=\jointTaboo\cup\odeBoundVars$} \\
			%
			% communication primitives
			%
			\usInOutOp{}{\jointTaboo\cup\{\ch{}, \historyVar\}}{\send{}{}{}} 
				& \equiv \send{}{}
					{\usTabooOp{\jointTaboo\cup\parallelCtx}{\rp}} 
					\\[.1em]
			\usInOutOp{}{\jointTaboo\cup\{\ch{}, \historyVar, x\}}{\receive{}{}{}} 
				& \equiv \receive{}{}{} 
				\\[.1em]
			%
			% compound programs
			%
			\usInOutOp{}{\jointOut_1 \cup \jointOut_2}{\alpha \cup \beta} 
				& \equiv \usInOutOp{}{\jointOut_1}{\alpha} \cup \usInOutOp{}{\jointOut_2}{\beta} 
				\\[.1em]
			\usInOutOp{}{\jointOut_2}{\alpha \seq \beta}
				& \equiv \usInOutOp{}{\jointOut_1}{\alpha} \seq \usInOutOp{\jointOut_1, \parallelCtx}{\jointOut_2}{\beta} 
				\\[.1em]
			\usInOutOp{}{}{\repetition{\alpha}} 
				& \equiv \repetition{(\usInOutOp{\jointOut, \parallelCtx}{}{\alpha})}
				\sidecondition[tag]{when $\usInOutOp{}{}{\alpha}$ is defined}
				\\[.1em]
			\usInOutOp{}{\jointOut_1 \cup \jointOut_2}{\alpha \parOp \beta} 
				& \equiv 
					\usInOutOp{\jointTaboo, \parCtxOp{\beta}}{\jointOut_1}{\alpha} 
					\parOp 
					\usInOutOp{\jointTaboo, \parCtxOp{\alpha}}{\jointOut_2}{\beta}
		\end{align*}
	\end{minipage}
	\caption{%
		Application of uniform substitution for taboo $\jointTaboo$ and parallel context $\parallelCtx$,
		where $\parCtxOp{\gamma} \equiv \parCtxOpExpanded{\gamma}$ for any program $\gamma$,
		and $\expr \downarrow \cset$ for term $\expr$ is recursive push down of projection $\downarrow \cset$, where $\psymb(\cset_0, \expr) \downarrow \cset \equiv \psymb(\cset_0 \cap \cset, \expr)$.
	}
	\label{fig:unisubs}
	\vspace*{-1.em}
\end{figure}

The side condition $(\SFV(\usubs \fsymb(\usarg)) \cup \SCN(\usubs \fsymb(\usarg))) \cap \jointTaboo = \emptyset$ implements locally that the replacement for $\fsymb$ must not introduce free parameters that are tabooed by~$\jointTaboo$ \ref{itm:channels}.
The substitution $\usAux{\expr}$ 
is responsible for the argument~$\expr$,%
\footnote{Extension to vectorial arguments is straightforward.}
where~$\emptyset$ suffices as the taboo $\jointTaboo$ is already checked on~$\expr \downarrow \cset$.
By the projection, $\expr \downarrow \cset$ only depends on channels~$\cset$.
Quantification $\fa{\arbitraryVar}\!$ taboos the bound variable~$\arbitraryVar$.
Program~$\alpha$ in a box or ac-box has an empty parallel context~$\emptyset$.

The substitution $\usInOutOp{}{}{\alpha}$ computes the output taboo $\jointOut$ by adding the written variables and channels of program~$\alpha$ to $\jointTaboo$,
\eg real variable $x$ for assignment $x \ceq \rp$ and for receiving $\receive{}{}{}$ additionally channel $\ch{}$ and trace variable~$\historyVar$.
The output taboo $\jointOut$ is passed to ac-formulas and postconditions of boxes and ac-boxes for recursive checks for clashes \wrt \ref{itm:channels}.
Crucially for soundness, \rref{lem:unisubs_correct_bound} below proves that $\usInOutOp{}{}{\cdot}$ correctly computes the output taboo $\jointOut$.

The taboo $\jointTaboo\cup\parallelCtx$ passed to nested expressions contains the parallel context~$\parallelCtx$ to prevent free variables in replacements of function and predicate symbols that are bound in parallel.
This prepares the substitution process to preserve the syntax restrictions for parallel composition from previous work \citeDLCHP.%
\footnote{
	For $\alpha\parOp\beta$, the restriction is $(\SV(\alpha) \cap \SBV(\beta)) \cup (\SV(\beta) \cap \SBV(\alpha)) \subseteq \{\globalTime, \globalTime'\} \cup \TVar$ \citeDLCHP.
	However, in this paper, programs obey a less restrictive syntax for simplicity.
}
Substitution for evolution $\evolution{}{}$ considers that the global time $\globalTime, \globalTime'$ is always implicitly bound regardless of whether it occurs in $x, x'$.
The fixpoint notation $\usInOutOp{\jointOut, \parallelCtx}{}{\alpha}$ for the replacement of repetition $\repetition{\alpha}$ ensures that the output taboo of the first iteration 
is tabooed in the subsequent iterations \cite{DBLP:conf/cade/Platzer19}.
Computing the parallel context of $\alpha$ and $\beta$ in case $\alpha \parOp \beta$ requires one additional pass for both subprograms
because what they potentially bind after substitution adds to the parallel context of the respective other subprogram.

\begin{lemma}[Correct output taboo] \label{lem:unisubs_correct_bound}
	Application $\usInOutOp{}{}{\alpha}$ of uniform substitution retains input taboo $\jointTaboo$ and correctly adds the bound variables and written channels of program $\alpha$,
	\iest $\jointOut \supseteq \jointTaboo \cup \SBV(\usInOutOp{}{}{\alpha}) \cup \SCN(\usInOutOp{}{}{\alpha})$.
	%\reportonly{ (\rref{app:soundness})}.
\end{lemma}

The side condition of $\usInOutOp{}{}{\rpconst{}{}}$ maintains local abstraction of subprograms \ref{itm:abstraction} because the replacement cannot bind more than $\rpconst{}{}$, 
thus cannot bind variables and channels of an abstraction that is independent of $\rpconst{}{}$.
This also preserves state-disjointness (well-formedness) of parallel programs.

\subsection{Semantic Effect of Uniform Substitution} \label{sec:unisubsLemmas}

The key ingredients for proving soundness of uniform substitution are \rref{lem:unisubs_term} and \ref{lem:unisubs_fml_prog} below.
They prove that the effect of the syntactic transformation applied by uniform substitution can be equally mimicked by  
semantically modifying the interpretation of function and predicate symbols, and program constants. 
This adjoint interpretation $\interadj[w]$ for interpretation $\inter$ and state~$\pstate{w}$ changes how symbols are interpreted according to their syntactic replacements in the substitution $\usubs$.

\newcommand{\argsort}{\anysort_\textnormal{arg}}
\begin{definition}[Adjoint substitution]
	\label{def:adjoint}
	For interpretation $\inter$ and state~$\pstate{w}$,
	the \emph{adjoint interpretation} $\interadj[w]$ changes the meaning of function and predicate symbols, and program constants according to the substitution $\usubs$ evaluated in state~$\pstate{w}$:
	\begin{align*}
		& \interadj[w](\fsymb[some] : \argsort) : \argsort \rightarrow \anysort; d \mapsto \sem{\usubs \fsymb(\usarg)}{\inter \subs{\usarg}{d} \pstate{w}} \hspace*{-1em}
			&&\sidecondition[color=black]{\quad where $\anysort, \argsort \in \{ \reals, \naturals, \Chan, \traces \}$} \\
		& \interadj[w](\psymb : \argsort) = \{ d \in \argsort \mid \inter \subs{\usarg}{d} \pstate{w} \vDash \usubs \psymb(\usarg) \} 
			&&\sidecondition[color=black]{\quad where $\argsort \in \{ \reals, \naturals, \Chan, \traces \}$} \\
		& \interadj[w](\rpconst{}{}) = \sem{\usubs \pconst}{\inter}
	\end{align*}
\end{definition}

We follow the observation for \dGL \cite{DBLP:conf/cade/Platzer19} that the more liberal one-pass substitution
requires stronger coincidence between the substitution and the adjoint on neighborhoods of the original state.
Where the \dGL soundness proof has succeeded by a neighborhood semantics of state on taboos,
the \dLCHP proof succeeds with a generalization to a neighborhood semantics of state and communication on taboos.
The neighborhood of a state consists of its variations: 

\begin{definition}[Variation] \label{def:variation}
	For a set $\jointTaboo \subseteq \V \cup \Chan$,
	a state $\pstate{v}$ is a \emph{$\jointTaboo$-variation of state~$\varioOrigin$} if $\pstate{v}$ and $\varioOrigin$ only differ on variables or projections onto channels in $\jointTaboo$, \iest $\pstate{v} \downarrow (\jointTaboo^\complement \cap \Chan) = \varioOrigin \downarrow (\jointTaboo^\complement \cap \Chan)$ on $\jointTaboo^\complement \cap \V$.
\end{definition}

\iflongversion
Variations are monotone in the set $\jointTaboo$,
\iest if $\pstate{v}$ is a $\jointTaboo_0$-variation of $\pstate{w}$ and $\jointTaboo_0 \subseteq \jointTaboo_1$,
then $\pstate{v}$ is a $\jointTaboo_1$-variation of $\pstate{w}$.
\todo{Transitive and monotone}
\else\fi

The proofs of \rref{lem:unisubs_term} and \ref{lem:unisubs_fml_prog} follow a lexicographic induction on the structure of substitution, and term, formula, or program.
In \rref{lem:unisubs_fml_prog}, the induction is mutual for formulas and programs.

\begin{lemma}[Semantic uniform substitution] \label{lem:unisubs_term}
	The term $\expr$ evaluates equally over $\jointTaboo$-variations under uniform substitution $\usTaboo{}$ and adjoint interpretation $\interadj$, 
	\iest $\sem{\usTabooOp{}{\expr}}{\lstate{v}} = \sem{\expr}{\lstateadj{v}}$ for all $\jointTaboo$-variations $\pstate{v}$ of $\varioOrigin$.
\end{lemma}

\begin{lemma}[Semantic uniform substitution] \label{lem:unisubs_fml_prog}
	The formula $\phi$ and the program $\alpha$ have equal truth value and semantics, respectively, over $\jointTaboo$-variations under uniform substitution $\usTaboo{}$ and adjoint interpretation $\interadj$, \iest
	\begin{enumerate}
		\item for all $\jointTaboo$-variations $\pstate{v}$ of $\varioOrigin$: $\lstate{v} \vDash \usTabooOp{}{\phi}$ iff $\lstateadj{v} \vDash \phi$
		\item for all $(\jointTaboo\cup\parallelCtx)$-variations $\pstate{v}$ of $\varioOrigin$: $\computation[fin={o}] \in \sem{\usInOutOp{}{}{\alpha}}{\inter}$ iff $\computation[fin={o}] \in \sem{\alpha}{\interadj}$
	\end{enumerate}
\end{lemma}

\subsection{Uniform Substitution Proof Rule}

The proof rule \RuleName{US} for uniform substitution is the single point of truth for the sound instantiation of axioms (plus renaming of bound variables \cite{DBLP:journals/jar/Platzer17} and written channels, \eg $[x \ceq \rp] \psi(x)$ to $[y \ceq \rp] \psi(y)$ and $[\receive{}{}{}] \psi(\ch{})$ to $[\receive{\ch{dh}}{}{}] \psi(\ch{dh})$).
Soundness of the rule, 
\iest that validity of its premise implies validity of the conclusion, immediately follows from \rref{lem:unisubs_fml_prog}.
Since the substitution process starts with no taboos, $\usubs(\phi)$ is short for $\usTabooOp{\emptyset}{\phi}$.
If the substitution clashes, \iest $\usTabooOp{\emptyset}{\phi}$ is not defined, then rule \RuleName{US} is not applicable.

\begin{theorem}[\RuleName{US} is sound] \label{thm:ussound}
	The proof rule \RuleName{US} is sound.
	%\reportonly{ (\rref{app:soundness})}.
	\begin{prooftree}[ProofTreeAboveSkip=-1em]
		\Axiom{$\phi$}
	
		\RuleNameRight{US}
		\UnaryInf{$\usof{\phi}$}
	\end{prooftree}
\end{theorem}

Unlike \dL \cite{DBLP:journals/jar/Platzer17} and \dGL \cite{DBLP:conf/cade/Platzer19}, \dLCHP
has a context-sensitive syntax for programs and formulas (see \rref{def:syntax_chps} and \rref{def:syntax_formulas}).
By \rref{prop:unisubs_well_formed}, uniform substitution, however, preserves syntactic well-formedness.
Since all axioms in \rref{sec:calculus} will be well-formed,
only well-formed formulas can be derived in \dLCHP.

\begin{proposition}[\RuleName{US} preserves well-formedness] \label{prop:unisubs_well_formed}
	The result $\usTabooOp{}{\phi}$ (if defined) of applying uniform substitution to a well-formed formula $\phi$ is well-formed.
	%\reportonly{ (\rref{app:soundness})}.
\end{proposition}

\vspace*{-.3cm}
\section{Axiomatic Proof Calculus} \label{sec:calculus}

Figure~\ref{fig:calculus} presents a sound proof calculus for \dLCHP.
The significant difference to \dLCHP's schematic calculus \citeDLCHP is that it completely abandons soundness-critical side conditions, 
internalizing them syntactically in the axioms.
Only axiom
\RuleName{assumptionWeak} was adjusted to obtain a symbolic representation and an ac-version~\RuleName{acModalMP} of modal modus ponens is included.
Now, distribution of ac-boxes over conjuncts~\RuleName{acBoxesDist} and ac-monotonicity \RuleName{acMono} derive from \RuleName{acModalMP},%\reportonly{ (\rref{app:calculus})},
thus are dropped.
Except for the small changes
%\reportonly{ (see \rref{app:calculus})} 
soundness is inherited from the schematic axioms \citeDLCHP.
\reportonly{All proofs and supplementary material for this section are in \rref{app:calculus}.}

Algebraic laws for reasoning about traces \citeDLCHP can be easily adapted to uniform substitution as well\ifreport%
% (\rref{app:calculus})
\else\citeReport\fi{}.
Decidable first-order real arithmetic \cite{Tarski1951} and Presburger arithmetic \cite{Presburger1931} have corresponding oracle proof rules \citeDLCHP. 

\begin{remark} \label{rem:finite}
	To obtain a truly finite list of axioms from \rref{fig:calculus},
	symbolic (co)finite sets can be finitely axiomatized as a boolean algebra together with extensionality, which can be unrolled to a finite disjunction for (co)finite sets
	\ifreport
	(see \rref{app:calculus})%
	\else\citeReport\fi{}.
\end{remark}

\vspace*{-.2cm}
\subsubsection{Parallel Composition}

The parallel injection axiom \RuleName{acDropComp} in \rref{fig:calculus} decomposes parallel CHPs by local abstraction \ref{itm:abstraction}.
Unlike \dLCHP's \citeDLCHP and Hoare-style \cite{AcHoare_Zwiers,AcSemantics_Zwiers} schematic calculi for ac-reasoning, 
axiom \RuleName{acDropComp} internalizes the noninterference property \citeDLCHP[Def.\,7] that determines valid instances of formula
\begin{equation}
	\setlength{\abovedisplayskip}{3pt}
	\setlength{\belowdisplayskip}{3pt}
	\label{eq:schematic_injection}
	[ \alpha ] \ac \psi \rightarrow [ \alpha \parOp \beta ] \ac \psi
\end{equation}
purely syntactically.
To focus on noninterference,
$\rpconst{\cset_a}{\bvarvec} \wfParOp \rpconst[alt]{\cset_b}{\bvarvec[alt]}$ abbreviates well-formed parallel composition $\rpconst{\cset_a}{\varvec_a} \parOp \rpconst[alt]{\cset_b}{(\varvec_b \cap \varvec_a^\complement) \cup \{\globalTime, \globalTime'\} \cup \TVar}$ using operator $\wfParOp$
for program constants $\rpconst{\cset_a}{\bvarvec}$, $\rpconst[alt]{\cset_b}{\bvarvec[alt]}$.
This notation ensures disjoint parallel state except for the global time $\globalTime, \globalTime'$ and recorder variables~$\TVar$.
%Since channels can be arbitrarily shared, $\cset_a$ and $\cset_b$ are independent.

Intuitively, axiom \RuleName{acDropComp} restricts $\beta$
in \rref{eq:schematic_injection} such that $\alpha$ overapproximates the behavior of $\alpha \parOp \beta$ influencing $\A$, $\C$, or $\psi$.
For this purpose, noninterference internalized in $\rpconst[alt]{\cset_b \cap (\cset^\complement \cup \cset_a)}{\varvec^\complement}$ forbids~$\pconst[alt]$ to bind variables $\varvec$ that are free in the postcondition $\psymb(\cset, \varvec)$, 
and $\cset^\complement$ forbids $\pconst[alt]$ to bind channels $\cset$ (except for channels $\cset_a$ written by $\pconst$ because joint parallel communication can already be observed from~$\pconst$, too).
\ifreport
The cut with $\cset_b$ allows downscaling of the channels $\pconst[alt]$ has to bind.
Since parallel programs always agree on the global time $\globalTime, \globalTime'$ and the communication recorded by trace variables $\TVar$, 
the operator~$\wfParOp$ allows their sharing even if $\varvec^\complement$ disallows it.
Note that $\cset_a$ and $\cset$, and $\varvec_a$ and $\varvec$ may overlap.
\else
Moreover, parallel programs always agree on the global time $\globalTime, \globalTime'$ and the communication recorded by trace variables $\TVar$.
Therefore, the operator~$\wfParOp$ explicitly allows their sharing even if $\varvec^\complement$ disallows it.
Note that $\cset_a$ and $\cset$, and $\varvec_a$ and $\varvec$ may overlap but can also be disjoint.
\fi

Despite its asymmetric shape, axiom \RuleName{acDropComp} decomposes $[ \alpha {\parOp} \beta ] (\phi \wedge \psi)$ into $[ \alpha ] \phi$ and $[ \beta ] \psi$ (if they mutually do not interfere) 
via independent proofs for $[ \alpha {\parOp} \beta] \phi$ and $[ \alpha {\parOp} \beta ] \psi$,
which drop either $\alpha$ or $\beta$ by \RuleName{acDropComp} modulo commutativity. 

\begin{figure}[tb]
	\begin{small}
		\begin{minipage}{\textwidth}
			\begin{calculus}
				\startAxiom{assign}
					$[ x \ceq \fsymb[alt, real]] \psymb(x) \leftrightarrow \psymb(\fsymb[alt, real])$%
					\footnote{\label{ft:folR}Replacements for function symbol $\fsymb[alt, real]$ and predicate symbol $\psymb[alt, real]$ are restricted to polynomials in $\RVar$ and first-order real arithmetic, respectively.}
				\stopAxiom
				\startAxiom{nondetAssign}
					$[ x \ceq * ] \psymb(x) \leftrightarrow \fa{x} \psymb(x)$
				\stopAxiom
				\startAxiom{test}
					$[ \test{\psymb[alt, real]} ] \psymb \leftrightarrow (\psymb[alt, real] \rightarrow \psymb)$%
					\footnoteref{ft:folR}
				\stopAxiom
				\startAxiom{boxesDual}
					$[ \pconst ] \Psymb \leftrightarrow [ \pconst ] \acpair{\true, \true} \Psymb$
				\stopAxiom
			\end{calculus}\hspace*{.25em}%
			\begin{calculus}
				\startAxiom{acComposition}
					$[ \pconst \seq \pconst[alt] ] \ACsymb \Psymb \leftrightarrow [ \pconst ] \ACsymb [ \pconst[alt] ] \ACsymb \Psymb$
				\stopAxiom
				\startAxiom{acChoice}
					$[ \pconst \cup \pconst[alt] ] \ACsymb \Psymb \leftrightarrow [ \pconst ] \ACsymb \Psymb \wedge [ \pconst[alt] ] \ACsymb \Psymb$
				\stopAxiom
				\startAxiom{acIteration}
					$[ \repetition{\pconst} ] \ACsymb \Psymb \leftrightarrow [ \pconst^0 ] \ACsymb \Psymb \wedge [ \pconst ] \ACsymb [ \repetition{\pconst} ] \ACsymb \Psymb$%
					\footnote{Recall that $[\alpha^0] \ACsymb \Psymb \leftrightarrow \Csymb \wedge (\Asymb \rightarrow \Psymb)$ by \RuleName{acNoCom} and \RuleName{test} since $\alpha^0 \equiv \; \test{\true}$.}
				\stopAxiom
				\startAxiom{assumptionWeak}
					$\compCondDomain \wedge [ \pconst ] \acpair{\Asymb_1 \wedge \Asymb_2, \Csymb_1 \wedge \Csymb_2} \Psymb \rightarrow [ \pconst ] \acpair{\Asymb, \Csymb_1 \wedge \Csymb_2} \Psymb$%
					\footnote{$\compCondition$ is the compositionality condition $\compConditionExpanded$.}
				\stopAxiom				
			\end{calculus}
				
			\begin{calculus}
				\startAxiom{acDropComp}
					$[ \rpconst{\cset_a}{\bvarvec} ] \ACsymb \psymb(\cset, \varvec) \rightarrow [ \rpconst{\cset_a}{\bvarvec} \wfParOp \rpconst[alt]{\cset_b\cap(\cset^\complement \cup \cset_a)}{\varvec^\complement} ] \ACsymb \psymb(\cset, \varvec)$%
					\footnote{The operator $\wfParOp$ abbreviates well-formed parallel composition (see above).}
				\stopAxiom
				\startAxiom{gtime}
					$[\evolution{\rvarvec' = \fsymb[alt, real](\rvarvec, \globalTime)}{\psymb[alt, real](\rvarvec, \globalTime)}] \psymb(\rvarvec, \globalTime) \leftrightarrow [\evolution{\globalTime' = 1, \rvarvec' = \fsymb[alt, real](\rvarvec, \globalTime)}{\psymb[alt, real](\rvarvec, \globalTime)}] \psymb(\rvarvec, \globalTime)$%
					\footnoteref{ft:folR}
				\stopAxiom
			\end{calculus}
			
			\begin{calculus}
				\startAxiom{send}
					$[ \send{}{}{\fsymb[alt, real]} ] \psymb(\ch{}, \historyVar) 
					\leftrightarrow \fa{\historyVar_0} 
						\big(
							\historyVar_0 = \historyVar \cdot \comItem{\ch{}, \fsymb[alt, real], \globalTime} \rightarrow \psymb(\ch{}, \historyVar_0)
						\big)$
				\stopAxiom
				\startAxiom{acCom}
					$[ \send{}{}{\fsymb[alt, real]} ] \acpair{\Acom, \Ccom} \Pcom$
					$\leftrightarrow \Ccom \wedge \Big( \Acom \rightarrow [ \send{}{}{\fsymb[alt, real]} ] \big( \Ccom \wedge (\Acom \rightarrow \Pcom ) \big)\Big)$%
				\stopAxiom 
				\startAxiom{comDual}
					$[ \receive{}{}{} ] \acpair{\Acom, \Ccom} \psymb(\ch{}, \historyVar, x) \leftrightarrow [ x \ceq * ] [ \send{}{}{x} ] \acpair{\Acom, \Ccom} \psymb(\ch{}, \historyVar, x)$
				\stopAxiom
				\startAxiom{acNoCom}
					$[ \rpconst{\emptyset}{\RVar} ] \acpair{\Asymb, \Csymb} \Psymb \leftrightarrow \Csymb \wedge (\Asymb \rightarrow [ \rpconst{\emptyset}{\RVar} ] \Psymb)$
				\stopAxiom
				\startAxiom{acWeak}
					$[ \pconst ] \ACsymb \Psymb \leftrightarrow \Csymb \wedge [ \pconst ] \ACsymb ( \Csymb \wedge (\Asymb \rightarrow \Psymb) )$
				\stopAxiom
				\startAxiom{acInduction}
					$[ \repetition{\pconst} ] \ACsymb \Psymb \leftrightarrow [ \pconst^0 ] \ACsymb \Psymb \wedge [\repetition{\pconst}] \acpair{\Asymb, \true} (\Psymb \rightarrow [ \pconst ] \ACsymb \Psymb)$
				\stopAxiom
				\startAxiom{acModalMP}
					$[ \pconst ] \acpair{\psymb[ac, opt], \psymb[ac, alt]_1 \rightarrow \psymb[ac, alt]_2} (\psymb[ac]_1 \rightarrow \psymb[ac]_2) \rightarrow ([ \pconst ] \acpair{\psymb[ac, opt], \psymb[ac, alt]_1} \psymb[ac]_1 \rightarrow [ \pconst ] \acpair{\psymb[ac, opt], \psymb[ac, alt]_2} \psymb[ac]_2)$
				\stopAxiom
			\end{calculus}
			\hspace*{0em}
			\begin{calculus}[r]
				\startRule{MP}
					\Axiom{$\psymb \rightarrow \psymb[alt]$}
					\Axiom{$\psymb$}
					\BinaryInf{$\psymb[alt]$}
				\stopRule
				\startRule{acG}
					\Axiom{$\Csymb \wedge \Psymb$}
					\UnaryInf{$[ \pconst ] \ACsymb \Psymb$}
				\stopRule
				\startRule{forall}
					\Axiom{$\psymb(x)$}
					\UnaryInf{$\fa{x} \psymb(x)$}
				\stopRule
				\startRule[OverhangLeft=0em, OverhangRight=0em]{CE}
					\Axiom{$\Psymb_1 \leftrightarrow \Psymb_2$}
					\UnaryInf{$C(\Psymb_1){\leftrightarrow}C(\Psymb_2)$}
				\stopRule
			\end{calculus}
			\vspace*{.4em}

			\begin{small}
				$\Psymb_j \equiv \psymb_j(\cset, \varvec)$, and $\Asymb_j \equiv \psymb[opt]_j(\cset, \hvarvec)$, and $\Csymb_j \equiv \psymb[alt]_j(\cset, \hvarvec)$, and $\widehat{\chi} \equiv \chi(\ch{}, \historyVar)$,
				where $j$ may be blank, 
				and $\cset \subseteq \Chan$, $\varvec \subseteq \RVar \cup \TVar$, and $\hvarvec \subseteq \TVar$ are (co)finite.
			\end{small}
		\end{minipage}
	\end{small}
	\caption{\dLCHP proof calculus}
	\vspace*{-1mm}
	\label{fig:calculus}
\end{figure}

\subsubsection{Axiom System}

For each program statement,
there is either a dynamic or an ac-axiom because the respective other version derives by axiom \RuleName{boxesDual} or \RuleName{acNoCom}.
Axioms \RuleName{assign}, \RuleName{nondetAssign}, and \RuleName{test} are as in \dL \cite{DBLP:journals/jar/Platzer17}.
Axioms \RuleName{acComposition}, \RuleName{acChoice}, and \RuleName{acIteration} for decomposition, and \RuleName{acInduction} for induction carefully generalize their versions in differential \cite{DBLP:journals/jar/Platzer17} dynamic \cite{Harel1979} logic to ac-reasoning.
Sending is handled step-wise via flattening the assumption-commitments by axiom \RuleName{acCom} and axiom \RuleName{send} that executes the effect onto the recorder $\historyVar$.
The duality \RuleName{comDual} turns receiving into arbitrary sending,
which only synchronizes if it agrees with the parallel context on the value.
Usage of axiom \RuleName{acWeak} is for convenience.
Axiom \RuleName{gtime} materializes the flow of global time~$\globalTime$ such that \dL's axiomatization of continuous evolution \cite{DBLP:journals/jar/Platzer17} gets applicable,
which requires ODE shape $\rvarvec' = \fsymb[real](\rvarvec)$.
The axiomatic proof rules \RuleName{acG}, \RuleName{MP}, \RuleName{forall}, and \RuleName{CE} are an ac-version of G\"odels generalization rule, modus ponens, quantifier elimination, and contextual equivalence, respectively.

The axiom \RuleName{assumptionWeak} can weaken assumptions.
Its slight change compared to \dLCHP's schematic calculus \citeDLCHP exploits that the compositionality condition $\compCondition$ is only required for $\pconst$'s reachable worlds. 
Interestingly, \dLCHP's monotonicity rule \RuleName{acMono} \citeDLCHP does not derive from modal modus ponens \RuleName{acModalMP} and G\"odel generalization \RuleName{acG} in analogy to \dL~\cite{DBLP:journals/jar/Platzer17}
but needs \RuleName{acWeak} handling monotonicity of assumptions,
which does not fit into \RuleName{acG} because necessitating the assumption in \RuleName{acG} would render the derivation of $[ \alpha ] \acpair{\false, \true} \true$ by \RuleName{acG} impossible.

Axioms using postcondition $\Psymb \equiv \psymb(\cset, \varvec)$,
\eg in \RuleName{acComposition},
allow any replacement of $\Psymb$ since accessed channels $\cset \subseteq \Chan$ and free variables $\varvec \subseteq \RVar \cup \TVar$ can be arbitrary.
Replacements of assumptions $\Asymb \equiv \asymb(\cset, \hvarvec)$ and commitments $\Csymb \equiv \csymb(\cset, \hvarvec)$ can instead only mention trace variables $\hvarvec \subseteq \TVar$ bound in their context. 
This reflects that trace variables are the only interface between the program $\alpha$ and the ac-formulas $\A$ and $\C$ in an ac-box $[ \alpha ] \ac \psi$ (well-formedness).

\begin{theorem}[Soundness] \label{thm:soundness}
	The proof calculus for \dLCHP presented in \rref{fig:calculus} is sound as an instantiation of the schematic calculus \citeDLCHP.
\end{theorem}

\subsubsection{Clashes}

\newcommand{\omittedAC}{\acpair{\true, \true}}
\newcommand{\psiPlaceholder}{\psymb(\cset, \historyVar, y)}

Clashes sort out unsound instantiations of axioms.
Unlike in \dL and \dGL \cite{DBLP:journals/jar/Platzer17, DBLP:conf/cade/Platzer19} whose clashes are solely due to tabooed variables in terms and formulas,
clashes in \dLCHP can also be due to tabooed channels, 
and even due to taboos in programs.
For example, the substitution $\usubs = \{ \pconst \smapsto \send{\ch{gh}}{}{1}, \pconst[alt] \smapsto \send{}{}{2}, \psymb \smapsto \psi, \asymb \smapsto \true, \csymb \smapsto \true \}$ 
with $\psi \equiv \len{\historyVar \downarrow \ch{}} > 0 \wedge \len{\historyVar \downarrow \ch{dh}} > 0 \wedge y < 0$ clashes below,
where $\cset = \{ \ch{}, \ch{dh} \}$, 
and $\varvec \equiv \historyVar, y$, 
and $\Asymb \equiv \asymb(\cset)$, and $\Csymb \equiv \csymb(\cset)$.
Writing channel $\ch{}$ in the replacement for $\pconst[alt]$ would break the local abstraction of~$\pconst$ as  $\ch{}$ is accessed in~$\psi$ 
but not written in the replacement for $\pconst$,
thus the clash indeed sorts out an unsound instantiation.
\vspace*{-1em}%
\begin{small}
	\begin{prooftree}[OverhangLeft=3pt, OverhangRight=3pt]
		\Axiom{$[ \rpconst{\{ \ch{gh} \}}{\historyVar} ] \acpair{\Asymb, \Csymb} \psymb(\cset, \varvec) \rightarrow [ \rpconst{\{ \ch{gh} \}}{\historyVar} \wfParOp \rpconst[alt]{\{\ch{}\} {\cap} (\cset^\complement {\cup} \{ \ch{gh}) \}}{\varvec^\complement} ] \acpair{\Asymb, \Csymb} \psymb(\cset, \varvec)$}
	
		\RuleNameRight{clash}
		\UnaryInf{$[ \send{\ch{gh}}{}{1} ] \omittedAC \psi \rightarrow [ \send{\ch{gh}}{}{1} \parOp \send{}{}{2} ] \omittedAC \psi$}
	\end{prooftree}
\end{small}

In contrast, $\usubs = \{ \pconst \smapsto \receive{}{}{x} \seq \send{\ch{gh}}{}{1}, \pconst[alt] \smapsto \send{}{}{2}, \psymb \smapsto \psi, \asymb \smapsto \true, \csymb \smapsto \true \}$
does not clash below,
where $\cset = \{ \ch{}, \ch{dh} \}$, and $\cset_a = \{ \ch{}, \ch{gh} \}$, 
and other abbreviations are as above,
because $\ch{} \in \cset^\complement \cup \cset_a = \{ \ch{dh} \}^\complement$.
Intuitively, the $\ch{}$-communication of $\pconst[alt]$ remains observable after dropping $\pconst[alt]$ from the parallel composition as it is joint with~$\pconst$.%
\vspace*{-1em}%
\begin{small}
	\begin{prooftree}
		\Axiom{$*$}
	
		\RuleNameRight{acDropComp}
		\UnaryInf{$[ \rpconst{\cset_a}{\historyVar, x} ] \acpair{\Asymb, \Csymb} \psymb(\cset, \varvec) \rightarrow [ \rpconst{\cset_a}{\historyVar, x} \wfParOp \rpconst[alt]{\{\ch{}\} \cap (\cset^\complement \cup \cset_a)}{\varvec^\complement} ] \acpair{\Asymb, \Csymb} \psymb(\cset, \varvec)$}
	
		\RuleNameRight{US}
		\UnaryInf{$[ \receive{}{}{x} \seq \send{\ch{gh}}{}{1} ] \omittedAC \psi \rightarrow [ (\receive{}{}{x} \seq \send{\ch{gh}}{}{1}) \parOp \send{}{}{2} ] \omittedAC \psi$}
	\end{prooftree}
\end{small}

Also note that by the operator $\wfParOp$ for well-formed parallel composition, the recorder variable $\historyVar$ can be shared without causing a clash above.
However, clashes prevent instantiation that would violate syntactic  well-formedness of programs (\rref{def:syntax_chps}) by binding the same state variable in parallel:%
\vspace*{-1em}%
\begin{small}
	\begin{prooftree}
		\Axiom{$[ \rpconst{\emptyset}{x} ] \acpair{\asymb, \csymb} \psymb(x, y) \rightarrow [ \rpconst{\emptyset}{x} \wfParOp \rpconst[alt]{\emptyset}{\{x, y\}^\complement} ] \acpair{\asymb, \csymb} \psymb(x, y)$}
	
		\RuleNameRight{clash}
		\UnaryInf{$[ x \ceq y ] \omittedAC y = x \rightarrow [ x \ceq y \parOp x \ceq 0 ] \omittedAC y = x$}
	\end{prooftree}
\end{small}

Well-formedness of programs and formulas is ensured in the axioms by well-formed parallel composition $\wfParOp$ and limitation to trace variables $\hvarvec$ in \(\Asymb_j \equiv \psymb[opt]_j(\cset, \hvarvec)\) and \(\Csymb_j \equiv \psymb[alt]_j(\cset, \hvarvec)\) in ac-boxes $[ \alpha ] \acpair{\Asymb_j, \Csymb_j} \psi$ in \rref{fig:calculus}, respectively.
By \rref{prop:unisubs_well_formed}, uniform substitution always preserves well-formedness.

\newcommand{\colorRatio}{5}
\newcommand{\colortikz}[2]{\tikz[baseline=-.5ex]\node[rectangle, fill=#1!\colorRatio, inner sep=.3mm]{#2};}

\newcommand{\acBox}[1]{[ #1 ] \acpair{\Asymb, \Csymb} \Psymb}

\tikzstyle{every picture}+=[remember picture]
\tikzstyle{na} = [baseline=-.5ex]

\tikzstyle{proofnode} = [inner sep=0, outer sep=0]

\newcommand{\ColoredAxiom}[2]{\Axiom{\tikz \node[fill=#1!\colorRatio, baseline=-.5ex, proofnode, inner sep=.5ex] {#2};}}
\newcommand{\ColoredUnary}[2]{\UnaryInf{\tikz \node[fill=#1!\colorRatio, baseline=-.5ex, proofnode] {#2};}}
\newcommand{\ColoredBinary}[2]{\BinaryInf{\tikz \node[fill=#1!\colorRatio, baseline=-.5ex, proofnode] {#2};}}

\begin{example}
	The proof tree below decomposes safety (\rref{ex:safety}) of cruise control (\rref{ex:prog}) into safety \circled{1} of controller $\controllerName$ and branch \circled{2} to be continued to safety of the vehicle \vehicleName.
	The \colortikz{green}{lower subproof} introduces the ac-formulas 
	\begin{equation*}
		\A \equiv \C \equiv \big( \len{\historyVar \downarrow \ch{tar}} > 0 \rightarrow \vrange{\val{\historyVar \downarrow \ch{tar}}} \big)
	\end{equation*}
	using axiom \RuleName{assumptionWeak} to abstract from the communication between $\controllerName$ and $\vehicleName$.
	The \colortikz{blue}{upper subproof} uses the parallel injection axiom \RuleName{acDropComp} to drop $\vehicleName$.
	Uniform substitution \RuleName{US} does not clash as the commitment $\C$ only refers to joint communication of $\controllerName$ and $\vehicleName$.
	Other applications of \RuleName{US} (\eg for \RuleName{assumptionWeak}) are omitted.
	Rule \RuleName{prop} denotes propositional reasoning.
	Abbreviations are as follows: $\alpha \equiv \rpconst{\ch{tar}}{\vtar{\controllerName}, t, t', \globalTime, \globalTime', \historyVar}$, $\Asymb \equiv \asymb(\ch{tar}, \historyVar)$, $\Csymb \equiv \csymb(\ch{tar}, \historyVar)$, $\Psymb  \equiv \psymb(\ch{tar})$.

	\begin{small}
		\hspace*{-3em}\begin{tikzpicture}
			\node (lower)
			{\vbox{
				\begin{prooftree*}
					\ColoredAxiom{green}{$*$}
			
					\RuleNameLeft{prop}
					\ColoredUnary{green}{$(\C \rightarrow \A) \wedge \true$}
			
					\RuleNameLeft{acG}
					\ColoredUnary{green}{$\ccpre \rightarrow [ \controller {\parOp} \vehicle ] \acpair{\true, \C \rightarrow \A} \true$}
			
					\Axiom{}
			
					\Axiom{\circled{2}}
			
					\UnaryInf{$\ccpre \rightarrow [ \controller {\parOp} \vehicle ] \acpair{\true \wedge \A, \true} \ccsafe$}

					\SetOption{HypSeparation}{8em}
			
					\RuleNameRight{andR}
					\BinaryInf{$\ccpre \rightarrow [ \controller {\parOp} \vehicle ] \acpair{\true \wedge \A, \C} \true \wedge [ \controller {\parOp} \vehicle ] \acpair{\true \wedge \A, \true} \ccsafe$}
			
					\RuleNameRight{acBoxesDist}
					\ColoredUnary{green}{$\ccpre \rightarrow [ \controller {\parOp} \vehicle ] \acpair{\true \wedge \A, \C \wedge \true} (\true \wedge \ccsafe)$}	
	
					\SetOption{HypSeparation}{1em}
	
					\RuleNameRight{andR}
					\ColoredBinary{green}{$\ccpre \rightarrow [ \controller {\parOp} \vehicle ] \acpair{\true, \C \rightarrow \A} \true \wedge [ \controller {\parOp} \vehicle ] \acpair{\true \wedge \A, \C \wedge \true} (\true \wedge \ccsafe)$}
				
					\RuleNameRight{assumptionWeak}
					\ColoredUnary{green}{$\ccpre \rightarrow [ \controller {\parOp} \vehicle ] \acpair{\true, \C \wedge \true} (\true \wedge \ccsafe)$}
				
					\RuleNameRight{boxesDual, acMono}
					\ColoredUnary{green}{$\ccpre \rightarrow [ \controller {\parOp} \vehicle ] \ccsafe$}
				\end{prooftree*}
			}};
	
			\node (upper) [above right=-2em and -37.5em of lower]
			{\vbox{
				\begin{prooftree*}
					\Axiom{\circled{1}}
			
					\ColoredUnary{blue}{$\ccpre \rightarrow [ \controller ] \acpair{\true, \C} \true$}
			
					\ColoredAxiom{blue}{$*$}
			
					\RuleNameRight{acDropComp}
					\ColoredUnary{blue}{$\acBox{\alpha} \rightarrow \acBox{\alpha \wfParOp \rpconst[alt]{\ch{tar}}{\vtar{\vehicleName}, \acceleration, t_0, \velo{\vehicleName}, \velo{\vehicleName}'}}$}
			
					\RuleNameRight{US}
					\ColoredUnary{blue}{$[ \controller ] \acpair{\true, \C} \true \rightarrow [ \controller {\parOp} \vehicle ] \acpair{\true, \C} \true$}
			
					\RuleNameRight{MP, CE}
					\ColoredBinary{blue}{$\ccpre \rightarrow [ \controller {\parOp} \vehicle ] \acpair{\true, \C} \true$}
			
					\RuleNameRight{acMono}
					\ColoredUnary{blue}{$\ccpre \rightarrow [ \controller {\parOp} \vehicle ] \acpair{\true \wedge \A, \C} \true$}
				\end{prooftree*}	
			}};
	
			\draw[-stealth, thick] (lower) ++ (.3, .85) to [in=320, out=140] ([shift={(-1.3, 0)}]upper.south);
		\end{tikzpicture}
	\end{small}
\end{example}

\section{Related Work} \label{sec:related}

Uniform substitution for differential dynamic logic \dL \cite{DBLP:journals/jar/Platzer17} generalizes Church's %original 
uniform substitution for first-order logic \cite[\S 35, 40]{Church1956}.
Unlike the lifting from \dL to differential game logic \dGL \cite{DBLP:conf/cade/Platzer18},
\dLCHP generalizes into the complementary direction of communication and parallelism.
Unlike schematic calculi \cite{OwickiGries1976, Apt1980, LevinGries1981, AcHoare_Zwiers, Xu1997}, whose treacherous schematic simplicity relies on encoding all subtlety of parallel systems in significant soundness-critical side conditions,
our development builds upon a minimalistic non-schematic parallel injection axiom \emph{and} sound instantiation encapsulated in uniform substitution.
This provides a new, more atomic and more modular understanding of parallel systems overcoming the root cause for large soundness-critical prover kernels \cite{Kirchner2015, Blom2017, Cohen2009, Schellhorn2022, Jacobs2011, GibsonRobinson2014}.
Usage of uniform substitution reduced the kernel of the theorem prover \keymaera from 105 \kloc to 2 \kloc in \keymaerax \cite{DBLP:series/lncs/MitschP20}.
We expect \dLCHP's integration into \keymaerax to stay in the same order of magnitude.

To the best of our knowledge,
assumption-commitment reasoning \cite{Misra1981,AcHoare_Zwiers}%
\footnote{	
Assumption-commitment and rely-guarantee reasoning are specific patterns for message-passing and shared variables concurrency, respectively.
The broader assume-guarantee principle has been used across diverse areas for various purposes.}
has no tool support,
which might be due to vast implementation effort.
The latter can be underpinned by analogy with tools \cite{Kirchner2015, Blom2017, Cohen2009, Schellhorn2022, Jacobs2011} for verification of shared-variables concurrency,
some of which use rely-guarantee reasoning \cite{Smans20140501,Schellhorn2022}.
Unlike uniform substitution for \dLCHP that enables a straightforward implementation of a small prover kernel, they all rely on large soundness-critical code bases.
Unlike refinement checking for CSP \cite{GibsonRobinson2014} and discrete-time CSP \cite{Armstrong2014}, \dLCHP supports safety properties of dense-time hybrid systems.
Contrary to our goal of small prover kernels, implementations of model checkers \cite{GibsonRobinson2014} are inherently large.

Beyond embeddings of concurrency reasoning for discrete systems into proof assistants \cite{Nieto2002, Nipkow1999, Shi2018, ArmstrongIsabelleAlgebra},
\dLCHP can verify parallel hybrid systems synchronizing in shared global time. 
The latter imposes even more complicated binding structures than parallel or hybrid systems alone but \dLCHP's uniform substitution calculus continues to manage them in a modular way.

The recent tool HHLPy \cite{Sheng2023} for hybrid CSP (HCSP) \cite{Jifeng1994} is limited to the sequential fragment.
Unlike extending HHLPy to parallelism,
which would require extensive soundness-critical side conditions and a treatment of the duration calculus,
integrating \dLCHP into \keymaerax \cite{DBLP:conf/cade/FultonMQVP15} boils down to %writing down 
adding a finite list of concrete object level formulas as axioms and only small changes to the uniform substitution process.
In contrast to \dLCHP's compositional parallel systems calculus \citeDLCHP,
HCSP calculi \cite{Liu2010,Wang2012,Guelev2017} are non-compositional \citeDLCHP as they either unroll exponentially many interleavings from the operational semantics \cite{Wang2012,Guelev2017} or can only decompose independent parallel components \cite{Liu2010} causing limited ability to reason about complex systems.
Former HCSP tools \cite{Zou2013, Wang2015} only implement a non-compositional calculus \cite{Liu2010} reinforcing the significance of our approach for managing parallel hybrid systems reasoning.
Other hybrid process algebras defer to model checkers for reasoning \cite{Man2005, Cong2013, Song2005}.
Further discussion of \dLCHP is in \citeDLCHP.

\section{Conclusion} \label{sec:conclusion}

This paper introduced a sound one-pass uniform substitution calculus for the dynamic logic of communicating hybrid programs \dLCHP thereby mastering the significant challenge of developing simple sound proof calculi for parallel hybrid systems with  communication.
Uniform substitution can separate even notoriously complicated binding structures from parallelism with communication in multi-dynamical logics into axioms and their instantiation.
In the case of \dLCHP, this applies to channel access in predicates and the need for local abstraction of subprograms in parallel statements, and it even turns out that uniform substitution can maintain a context-sensitive syntax along the way.
Thanks to uniform substitution, parallel systems reasoning reduces to multiple uses of an asymmetric parallel injection axiom.

Now, with uniform substitution a straightforward implementation of \dLCHP in \keymaerax is only one step away.

\subsubsection{Acknowledgments}
This project was funded in part by the Deutsche For\-schungs-gemeinschaft (DFG) -- \href{https://gepris.dfg.de/gepris/projekt/378803395?context=projekt&task=showDetail&id=378803395&}{378803395} (ConVeY), 
an Alexander von Humboldt Professorship, and
by the AFOSR under grant number FA9550-16-1-0288.

\ifreport\else
\newpage
\renewcommand{\doi}[1]{doi: \href{https://doi.org/#1}{\nolinkurl{#1}}}
\bibliographystyle{splncs04}
\bibliography{platzer,literature}
\fi

\ifreport
\appendix

\section{Details of the Static Semantics}
\label{app:staticSemantics}

This appendix reports proofs of the bound effect property and coincidence lemmas given in \rref{sec:static_semantics}.
Moreover, sound syntactical overapproximations of the static semantics from previous work \citeDLCHP are given and extended to function and predicate symbols, and program constants.

\begin{proof}[of \rref{lem:boundEffect}]
	Let $\computation \in \sem{\alpha}{\inter}$ with $\pstate{w} \neq \bot$.	
	Then $\pstate{v} = \pstate{w}$ on $\TVar$ can be easily proven by induction on $\alpha$ because no program ever changes a trace variable.
	To prove $\pstate{v} = \stconcat{\pstate{w}}{\trace}$ on $\SBV(\alpha)^\complement$, 
	let $\arbitraryVar \not \in \SBV(\alpha)$.
	Then $(\stconcat{\pstate{w}}{\trace})(\arbitraryVar) = \pstate{w}(\arbitraryVar)$ by definition of $\SBV(\cdot)$.
	To prove $\trace \downarrow \SCN(\alpha)^\complement = \epsilon$,
	let $\computation \in \sem{\alpha}{\inter}$ and $\ch{} \not \in \SCN(\alpha)$.
	Then $\trace \downarrow \{\ch{}\} = \epsilon$ by definition of $\SCN(\cdot)$.
	Since this holds for all $\ch{} \in \SCN(\alpha)^\complement$,
	we obtain $\trace \downarrow \SCN(\alpha)^\complement = \epsilon$.

	Suppose that $\SBV(\alpha)$ and $\SCN(\alpha)$ are not the smallest sets with the bound effect property but $\varset \subseteq \V$ and $\cset \subseteq \Chan$ with $\varset \not\supseteq \SBV(\alpha)$ or $\cset \not\supseteq \SCN(\alpha)$ have it, too.
	Then there is $\arbitraryVar \in \SBV(\alpha)$ with $\arbitraryVar \not \in \varset$ or $\ch{} \in \SCN(\alpha)$ with $\ch{} \not \in \cset$.
	If $\arbitraryVar \in \SBV(\alpha)$ and $\arbitraryVar \not\in \varset$,
	then $\inter$ and $\computation \in \sem{\alpha}{\inter}$ exist such that $\pstate{v}(\arbitraryVar) \neq (\stconcat{\pstate{w}}{\trace})(\arbitraryVar)$.
	But then $\varset$ does not have the bound effect property as $\arbitraryVar$ changed by $\alpha$.
	If $\ch{} \in \SCN(\alpha)$ and $\ch{} \not\in \cset$,
	then $\inter$ and $\computation \in \sem{\alpha}{\inter}$ exist such that $\trace \downarrow \{\ch{}\} \neq \epsilon$.
	But then $\trace \downarrow \cset^\complement \neq \epsilon$ such that $\cset$ does not have the bound effect property.
	\qedhere
\end{proof}

The following lemma prepares the proof of the communication-aware coincidence property (\rref{lem:termCoincidence}) for terms and formulas:

\newcommand{\chanTraces}[2]{%
	T_{\ifempty{#1}{\cset}{#1}}%
	^{\ifempty{#2}{\trace}{#2}}}

\begin{lemma} \label{lem:chan_difference}
	Let $\chanTraces{}{} = \{ \trace[pre] \mid \trace[pre] \downarrow \cset = \trace \downarrow \cset \}$ for $\cset \subseteq \Chan$ and $\cset_0 = \cset \cup \{ \ch{} \}$.
	Then for all $\trace[pre], \trace \in \traces$,
	if $\trace[pre] \in \chanTraces{}{}$,
	then $\trace[ppre] \in \chanTraces{\cset_0}{}$ exists such that $\trace[ppre] \downarrow \{\ch{}\}^\complement = \trace[pre] \downarrow \{\ch{}\}^\complement$.
\end{lemma}
\begin{proof}
	For $\rawtrace = \comItem{\ch{}, a, \duration}$, we define $\Chan(\rawtrace) = \ch{}$.
	Moreover, we identify the item $\ch{}$ with the singleton $\{\ch{}\}$.
	Now, the proof is by induction on the structure of $\trace[pre]$:
	\begin{enumerate}
		\item $\trace[pre] = \epsilon$, then let $\trace[pre] \in \chanTraces{}{}$.
		Since $\trace[pre] \downarrow \cset = \epsilon$,
		we obtain $\trace \downarrow \cset = \epsilon$.
		We define $\trace[ppre] = \trace \downarrow \ch{}$.
		Now, $\trace[ppre] \downarrow \cset_0 = \trace \downarrow (\cset_0 \cap \ch{}) = \trace \downarrow \ch{}$,
		which equals $\trace \downarrow \cset_0$ because $\trace \downarrow \cset = \epsilon$.
		Hence, $\trace[ppre] \in \chanTraces{\cset_0}{\trace}$.
		Finally, $\trace[ppre] \downarrow \ch{}^\complement = \trace \downarrow (\ch{} \cap \ch{}^\complement) = \epsilon = \trace[pre] \downarrow \ch{}^\complement$.

		\item $\trace[pre] = \rawtrace \cdot \trace[pre]_0$ with $\semLen{\rawtrace} = 1$, then let $\trace[pre] \in \chanTraces{}{}$.
		Hence, $\trace[pre] \downarrow \cset = \trace \downarrow \cset$.
		
		If $\Chan(\rawtrace) \in \cset$, then $\trace[pre] \downarrow \cset = \rawtrace \cdot \trace[pre]_0 \downarrow \cset$ and $\trace = \trace_1 \cdot \rawtrace \cdot \trace_2$ for some $\trace_1, \trace_2$ with $\trace_1 \downarrow \cset = \epsilon$.
		Hence, $\trace[pre]_0 \downarrow \cset = \trace_2 \downarrow \cset$ such that $\trace'_0 \in \chanTraces{}{\trace_2}$.
		By IH, $\trace[ppre]_0 \in \chanTraces{\cset_0}{\trace_2}$ exists such that $\trace[ppre]_0 \downarrow \ch{}^\complement = \trace[pre]_0 \downarrow \ch{}^\complement$.
		We define $\trace[ppre] = \trace_1 \downarrow \ch{} \cdot \rawtrace \cdot \trace[ppre]_0$.
		Since $\trace[ppre]_0 \in \chanTraces{\cset_0}{\trace_2}$,
		we have $\trace[ppre]_0 \downarrow \cset_0 = \trace_2 \downarrow \cset_0$.
		Moreover, $\trace_1 \downarrow \cset = \epsilon$ implies $\trace_1 \downarrow \ch{} = \trace_1 \downarrow \cset_0$.
		Therefore, $\trace[ppre] \downarrow \cset_0 
		= \trace_1 \downarrow (\ch{} \cap \cset_0) \cdot \rawtrace \downarrow \cset_0 \cdot \trace[ppre]_0 \downarrow \cset_0
		= \trace_1 \downarrow \cset_0 \cdot \rawtrace \downarrow \cset_0 \cdot \trace_2 \downarrow \cset_0
		= \trace \downarrow \cset_0$.
		Hence, $\trace[ppre] \in \chanTraces{\cset_0}{}$.
		Moreover, $\trace[ppre] \downarrow \ch{}^\complement 
		= \trace_1 \downarrow (\ch{} \cap \ch{}^\complement) \cdot \rawtrace \downarrow \ch{}^\complement \cdot \trace[ppre]_0 \downarrow \ch{}^\complement
		= \rawtrace \downarrow \ch{}^\complement \cdot \trace[ppre]_0 \downarrow \ch{}^\complement
		= \trace[pre] \downarrow \ch{}^\complement$.

		Otherwise, if $\Chan(\rawtrace) \not \in \cset$,
		then $\trace[pre]_0 \downarrow \cset = \trace[pre] \downarrow \cset = \trace \downarrow \cset$ such that $\trace[pre]_0 \in \chanTraces{}{}$.
		By IH, $\trace[ppre]_0 \in \chanTraces{\cset_0}{}$ exists such that $\trace[ppre]_0 \downarrow \ch{}^\complement = \trace[pre]_0 \downarrow \ch{}^\complement$.
		Now, we define $\trace[ppre] = \rawtrace \downarrow \ch{}^\complement \cdot \trace[ppre]_0$.
		Since $\trace[ppre] \downarrow \cset_0
		= \rawtrace \downarrow (\ch{}^\complement \cap \cset_0) \cdot \trace[ppre]_0 \downarrow \cset_0
		= \rawtrace \downarrow \cset \cdot \trace[ppre]_0 \downarrow \cset_0
		= \trace[ppre]_0 \downarrow \cset_0$,
		we have $\trace[ppre] \in \chanTraces{\cset_0}{}$.
		Finally, $\trace[ppre] \downarrow \ch{}^\complement
		= \rawtrace \downarrow (\ch{}^\complement \cap \ch{}^\complement) \cdot \trace[ppre]_0 \downarrow \ch{}^\complement
		= \rawtrace \downarrow \ch{}^\complement \cdot \trace[ppre]_0 \downarrow \ch{}^\complement
		= \rawtrace \downarrow \ch{}^\complement \cdot \trace[pre]_0 \downarrow \ch{}^\complement
		= \trace[pre] \downarrow \ch{}^\complement$.
	\end{enumerate}
	\qedhere
\end{proof}

\begin{proof}[of \rref{lem:termCoincidence}]
	The proof generalizes the coincidence property proofs of $\dL$ \cite[Lemma 10]{DBLP:journals/jar/Platzer17} to communication-aware coincidence.
	Since $\sem{\expr}{\lstate[opt]{v}} = \sem{\expr}{\lstate[alt]{v}}$ if $\inter = \inter[alt]$ on $\sigof{\expr}$ by an induction on the structure of $\expr$, 
	it suffices to prove that $\sem{\expr}{\lstate{v}} = \sem{\expr}{\lstate[opt]{v}}$ for all $\inter$.
	Let $S_{\varset, \cset}$ be a set of states between $\pstate{v}$ and $\pstate[alt]{v}$ according to variables $\varset \subseteq \V$ and channels $\cset \subseteq \Chan$ as follows:
	\begin{equation*}
		S_{\varset, \cset} = \{ \pstate[pre]{v} \mid 
			\pstate[pre]{v} \downarrow \cset = \pstate{v} \downarrow \cset \text{ on } \varset \text{ and }
			\pstate[pre]{v} \downarrow \cset^\complement = \pstate[alt]{v} \downarrow \cset^\complement \text{ on } \varset^\complement \}
	\end{equation*}

	Fix an interpretation $\inter$
	and prove $\sem{\expr}{\lstate[pre]{v}} = \sem{\expr}{\lstate[opt]{v}}$
	for all $\varset \subseteq \SFV(\expr)^\complement$, and $\cset \subseteq \SCN(\expr)^\complement$, and $\pstate[pre]{v} \in S_{\varset, \cset}$.
	Therefore, we increase the sets $\varset$ and $\cset$ starting from $\emptyset$ for both, 
	where $\pstate[pre]{v}$ may differ from $\pstate[alt]{v}$, 
	by lexicographic induction on $\varset$ and~$\cset$ till we reach $\varset = \SFV(\expr)^\complement$ and $\cset = \SCN(\expr)^\complement$.
	This suffices for $\sem{\expr}{\lstate{v}} = \sem{\expr}{\lstate[opt]{v}}$ because $\pstate{v} \in S_{\SFV(\expr)^\complement, \SCN(\expr)^\complement}$ by the premise that $\pstate{v} \downarrow \SCN(\expr) = \pstate[alt]{v} \downarrow \SCN(\expr)$ on $\SFV(\expr)$.

	\begin{enumerate}
		\item $\varset = \emptyset$ and $\cset = \emptyset$, then $S_{\varset, \cset} = \{ \pstate[alt]{v} \}$ such that $\sem{\expr}{\lstate[pre]{v}} = \sem{\expr}{\lstate[opt]{v}}$ for all $\pstate[pre]{v} \in S_{\varset, \cset}$ holds trivially.
		
		\item $\varset = \varset_0 \cup \{ \arbitraryVar \}$ with $\arbitraryVar \not \in \varset_0$ and $\arbitraryVar \not \in \SFV(\expr)$,
		then let $\pstate[pre]{v} \in S_{\varset, \cset}$.
		We define $\pstate[ppre]{v} = \pstate[pre]{v} {}\subs{\arbitraryVar}{\pstate[alt]{v}(\arbitraryVar)}$.
		By $\pstate[pre]{v} \in S_{\varset, \cset}$,
		we obtain $\pstate[ppre]{v} \downarrow \cset = \pstate{v} \downarrow \cset$ on $\varset_0$ since $\pstate[pre]{v} \downarrow \cset = \pstate{v} \downarrow \cset$ on $\varset$,
		and $\pstate[ppre]{v} \downarrow \cset^\complement = \pstate[alt]{v} \downarrow \cset^\complement$ on $\varset^\complement$ since $\pstate[pre]{v} \downarrow \cset^\complement = \pstate[alt]{v} \downarrow \cset^\complement$ on $\varset^\complement$.
		Moreover, $\pstate[ppre]{v}(\arbitraryVar) = \pstate[alt]{v}(\arbitraryVar)$ if $\arbitraryVar \in \RVar \cup \NVar$,
		and $\pstate[ppre]{v}(\arbitraryVar) \downarrow \cset^\complement = \pstate[alt]{v}(\historyVar) \downarrow \cset^\complement$ if $\arbitraryVar \in \TVar$, respectively, by definition of $\pstate[ppre]{v}$.
		Therefore, $\pstate[ppre]{v} \downarrow \cset^\complement = \pstate[alt]{v} \downarrow \cset^\complement$ on $\varset^\complement \cup \{ \arbitraryVar \} = \varset_0^\complement$ such that $\pstate[ppre]{v} \in S_{\varset_0, \cset}$.
		
		By definition of $\SFV(\expr)$, 
		we obtain $\sem{\expr}{\lstate[pre]{v}} = \sem{\expr}{\lstate[ppre]{v}}$ because $\pstate[ppre]{v} = \pstate[pre]{v}$ on $\{ \arbitraryVar \}^\complement$ but $\arbitraryVar \not \in \SFV(\expr)$.
		Finally, $\sem{\expr}{\lstate[pre]{v}} = \sem{\expr}{\lstate[ppre]{v}} \IH{=} \sem{\expr}{\lstate[opt]{v}}$ by IH using $\pstate[ppre]{v} \in S_{\varset_0, \cset}$.

		\item $\varset = \varset_0$ and $\cset = \cset_0 \cup \{ \ch{} \}$ with $\ch{} \not \in \cset_0$ and $\ch{} \not \in \SCN(\expr)$,
		then let $\pstate[pre]{v} \in S_{\varset, \cset}$.
		Consider $\arbitraryVar \in \varset^\complement \cap \TVar$.
		Then $\pstate[pre]{v}(\arbitraryVar) \downarrow \cset^\complement = \pstate[alt]{v}(\arbitraryVar) \downarrow \cset^\complement$.
		Moreover, $\cset_0^\complement = \cset^\complement \cup \{ \ch{} \}$.
		Therefore, by \rref{lem:chan_difference},
		$\trace[ppre]_\arbitraryVar$ with $\trace[ppre]_\arbitraryVar \downarrow \cset_0^\complement = \pstate[alt]{v}(\arbitraryVar) \downarrow \cset_0^\complement$ exists such that $\trace[ppre]_\arbitraryVar \downarrow \{ \ch{} \}^\complement = \pstate[pre]{v}(\arbitraryVar) \downarrow \{ \ch{} \}^\complement$.

		Using one $\trace[ppre]_\arbitraryVar$ for each $\arbitraryVar \in \varset^\complement \cap \TVar$,
		we define a state $\pstate[ppre]{v}$ as follows:
		\begin{equation*}
			\pstate[ppre]{v}(\arbitraryVar) = \begin{cases}
				\pstate[pre]{v}(\arbitraryVar) & \text{for } \arbitraryVar \in \RVar \cup \NVar \\
				\pstate[pre]{v}(\arbitraryVar) & \text{for } \arbitraryVar \in \varset \cap \TVar \\
				\trace[ppre]_\arbitraryVar & \text{for } \arbitraryVar \in \varset^\complement \cap \TVar
			\end{cases}
		\end{equation*}

		For $\arbitraryVar \in \varset \cap (\RealNatVar)$,
		we have $\pstate[ppre]{v}(\arbitraryVar) = \pstate[pre]{v}(\arbitraryVar) = \pstate{v}(\arbitraryVar)$.
		Moreover, for $\arbitraryVar \in \varset^\complement \cap (\RealNatVar)$,
		we have $\pstate[ppre]{v}(\arbitraryVar) = \pstate[pre]{v}(\arbitraryVar) = \pstate[alt]{v}(\arbitraryVar)$.
		Further, for $\arbitraryVar \in \varset \cap \TVar$,
		we have $\pstate[ppre]{v}(\arbitraryVar) \downarrow \cset_0 = \pstate[pre]{v}(\arbitraryVar) \downarrow \cset_0 = \pstate{v}(\arbitraryVar) \downarrow \cset_0$ because $\pstate[pre]{v}(\arbitraryVar) \downarrow \cset = \pstate{v}(\arbitraryVar) \downarrow \cset$ as $\pstate[pre]{v} \in S_{\varset, \cset}$.
		Finally, for $\arbitraryVar \in \varset^\complement \cap \TVar$,
		we have $\pstate[ppre]{v}(\arbitraryVar) \downarrow \cset_0^\complement
		= \trace[ppre]_\arbitraryVar \downarrow \cset_0^\complement
		= \pstate[alt]{v}(\arbitraryVar) \downarrow \cset_0^\complement$ due to \rref{lem:chan_difference}.
		Therefore, $\pstate[ppre]{v} \in S_{\varset, \cset_0}$ such that $\sem{\expr}{\lstate[ppre]{v}} = \sem{\expr}{\lstate[opt]{v}}$ by IH.

		Observe that $\pstate[ppre]{v} = \pstate[pre]{v}$ on $\RVar \cup \NVar$.
		Moreover, for $\arbitraryVar \in \varset \cap \TVar$,
		we have $\pstate[ppre]{v}(\arbitraryVar) \downarrow \{ \ch{} \}^\complement 
		= \pstate[pre]{v}(\arbitraryVar) \downarrow \{ \ch{} \}^\complement$.
		Finally, for $\arbitraryVar \in \varset^\complement \cap \TVar$,
		we have $\pstate[ppre]{v}(\arbitraryVar) \downarrow \{ \ch{} \}^\complement 
		= \trace[ppre]_\arbitraryVar \downarrow \{ \ch{} \}^\complement
		= \pstate[pre]{v}(\arbitraryVar) \downarrow \{ \ch{} \}^\complement$
		due to \rref{lem:chan_difference}.
		Overall, $\pstate[ppre]{v} \downarrow \{ \ch{} \}^\complement = \pstate[pre]{v} \downarrow \{ \ch{} \}^\complement$.
		Since $\ch{} \not \in \SCN(\expr)$,
		we obtain $\sem{\expr}{\lstate[ppre]{v}} = \sem{\expr}{\lstate[pre]{v}}$ from the definition of $\SCN(\expr)$.

		Finally, $\sem{\expr}{\lstate[pre]{v}} = \sem{\expr}{\lstate[ppre]{v}} = \sem{\expr}{\lstate[opt]{v}}$.
	\end{enumerate}

	Suppose that $\SFV(\expr)$ and $\SCN(\expr)$ are not the smallest sets with the coincidence property 
	but $\varset \subseteq \V$ and $\cset \subseteq \Chan$ with
	$\varset \not\supseteq \SFV(\expr)$ or $\cset \not\supseteq \SCN(\expr)$ have the coincidence property, too.
	Then there is $\arbitraryVar \in \SFV(\expr)$ with $\arbitraryVar \not \in \varset$ or $\ch{} \in \SCN(\expr)$ with $\ch{} \not \in \cset$.
	If $\arbitraryVar \in \SFV(\expr)$ and $\arbitraryVar \not\in \varset$,
	then by definition of $\SFV(\expr)$,
	states $\pstate{v}, \pstate[alt]{v}$ with $\pstate{v} = \pstate[alt]{v}$ on $\{ z \}^\complement$ exist such that $\sem{\expr}{\lstate{v}} \neq \sem{\expr}{\lstate[opt]{v}}$.
	But then $(\varset, \cset)$ does not have the coincidence property because $\pstate{v} \downarrow \cset = \pstate[alt]{v} \downarrow \cset$ on $\varset$ but $\sem{\expr}{\lstate{v}} \neq \sem{\expr}{\lstate[opt]{v}}$.
	If $\ch{} \in \SCN(\expr)$ and $\ch{} \not\in \cset$,
	then by definition of $\SCN(\expr)$,
	states $\pstate{v}, \pstate[alt]{v}$ with $\pstate{v} \downarrow \{\ch{}\}^\complement = \pstate[alt]{v} \downarrow \{\ch{}\}^\complement$ exist such that $\sem{\expr}{\lstate{v}} = \sem{\expr}{\lstate[opt]{v}}$.
	But then $(\varset, \cset)$ does not have the coincidence property because $\pstate{v} \downarrow \cset = \pstate[alt]{v} \downarrow \cset$ but $\sem{\expr}{\lstate{v}} \neq \sem{\expr}{\lstate[opt]{v}}$.
	\qedhere
\end{proof}

\begin{proof}[of \rref{lem:programCoincidence}]
	\newcommand{\stateS}[1]{\pstate[pre]{v}_{#1}}
	The proof generalizes the coincidence property proofs of $\dL$ \cite[Lemma 12]{DBLP:journals/jar/Platzer17} to a coincidence property for CHPs.
	Let $\stateS{\varset[alt]}$ be the state between $\pstate{v}$ and $\pstate[alt]{v}$ according to the variables~$S$,
	\iest $\stateS{\varset[alt]} = \pstate[alt]{v}$ on~$S$ and $\stateS{\varset[alt]} = \pstate{v}$ on~$S^\complement$.
	Then we prove by induction on $S \subseteq \SFV(\alpha)^\complement$ that for all $\stateS{\varset[alt]}$ a computation $(\stateS{\varset[alt]}, \trace[pre], \pstate[pre]{w}) \in \sem{\alpha}{\inter}$ exists such that $\pstate[pre]{w} = \pstate{w}$ on $S^\complement$, and $\trace[pre] = \trace$, and ($\pstate[pre]{w} = \bot$ iff $\pstate{w} = \bot$).
	This suffices to prove the lemma because first, 
	$\varset^\complement \subseteq \SFV(\alpha)^\complement$ such that $(\pstate[pre]{v}_{\varset^\complement}, \trace[pre], \pstate[pre]{w}) \in \sem{\alpha}{\inter}$,
	where $\pstate[pre]{v}_{\varset^\complement} = \pstate[alt]{v}$ on $\varset^\complement$ and $\pstate[pre]{v}_{\varset^\complement} = \pstate{v} = \pstate[alt]{v}$ on $(\varset^\complement)^\complement = \varset$, and $\trace[pre] = \trace$, and $\pstate[pre]{w} = \pstate{w}$ on $(\varset^\complement)^\complement = \varset$.
	Second, $\sem{\alpha}{\inter} = \sem{\alpha}{\inter[alt]}$ by an induction on $\alpha$.
	Now, we proceed with the induction on $\varset[alt]$:

	\begin{enumerate}
		\item $\varset[alt] = \emptyset$, then $\stateS{\varset[alt]} = \pstate{v}$.
		If we define $\trace[pre] = \trace$ and $\pstate[pre]{w} = \pstate{w}$,
		then $(\stateS{\varset[alt]}, \trace[pre], \pstate[pre]{w}) \in \sem{\alpha}{\inter}$ and fulfills the conditions.

		\item $\varset[alt] = \varset[alt]_0 \cup \{ \arbitraryVar \}$ with $\arbitraryVar \not\in \varset[alt]_0$ and $\arbitraryVar \not\in \SFV(\alpha)$,
		then let $\stateS{\varset[alt]}$ be between $\pstate{v}$ and $\pstate[alt]{v}$ according to $S$.
		Moreover, let $\pstate[ppre]{v} = (\stateS{\varset[alt]}) \subs{\arbitraryVar}{\pstate{v}(\arbitraryVar)}$.
		Since $\pstate[ppre]{v} = \pstate[alt]{v}$ on $\varset[alt]_0$ and $\pstate[ppre]{v} = \pstate{v}$ on $\varset[alt]_0^\complement$,
		we have $\pstate[ppre]{v} = \stateS{\varset[alt]_0}$
		such that by IH, $(\pstate[ppre]{v}, \trace[ppre], \pstate[ppre]{w})$ exists with $\pstate[ppre]{w} = \pstate{w}$ on $\varset[alt]_0^\complement$, and $\trace[ppre] = \trace$, and ($\pstate[ppre]{w} = \bot$ iff $\pstate{w} = \bot$).
		Since $\pstate[ppre]{v} = \stateS{\varset[alt]}$ on $\{\arbitraryVar\}^\complement$ but $\arbitraryVar \not\in \SFV(\alpha)$,
		there is $(\stateS{\varset[alt]}, \trace[pre], \pstate[pre]{w}) \in \sem{\alpha}{\inter}$ by the definition of $\SFV(\alpha)$ such that $\trace[pre] = \trace[ppre]$ and $\pstate[pre]{w} = \pstate[ppre]{w}$ on $\{\arbitraryVar\}^\complement$,
		which includes that ($\pstate[pre]{w} = \bot$ iff $\pstate[ppre]{w} = \bot$).
		Thus, $\pstate[pre]{w} = \pstate[ppre]{w} = \pstate{w}$ on $\{\arbitraryVar\}^\complement \cap \varset[alt]_0^\complement = S^\complement$.
		Moreover, $\trace[pre] = \trace[ppre] = \trace$ and ($\pstate[pre]{w} = \bot$ iff $\pstate{w} = \bot$).
	\end{enumerate}

	Suppose that $\SFV(\alpha)$ is not the smallest set with the coincidence property but $\varset \not\supseteq \SFV(\alpha)$ has the property, too.
	Then there is $\arbitraryVar \in \SFV(\alpha)$ with $\arbitraryVar \not\in \varset$.
	By definition of $\SFV(\alpha)$, interpretation $\inter$, and $\computation \in \sem{\alpha}{\inter}$, and $\pstate[alt]{v}$ exist such that $\pstate{v} = \pstate[alt]{v}$ on $\{\arbitraryVar\}^\complement$ and $\computation \in \sem{\alpha}{\inter}$ but there are no $\trace[alt], \pstate[alt]{w}$ such that $\computation[alt] \in \sem{\alpha}{\inter}$, and $\pstate[alt]{w} = \pstate{w}$ on $\{\arbitraryVar\}^\complement$, and $\trace[alt] = \trace$, and ($\pstate[alt]{w} = \bot$ iff $\pstate{w} = \bot$).
	But then $\varset$ does not have the coincidence property
	because $\pstate{v} = \pstate[alt]{v}$ on $\{\arbitraryVar\}^\complement \supseteq \varset$ but no $\trace[alt]$ and $\pstate[alt]{w}$ exist such that $\computation[alt] \in \sem{\alpha}{\inter}$, and $\pstate[alt]{w} = \pstate{w}$ on $\varset$, and $\trace[alt] = \trace$, and ($\pstate[alt]{w} = \bot$ iff $\pstate{w} = \bot$).
	\qedhere
\end{proof}

The static semantics of \rref{def:staticSemantics} is not computable \cite{Rice1953}.
\rref{def:boundVariables}--\ref{def:freeFormulaParameters} adapt sound overapproximations of the static semantics computed from the syntactical structure \citeDLCHP to \dLCHP.
The definitions add the cases for function and predicate symbols, and program constants, which were only introduced in this paper.

Crucially, the bound effect property and the coincidence lemmas apply for overapproximations of the static semantics as well.
Thus, the overapproximations can be soundly used in an implementation of uniform substitution.

\begin{definition}[Bound variables] \label{def:boundVariables}
	The set of (syntactically) \emph{bound variables} $\BV(\alpha)$ of a program $\alpha$ is inductively defined 
	as follows,
	where $\{\varvec\} = \cup_{\arbitraryVar\in\varvec} \{\arbitraryVar\}$:
	\begin{small}
		\begin{align*}
			\BV(\rpconst{}{}) & = \{ \varvec \} \\
			\BV(x \ceq \expr) = \BV(x \ceq *) & = \{ x \} \\
			\BV(\evolution{}{}) & = \odeBoundVars \\
			\BV(\test{}) & = \emptyset \\
			\BV(\send{}{}{}) & = \{ \historyVar \} \\
			\BV(\receive{}{}{}) & = \{ \historyVar, x \} \\
			\BV(\alpha \cup \beta) = \BV(\alpha \seq \beta) = \BV(\alpha \parOp \beta) & = \BV(\alpha) \cup \BV(\beta) \\
			\BV(\repetition{\alpha}) & = \BV(\alpha)
		\end{align*}%
	\end{small}
\end{definition}

\begin{definition}[Written channels] \label{def:writtenChannels}
	The set of (syntactically) \emph{written channels} $\CN(\alpha)$ of a program $\alpha$ is inductively defined 
	as follows:
	\begin{small}
		\begin{align*}
			\CN(\rpconst{}{}) & = \cset \\
			\CN(x \ceq \expr) = \CN(x \ceq *) = \CN(\evolution{}{}) = \CN(\test{}) & = \emptyset \\
			\CN(\send{}{}{}) = \CN(\receive{}{}{}) & = \{ \ch{} \} \\
			\CN(\alpha \cup \beta) = \CN(\alpha \seq \beta) = \CN(\alpha \parOp \beta) & = \CN(\alpha) \cup \CN(\beta) \\
			\CN(\repetition{\alpha}) & = \CN(\alpha)
		\end{align*}%
	\end{small}
\end{definition}

\begin{definition}[Parameters of terms]
	\label{def:freeTermParameters}
	The sets of (syntactically) free variables $\FV(\expr)$ and (syntactically) accessed channels $\CN(\expr)$ of a term $\expr$ are inductively defined below,
	where~$\fsymb[builtin]$ is any built-in function symbol of fixed interpretation (see \rref{def:syntax_terms}), \eg~$\usarg + \usarg$, except for projection $\usarg \downarrow \usarg$.
	Moreover, let $\FV(\exprvec) = \cup_{\expr\in\exprvec} \FV(\expr)$ and $\CN(\exprvec) = \cup_{\expr\in\exprvec} \CN(\expr)$.
	\par\begin{small}
		\noindent\begin{minipage}{.5\textwidth}
			\begin{align*}
				\FV(\fsymb(\cset, \exprvec)) & = \FV(\exprvec) \\
				\FV(\arbitraryVar) & = \{ \arbitraryVar \} & \text{for } \arbitraryVar\in\V \\
				\\
				\FV(\fsymb[builtin](\exprvec)) & = \FV(\exprvec) \\
				\FV(\te \downarrow \cset) & = \FV(\te)
			\end{align*}
		\end{minipage}
		\begin{minipage}{.5\textwidth}
			\begin{align*}
				\CN(\fsymb(\cset, \exprvec)) & = \cset \cap \CN(\exprvec) \\
				\CN(\historyVar) & = \Chan & \text{for } \historyVar\in\TVar \\
				\CN(\arbitraryVar) & = \emptyset & \text{for } \arbitraryVar\not\in\TVar \\
				\CN(\fsymb[builtin](\exprvec)) & = \CN(\exprvec) \\
				\CN(\te \downarrow \cset) & = \cset \cap \CN(\te)
			\end{align*}
		\end{minipage}
	\end{small}
\end{definition}

The must-bound variables $\MBV(\alpha)$ (\rref{def:mustBoundVariables}) are those variables that are bound on all execution paths of a program $\alpha$.
In contrast to $\BV(\alpha)$, they can be soundly used \cite{DBLP:journals/jar/Platzer17} in the cases for $\FV(\alpha \seq \beta)$ in \rref{def:freeProgramParameters} and $\FV([ \alpha ] \psi)$ in~\rref{def:freeFormulaParameters}.

\begin{definition}[Must-bound variables] \label{def:mustBoundVariables}
	The set of \emph{must-bound variables} $\MBV(\alpha)$ of a program $\alpha$ is inductively defined as follows:
	\begin{small}
		\begin{align*}
			\MBV(\rpconst{}{}) & = \emptyset \\
			\MBV(\alpha) & = \BV(\alpha) \quad\text{for atomic CHPs $\alpha$ except for program constants} \\
			\MBV(\alpha \cup \beta) & = \MBV(\alpha) \cap \MBV(\beta) \\
			\MBV(\alpha \seq \beta) = \MBV(\alpha \parOp \beta) & = \MBV(\alpha) \cup \MBV(\beta) \\
			\MBV(\repetition{\alpha}) & = \emptyset
		\end{align*}
	\end{small}
\end{definition}

\begin{definition}[Free variables of programs] \label{def:freeProgramParameters}
	The set of (syntactically) free variables $\FV(\alpha)$ of a program $\alpha$ is inductively defined 
	%in \rref{fig:freeProgramParameters}.
	as follows:
	\begin{small}
		\begin{align*}
			\FV(\rpconst{}{}) & = \RVar \cup \TVar \\
			\FV(x \ceq \rp) & = \FV(\rp) \\
			\FV(x \ceq *) & = \emptyset \\
			\FV(\test{}) & = \FV(\chi) \\
			\FV(\evolution{}{}) & = \odeFreeVars \cup \FV(\rp) \cup \FV(\chi) \\
			\FV(\send{}{}{}) & = \comFreeVars \cup \FV(\rp) \\
			\FV(\receive{}{}{}) & = \comFreeVars \\
			\FV(\alpha \seq \beta) & = \FV(\alpha) \cup (\FV(\beta) \setminus \MBV(\alpha)) \\
			\FV(\alpha \cup \beta) = \FV(\alpha \parOp \beta) & = \FV(\alpha) \cup \FV(\beta) \\
			\FV(\repetition{\alpha}) & = \FV(\alpha)
		\end{align*}
	\end{small}
\end{definition}

\begin{definition}[Parameters of formulas] \label{def:freeFormulaParameters}
	The sets of (syntactically) free variables $\FV(\phi)$ and (syntactically) accessed channels $\CN(\phi)$ of a formula $\phi$ are inductively defined as follows,
	where $\FV(\exprvec) = \cup_{\expr\in\exprvec} \FV(\expr)$ and $\CN(\exprvec) = \cup_{\expr\in\exprvec} \CN(\expr)$:

	\renewcommand{\rightleftharpoons}{\wedge}
	\renewcommand{\quantor}[1]{\fa{#1}}

	\begin{small}
		\noindent\begin{minipage}{.5\textwidth}
			\begin{align*}
				\FV(\psymb(\cset, \exprvec)) & = \FV(\exprvec) \\
				\FV(\expr_1 \sim \expr_2) & = \FV(\expr_1) \cup \FV(\expr_2) \\
				\FV(\neg \varphi) & = \FV(\varphi) \\
				\FV(\varphi \rightleftharpoons \psi) & = \FV(\varphi) \cup \FV(\psi) \\
				\FV(\quantor{\arbitraryVar} \varphi) & = \FV(\varphi) \setminus \{ \arbitraryVar \} \\
				\FV([ \alpha ] \psi) & = \FV(\alpha) \cup (\FV(\psi) \setminus \MBV(\alpha)) \\
				\FV([ \alpha ] \ac \psi) & = \FV([ \alpha ] \psi) \cup \FV(\A) \cup \FV(\C)
			\end{align*}
		\end{minipage}
		\begin{minipage}{.5\textwidth}
			\begin{align*}
				\CN(\psymb(\cset, \exprvec)) & = \cset \cap \CN(\exprvec) \\
				\CN(\expr_1 \sim \expr_2) & = \CN(\expr_1) \cup \CN(\expr_2) \\
				\CN(\neg \varphi) & = \CN(\varphi) \\
				\CN(\varphi \rightleftharpoons \psi) & = \CN(\varphi) \cup \CN(\psi) \\
				\CN(\quantor{\arbitraryVar} \varphi) & = \CN(\varphi) \\
				\CN( [ \alpha ] \psi) & = \CN(\psi) \\
				\CN([ \alpha ] \ac \psi) & = \CN(\A) \cup \CN(\C) \cup \CN(\psi)
			\end{align*}
		\end{minipage}
	\end{small}
\end{definition}

\section{Soundness of Uniform Substitution}
\label{app:soundness}

This appendix reports the soundness proof of uniform substitution for \dLCHP (\rref{thm:soundness}) and a proof that uniform substitution preserves the syntactic well-formedness of formulas (\rref{prop:unisubs_well_formed}).
Moreover, \rref{thm:usrulesound} given in this section enables the instantiation of axiomatic proof rules by uniform substitution.

\newcommand{\tabooWithCtx}{\jointTaboo_0}
\newcommand{\SBP}[1]{\mathsf{B\kern-1.5pt P}(#1)}

\begin{proof}[of \rref{lem:unisubs_correct_bound}]
	The proof is by induction on the structure of program $\alpha$ and generalizes the corresponding proof for \dGL \cite[Lemma 13]{DBLP:conf/cade/Platzer19},
	where $\tabooWithCtx$ is short for $\jointTaboo\cup\parallelCtx$ and $\SBP{\cdot} = \SBV(\cdot) \cup \SCN(\cdot)$ denotes all bound parameters:
	\begin{enumerate}
		\item $\alpha \equiv \rpconst{}{}$, then $\jointOut = \jointTaboo \cup \SBP{\usubs \pconst} = \jointTaboo \cup \SBP{\usInOutOp{}{}{\rpconst{}{}}}$. 
		
		\item $\alpha \equiv x \ceq \rp$, then $\jointOut = \jointTaboo \cup \{ x \}$ and $\{ x \} \supseteq \SBV(x \ceq \usTabooOp{\tabooWithCtx}{\rp}) = \SBV(\usInOutOp{}{}{\alpha})$.
		Moreover, $\emptyset = \SCN(\usInOutOp{}{}{\alpha})$.
		Hence, $\jointOut \supseteq \jointTaboo \cup \SBP{\usInOutOp{}{}{\alpha}}$.
		
		\item $\alpha \equiv x \ceq *$, then $\jointOut = \jointTaboo \cup \{ x \}$ and $\{ x \} = \SBV(\alpha) = \SBV(\usInOutOp{}{}{\alpha})$. 
		Moreover, $\emptyset = \SCN(\usInOutOp{}{}{\alpha})$.
		Hence, $\jointOut \supseteq \jointTaboo \cup \SBP{\usInOutOp{}{}{\alpha}}$.
		
		\item $\alpha \equiv \evolution{}{}$, then $\jointOut = \jointTaboo \cup \odeBoundVars$.
		Moreover, $\odeBoundVars 
		\supseteq \SBV(\evolution{x' = \usTabooOp{\tabooWithCtx}{\rp}}{\usTabooOp{\tabooWithCtx}{\chi}}) 
		= \SBV(\usInOutOp{}{}{\alpha})$ 
		and $\emptyset = \SCN(\usInOutOp{}{}{\alpha})$.
		Thus, $\jointOut \supseteq \jointTaboo \cup \SBP{\usInOutOp{}{}{\alpha}}$.

		\item $\alpha \equiv \test{}$, then $\jointOut = \jointTaboo$ and $\emptyset = \SBV(\test{\usTabooOp{\tabooWithCtx}{\chi}}) = \SBV(\usInOutOp{}{}{\test{}})$.
		Moreover, $\emptyset = \SCN(\usInOutOp{}{}{\test{}})$.
		Thus, $\jointOut \supseteq \jointTaboo \cup \SBP{\usInOutOp{}{}{\test{}}}$.

		\item $\alpha \equiv \beta \seq \gamma$, then $\usInOutOp{}{}{\beta \seq \gamma} \equiv \usInOutOp{}{\jointOut_0}{\beta} \seq \usInOutOp{\jointOut_0, \parallelCtx}{}{\gamma}$.
		By IH, $\jointOut_0 \IH{\supseteq} \jointTaboo \cup \SBP{\usInOutOp{}{\jointOut_0}{\beta}}$ 
		and $\jointOut \IH{\supseteq} \jointOut_0 \cup \SBP{\usInOutOp{\jointOut_0,\parallelCtx}{}{\gamma}}$.
		Thus, $\jointOut \supseteq \jointOut_0 \cup \SBP{\usInOutOp{\jointOut_0,\parallelCtx}{}{\gamma}}
		\supseteq \jointTaboo \cup \SBP{\usInOutOp{}{\jointOut_0}{\beta}} \cup \SBP{\usInOutOp{\jointOut_0,\parallelCtx}{}{\gamma}}
		\supseteq \jointTaboo \cup \SBP{\usInOutOp{}{\jointOut_0}{\beta} \seq \usInOutOp{\jointOut_0,\parallelCtx}{}{\gamma}}		
		\supseteq \jointTaboo \cup \SBP{\usInOutOp{}{}{\beta\seq\gamma}}$.

		\item $\alpha \equiv \beta \cup \gamma$, then $\usInOutOp{}{}{\beta \cup \gamma} \equiv \usInOutOp{}{\jointOut_1}{\beta} \cup \usInOutOp{}{\jointOut_2}{\gamma}$ with $\jointOut = \jointOut_1 \cup \jointOut_2$.
		Using IH, $\jointOut 
		= \jointOut_1 \cup \jointOut_2 
		\IH{\supseteq} \jointTaboo \cup \SBP{\usInOutOp{}{\jointOut_1}{\beta}} \cup \SBP{\usInOutOp{}{\jointOut_2}{\gamma}}$,
		which equals $\jointTaboo \cup \SBP{\usInOutOp{}{\jointOut_1}{\beta} \cup \usInOutOp{}{\jointOut_2}{\gamma}}
		= \jointTaboo \cup \SBP{\usInOutOp{}{}{\beta \cup \gamma}}$.

		\item $\alpha \equiv \repetition{\beta}$, then by IH, $\jointOut \supseteq \jointTaboo \cup \SBP{\usInOutOp{}{}{\beta}}$,
		\iest the input taboo $\jointTaboo$ is retained for $\beta$.
		Moreover, by IH, $\jointOut \supseteq \SBP{\usInOutOp{\jointOut, \parallelCtx}{}{\beta}}$ the taboo set $\jointOut$ after one iteration is retained for $\beta$.
		Since $\SBV(\beta) \supseteq \SBV(\repetition{\beta})$ and $\SCN(\beta) \supseteq \SCN(\repetition{\beta})$,
		we obtain $\jointOut \supseteq \jointTaboo \cup \SBP{\usInOutOp{\jointOut, \parallelCtx}{}{\beta}} = \jointTaboo \cup \SBP{\repetition{(\usInOutOp{\jointOut,\parallelCtx}{}{\beta})}} = \jointOut \cup \SBP{\usInOutOp{}{}{\repetition{\beta}}}$.

		\item $\alpha \equiv \send{}{}{}$, then $\jointOut = \jointTaboo \cup \{ \ch{}, \historyVar \}$ 
		and $\{ \ch{} \} = \SCN(\send{}{}{\usTabooOp{\tabooWithCtx}{\rp}})$,
		which equals
		$\SCN(\usInOutOp{}{}{\send{}{}{}})$
		and $\{ \historyVar \} = \SBV(\send{}{}{\usTabooOp{\tabooWithCtx}{\rp}}) = \SBV(\usInOutOp{}{}{\send{}{}{}})$.
		Thus, $\jointOut \supseteq \jointTaboo \cup \SBP{\usInOutOp{}{}{\send{}{}{}}}$.

		\item $\alpha \equiv \receive{}{}{}$, then $\jointOut = \jointTaboo \cup \{ \ch{}, \historyVar, x \}$.
		Now, observe $\{ \historyVar, x \} \supseteq \SBV(\receive{}{}{}) = \SBV(\usInOutOp{}{}{\receive{}{}{}})$.
		Moreover, $\{ \ch{} \} = \SCN(\receive{}{}{}) = \SCN(\usInOutOp{}{}{\receive{}{}{}})$.
		Thus, $\jointOut \supseteq \jointTaboo \cup \SBP{\usInOutOp{}{}{\receive{}{}{}}}$.

		\item $\alpha \equiv \beta \parOp \gamma$, then
		$\usInOutOp{}{}{\beta \parOp \gamma} \equiv 
			\usInOutOp{\jointTaboo, \parCtxOp{\gamma}}{\jointOut_1}{\beta} \parOp 
			\usInOutOp{\jointTaboo, \parCtxOp{\beta}}{\jointOut_1}{\gamma}$
		with $\jointOut = \jointOut_1 \cup \jointOut_2$.
		Using IH, $\jointOut 
		= \jointOut_1 \cup \jointOut_2 
		\IH{\supseteq} \jointTaboo \cup \SBP{\usInOutOp{\jointTaboo, \parCtxOp{\gamma}}{\jointOut_1}{\beta}} \cup \SBP{\usInOutOp{\jointTaboo, \parCtxOp{\beta}}{\jointOut_1}{\gamma}}
		\supseteq \jointTaboo \cup \SBP{\usInOutOp{\jointTaboo, \parCtxOp{\gamma}}{}{\beta} \parOp \usInOutOp{\jointTaboo, \parCtxOp{\beta}}{}{\gamma}}
		= \jointTaboo \cup \SBP{\usInOutOp{}{}{\beta \parOp \gamma}}$.
		\qedhere
	\end{enumerate}
\end{proof}

\begin{proof}[of \rref{lem:unisubs_term}]
	The proof generalizes the substitution lemma proof for \dGL \cite[Lemma 15]{DBLP:conf/cade/Platzer19} to multi-sorted terms and taboos with channels.
	The proof is by induction along the lexicographical order $\lexorder$ of substitution-term tuples $(\usubs, \expr)$ defined by $(\usubs', \expr') \lexorder (\usubs, \expr)$ if $\usubs' \lexorder \usubs$ or ($\usubs' = \usubs$ but $\expr' \lexorder \expr)$, where $\lexorder$ on substitutions and terms, respectively, denotes the structural order,
	simultaneously for all $\jointTaboo$, $\pstate{v}$, and $\varioOrigin$.
	In the following, let $\pstate{v}$ be any $\jointTaboo$-variation of $\varioOrigin$,
	\iest $\pstate{v} \downarrow \cset = \varioOrigin \downarrow \cset$ on $\jointTaboo^\complement \cap \V$ with $\cset = \jointTaboo^\complement \cap \Chan$:

	\begin{enumerate}
		\item For $\arbitraryVar \in \V$, simply $\sem{\usTabooOp{}{\arbitraryVar}}{\lstate{v}} = \sem{\arbitraryVar}{\lstate{v}} = \pstate{v}(\arbitraryVar) = \sem{\arbitraryVar}{\lstateadj{v}}$.
		
		\item In case $\fsymb(\cset, \expr)$, let $d = \sem{\usTabooOp{}{\expr \downarrow \cset}}{\lstate{v}}$. 
		Then by IH, $\sem{\usTabooOp{}{\fsymb(\expr \downarrow \cset)}}{\lstate{v}} = \sem{\usAuxOp{\expr}{\usubs \fsymb(\usarg)}}{\lstate{v}} \IH{=} \sem{\usubs \fsymb(\usarg)}{\inter \subs{\usarg}{d} \pstate{v}}$ because $\usarg$ has arity $0$ and $\fsymb$ has arity $1$ such that $\usAux{\expr} \lexorder \usubs$.
		Since $\usTabooOp{}{\fsymb(\expr \downarrow \cset)}$ is defined, $\SFV(\usubs \fsymb(\usarg)) \cap \jointTaboo = \emptyset$ and $\SCN(\usubs \fsymb(\usarg)) \cap \jointTaboo = \emptyset$.
		By premise, $\pstate{v}$ is a $\jointTaboo$-variation of~$\varioOrigin$, 
		\iest $\pstate{v} \downarrow (\jointTaboo^\complement \cap \Chan) = \varioOrigin \downarrow (\jointTaboo^\complement \cap \Chan)$ on $\jointTaboo^\complement \cap \V$.
		Thus, $\pstate{v} \downarrow \SCN(\usubs \fsymb(\usarg)) = \varioOrigin \downarrow \SCN(\usubs \fsymb(\usarg))$ on $\SFV(\usubs \fsymb(\usarg))$ such that $\sem{\usubs \fsymb(\usarg)}{\inter \subs{\usarg}{d} \pstate{v}} = \sem{\usubs \fsymb(\usarg)}{\inter \subs{\usarg}{d} \varioOrigin}$ by coincidence (see \rref{lem:termCoincidence}).
		Finally, by \rref{def:adjoint}, $\sem{\usubs \fsymb(\usarg)}{\inter \subs{\usarg}{d} \varioOrigin} = (\interadj (\fsymb)) (d)$,
		which equals by IH, $(\interadj (\fsymb)) (\sem{\expr \downarrow \cset}{\lstateadj{v}}) = \sem{\fsymb(\expr \downarrow \cset)}{\lstateadj{v}}$ because $\expr \downarrow \cset \lexorder \fsymb (\expr \downarrow \cset)$.
		
		\item Let $\fsymb[builtin](\expr_1, \ldots, \expr_k)$ be a concrete function.
		For $1 \le i \le k$, by IH, $\sem{\usTabooOp{}{\expr_i}}{\lstate{v}} = \sem{\expr_i}{\lstateadj{v}}$. 
		Finally, $\sem{\usTabooOp{}{\fsymb[builtin](\expr_1, \ldots, \expr_k)}}{\lstate{v}} 
		= \sem{\fsymb[builtin](\usTabooOp{}{\expr_1}, \ldots, \usTabooOp{}{\expr_k})}{\lstate{v}}
		= \sem{\fsymb[builtin](\expr_1, \ldots, \expr_k)}{\lstateadj{v}}$ 
		because $\sem{\fsymb[builtin](\expr_1, \ldots, \expr_k)}{\lstate{v}}$ is a function of $\sem{\expr_i}{\lstate{v}}$ for $1 \le i \le k$.

		\item By IH, $\sem{\usTabooOp{\V\cup\Chan}{\rp}}{\lstate{v}} = \sem{\rp} \lstateadj{v}$ for all $\pstate{v}, \pstate{w}$ since $\pstate{v}$ is a $(\V\cup\Chan)$-variation of any state.
		Hence, $\sem{\usTabooOp{}{(\rp)'}}{\lstate{v}} = \sem{(\usTabooOp{\V\cup\Chan}{\rp})'}{\lstate{v}} 
		= \sum_{x} \pstate{v}(x') \frac{\partial \sem{\usTabooOp{\V\cup\Chan}{\rp}}{\lstate{v}}}{\partial x} 
		\IH{=} \sum_{x} \pstate{v}(x') \frac{\partial \sem{\rp}{\lstateadj{v}}}{\partial x} = \sem{(\rp)'}{\lstateadj{v}}$.
		\qedhere
	\end{enumerate}
\end{proof}

\begin{lemma} \label{lem:outputVariation}
	Let $\varioComp \in \sem{\alpha}{\interadj}$, the uniform substitution $\usInOutOp{\jointTaboo, \emptyset}{}{\alpha}$ be defined, and $\pstate{v}$ be a $\jointTaboo$-variation of $\varioOrigin$.
	Then $\stconcat{\varioFin}{\trace}$ is a $\jointOut$-variation of $\varioOrigin$.
\end{lemma}
\begin{proof}
	Since $\pstate{v}$ is a $\jointTaboo$-variation of $\varioOrigin$ (\rref{def:variation}),
	we have
	\begin{equation} \label{eq:variation}
		\pstate{v} \downarrow \cset = \varioOrigin \downarrow \cset \text{ on } \jointTaboo^\complement \cap \V \text{ with }\cset = \jointTaboo^\complement \cap \Chan \text{,}
	\end{equation}
	and by the bound effect property (\rref{lem:boundEffect}), we obtain
	\begin{equation} \label{eq:bound}
		\pstate{v} = \varioOrigin \text{ on } \SBV(\usInOutOp{\jointTaboo, \emptyset}{}{\alpha})^\complement \cup \TVar \text{ and }\trace \downarrow \SCN(\usInOutOp{\jointTaboo, \emptyset}{}{\alpha})^\complement = \epsilon \text{.}
	\end{equation}
	
	If only $\stconcat{\varioFin}{\trace}$ is a $\jointTaboo_0$-variation of $\varioOrigin$,
	where $\jointTaboo_0 = \jointTaboo \cup \SBV(\usInOutOp{\jointTaboo, \emptyset}{}{\alpha}) \cup \SCN(\usInOutOp{\jointTaboo, \emptyset}{}{\alpha})$,
	then $\stconcat{\varioFin}{\trace}$ is a $\jointOut$-variation of $\varioOrigin$ by \rref{lem:unisubs_correct_bound}.
	To prove $\stconcat{\varioFin}{\trace}$ is a $\jointTaboo_0$-variation of $\varioOrigin$,
	we handle the variable restriction $\jointTaboo_0^\complement \cap \V$ separately for $\RealNatVar$ and~$\TVar$:	
	In case $\RealNatVar$, by \rref{eq:variation}, $\pstate{v} = \varioOrigin$ on $\jointTaboo^\complement \cap (\RealNatVar)$,
	and by \rref{eq:bound}, $\pstate{v} = \stconcat{\varioFin}{\trace}$ on $\SBV(\usInOutOp{\jointTaboo, \emptyset}{}{\alpha})^\complement \cap (\RealNatVar)$ such that
	$\stconcat{\varioFin}{\trace} = \pstate{v} = \varioOrigin$ on $\SBV(\usInOutOp{\jointTaboo, \emptyset}{\jointOut}{\alpha})^\complement \cap \jointTaboo^\complement \cap (\RealNatVar)$.
	Moreover, observe that $\SBV(\usInOutOp{\jointTaboo, \emptyset}{\jointOut}{\alpha})^\complement \cap \jointTaboo^\complement \cap (\RealNatVar) = \jointTaboo_0^\complement \cap (\RealNatVar)$ since $\SCN(\usInOutOp{\jointTaboo, \emptyset}{}{\alpha}) \cap (\RealNatVar) = \emptyset$.
	In case $\TVar$, let $\historyVar \in \jointTaboo_0^\complement \cap \TVar = \SBV(\usInOutOp{\jointTaboo, \emptyset}{\jointOut}{\alpha})^\complement \cap \jointTaboo^\complement \cap \TVar$ 
	and $\cset_0 = \jointTaboo_0^\complement \cap \Chan = \SCN(\usInOutOp{\jointTaboo, \emptyset}{\jointOut}{\alpha})^\complement \cap \cset$.
	Then we have
	\begin{equation*}
		(\stconcat{\varioFin}{\trace})(\historyVar) \downarrow \cset_0 = \varioFin(\historyVar) \downarrow \cset_0 \cdot \trace(\historyVar) \downarrow \cset_0 
		\overset{\text{(\ref{eq:bound})}}{=} \varioFin(\historyVar) \downarrow \cset_0 
		\overset{\text{(\ref{eq:bound})}}{=} \pstate{v}(\historyVar) \downarrow \cset_0 \
		\overset{\text{(\ref{eq:variation})}}{=} \varioFin(\historyVar) \downarrow \cset_0
	\end{equation*}
	because $\SCN(\usInOutOp{\jointTaboo, \emptyset}{}{\alpha}) \subseteq \cset_0$ and $\cset \subseteq \cset_0$.
	\qedhere
\end{proof}

\newcommand{\shortSubs}[1]{\usubs(#1)}

By mentioning $\cset$, 
the program constant $\rpconst{}{}$ signals that $\pconst$ synchronizes along all the channels $\cset$.%
\footnote{Synchronization must not be confused with actually reading or writing all the channels since a program can simply not communicate along a channel,
\eg $\skipProg \cup \send{}{}{}$ synchronizes on $\ch{}$ but not all runs communicate along $\ch{}$.}
Synchronization forces the parallel context to agree with the local program on the communication along the synchronized channels.
Uniform substitution must preserve synchronization as otherwise the parallel context could unsoundly perform additional communication.

For example, the substitution $\usubs = \{ \pconst \mapsto \test{\true}, \pconst[alt] \mapsto \send{}{}{} \}$
would turn the valid formula
\begin{equation*}
	[ \rpconst{\ch{}}{} ] \len{\historyVar \downarrow \ch{}} = 1 \rightarrow [ \rpconst{\ch{}}{} \parOp \rpconst[alt]{\ch{}}{\varvec^\complement} ] \len{\historyVar \downarrow \ch{}} = 1
\end{equation*}
into the invalid formula
\begin{equation*}
	[ \test{\true} ] \len{\historyVar \downarrow \ch{}} = 1 \rightarrow [ \test{\true} \parOp \send{}{}{} ] \len{\historyVar \downarrow \ch{}} = 1
\end{equation*}
stating that if the initial history contains one $\ch{}$-communication ($[ \test{\true} ] \len{\historyVar \downarrow \ch{}} = 1$),
there is still only one after sending along $\ch{}$ by $\send{}{}{}$ .
The problem is that the replacement for $\pconst$ no longer forces the replacement for $\pconst[alt]$ to synchronize along~$\ch{}$.
Where $\pconst[alt]$ could only communicate along $\ch{}$ if it agreed on this communication with $\pconst$,
the replacement for $\pconst[alt]$ unsoundly can perform additional communication independent of the replacement for $\pconst$.

Uniform substitution (see \rref{fig:unisubs}) preserves synchronization by the side condition $\SCN(\usubs \pconst) = \cset$ for replacing program constants $\rpconst{}{}$. 
The standard intuition for uniform substitution would suggest that $\SCN(\usubs \pconst) \subseteq \cset$ suffices since this already prevents the local replacement of $\pconst$ to unsoundly bind accessed channels. 

The additional inclusion $\SCN(\usubs \pconst) \supseteq \cset$ ensures that uniform substitution~$\shortSubs{\alpha}$ for program~$\alpha$ preserves the synchronization along the channels $\CN(\alpha)$ where~$\alpha$ synchronizes on,
\iest $\CN(\alpha) \subseteq \CN(\shortSubs{\alpha})$ with syntactical channels $\CN(\cdot)$ (see \rref{app:staticSemantics}).
Since the upper bound $\SCN(\usubs \alpha) \subseteq \cset$ is defined semantically, 
the substitution
$\shortSubs{\alpha}$ might introduce extra synchronization on the channels $\cset[alt]_\alpha = \CN(\shortSubs{\alpha}) \setminus \CN(\alpha)$.
However, this is harmless as $\shortSubs{\alpha}$ still performs no actual communication on $\cset[alt]_\alpha$,
\iest $C_\alpha \cap \SCN(\shortSubs{\alpha}) = \emptyset$.
Hence, potential extra synchronization in $\shortSubs{\alpha}$ at most leads to less behavior when parallel programs do not agree with empty communication along $\cset[alt]_\alpha$.
In summary, $\CN(\alpha)$ is bound as follows: $\SCN(\shortSubs{\alpha}) \subseteq \CN(\alpha) \subseteq \CN(\shortSubs{\alpha})$.
\rref{lem:substitutionPreservesSynchronization} heavily relies on these bounds and is crucially used later in proving \rref{lem:unisubs_fml_prog} about uniform substitution.

\begin{lemma}[Uniform substitution preserves synchronization]
	\label{lem:substitutionPreservesSynchronization}
	Let $\alpha$ and $\beta$ be programs, the substitutions $\shortSubs{\gamma} = \usInOutOp{}{}{\gamma}$ for $\gamma \in \{ \alpha, \beta \}$ be defined,
	and $\inter$ an interpretation.
	Then for $\gamma \in \{ \alpha, \beta \}$, the following are equivalent:
	\begin{enumerate}
		\item $(\pstate{v}, \trace \downarrow \shortSubs{\gamma}, \varioFin_\gamma) \in \sem{\shortSubs{\gamma}}{\inter}$ for $\gamma \in \{ \alpha, \beta \}$  and $\trace \downarrow (\shortSubs{\alpha} \parOp \shortSubs{\beta}) = \trace$
		\label{itm:syncUS}

		\item $(\pstate{v}, \trace \downarrow \gamma, \varioFin_\gamma) \in \sem{\shortSubs{\gamma}}{\inter}$ for $\gamma \in \{ \alpha, \beta \}$ and $\trace \downarrow (\alpha \parOp \beta) = \trace$
		\label{itm:syncNoUS}
	\end{enumerate}
\end{lemma}
\begin{proof}
	For $\gamma \in \{ \alpha, \beta \}$, the following set inclusions can be proven  by induction on the structure of $\gamma$:
	\begin{equation}
		\label{eq:channelChain}
		\SCN(\shortSubs{\gamma}) \subseteq \SCN(\gamma) \subseteq \CN(\gamma) \subseteq \CN(\shortSubs{\gamma})
	\end{equation}
	First, $\SCN(\shortSubs{\gamma}) \subseteq \SCN(\gamma)$ because locally each replacement $\usubs \pconst$ of a program constant $\rpconst{}{}{}$ cannot write more channels than $\rpconst{}{}$ due to the side condition $\SCN(\usubs \pconst) \subseteq \cset$ in \rref{fig:unisubs}.
	The inclusion $\SCN(\gamma) \subseteq \CN(\gamma)$ holds as $\CN(\gamma)$ is a sound overapproximation of $\SCN(\gamma)$ (see \rref{app:staticSemantics}).
	In the induction for $\CN(\gamma) \subseteq \CN(\shortSubs{\gamma})$,
	the case
	$\CN(\rpconst{}{}) = \cset \subseteq \SCN(\usubs \pconst) \subseteq \CN(\usubs \pconst)$ if $\gamma \equiv \rpconst{}{}$ uses $\SCN(\usubs \pconst) \supseteq \cset$ from \rref{fig:unisubs} again.

	First, we prove that \rref{itm:syncUS} implies \rref{itm:syncNoUS}.
	Therefore, let $\cset[alt]_\gamma = \CN(\shortSubs{\gamma}) \setminus \CN(\gamma)$ be the extra channels introduced by substitution.
	Now, let $(\pstate{v}, \trace \downarrow \shortSubs{\gamma}, \varioFin_\gamma) \in \sem{\shortSubs{\gamma}}{\inter}$ for $\gamma \in \{ \alpha, \beta \}$. 
	Then $\trace$ contains no communication along $\cset[alt]_\gamma$, \iest $\trace \downarrow \cset[alt]_\gamma = \epsilon$, since $\cset[alt]_\gamma \cap \SCN(\shortSubs{\gamma}) = \emptyset$ by \rref{eq:channelChain}.
	Hence, $\trace \downarrow \shortSubs{\gamma} = \trace \downarrow \gamma$ such that $(\pstate{v}, \trace \downarrow \gamma, \varioFin_\gamma) \in \sem{\shortSubs{\gamma}}{\inter}$.
	Moreover, let $\trace \downarrow (\shortSubs{\alpha}\parOp\shortSubs{\beta}) = \trace$.
	Intuitively,~$\trace$ has no extra communication outside $\CN(\alpha) \cup \CN(\beta)$ because $\trace \downarrow (\shortSubs{\alpha}\parOp\shortSubs{\beta}) = \trace$ and in~$\trace$ there is no communication on the difference $(\CN(\shortSubs{\alpha}) \cup \CN(\shortSubs{\beta})) \setminus (\CN(\alpha) \cup \CN(\beta))$.
	Formally, since $\trace \downarrow \cset[alt]_\gamma = \epsilon$ for $\gamma \in \{ \alpha, \beta \}$,
	the trace $\trace$ does not contain communication on $(\CN(\shortSubs{\alpha}) \cup \CN(\shortSubs{\beta})) \setminus (\CN(\alpha) \cup \CN(\beta)) 
	\subseteq \cset[alt]_\alpha \cup \cset[alt]_\beta$,
	which justifies the equality $(\star)$ in the following:
	\begin{equation*}
		\label{eq:nojunk}
		\trace \downarrow (\alpha\parOp\beta) 
		= (\trace \downarrow (\shortSubs{\alpha}\parOp\shortSubs{\beta})) \downarrow (\alpha\parOp\beta)
		\overset{(\star)}{=} \trace \downarrow (\shortSubs{\alpha} \parOp \shortSubs{\beta}) = \trace \text{.}
	\end{equation*}

	Conversely, let $(\pstate{v}, \trace \downarrow \gamma, \varioFin_\gamma) \in \sem{\shortSubs{\gamma}}{\inter}$ for $\gamma \in \{ \alpha, \beta \}$ and $\trace \downarrow (\alpha \parOp \beta) = \trace$.
	Moreover, let $\overline{\alpha} = \beta$ and $\overline{\beta} = \alpha$.
	Then 
	\begin{equation*}
		\trace \downarrow \shortSubs{\gamma} = (\trace \downarrow (\alpha\parOp\beta)) \downarrow \shortSubs{\gamma} = \trace \downarrow \Big( \big( \CN(\alpha) \cup \CN(\beta) \big) \cap \CN(\shortSubs{\gamma}) \Big) \text{,}
	\end{equation*}
	which equals $\trace \downarrow (\CN(\gamma) \cup (\CN(\overline{\gamma}) \cap \CN(\shortSubs{\gamma})))$
	because $\CN(\gamma) \subseteq \CN(\shortSubs{\gamma})$ by \rref{eq:channelChain}. 
	Since there is no communication along $\cset[alt]_\gamma = \CN(\shortSubs{\gamma}) \setminus \CN(\gamma)$ by $\cset[alt]_\gamma \cap \SCN(\shortSubs{\gamma}) = \emptyset$ again, 
	there is also non along
	$(\CN(\overline{\gamma}) \cap \CN(\shortSubs{\gamma})) \setminus \CN(\gamma) \subseteq \cset[alt]_\gamma$.
	Hence, $\trace \downarrow (\CN(\gamma) \cup (\CN(\overline{\gamma}) \cap \CN(\shortSubs{\gamma}))) = \trace \downarrow \gamma$ such that $\trace \downarrow \shortSubs{\gamma} = \trace \downarrow \gamma$,
	which implies $(\pstate{v}, \trace \downarrow \shortSubs{\gamma}, \varioFin_\gamma) \in \sem{\shortSubs{\gamma}}{\inter}$.
	Moreover, $\trace \downarrow (\alpha\parOp\beta) = \trace$ implies $\trace \downarrow (\shortSubs{\alpha}\parOp\shortSubs{\beta}) = (\trace \downarrow (\alpha\parOp\beta)) \downarrow (\shortSubs{\alpha}\parOp\shortSubs{\beta})$, which equals $\trace \downarrow (\alpha\parOp\beta) = \trace$ since $\CN(\shortSubs{\alpha}\parOp\shortSubs{\beta}) \supseteq \CN(\alpha\parOp\beta)$.
	\qedhere
\end{proof}

\begin{proof}[of \rref{lem:unisubs_fml_prog}]
	The proof generalizes the substitution lemma proof for \dGL \cite[Lemma 16 and Lemma 17]{DBLP:conf/cade/Platzer19} to \dLCHP,
	where taboos contain channels and the parallel context needs to be respected.
	The proof is by lexicographic mutual structural induction, 
	\iest along the lexicographic order $\lexorder$ on tuples $(\usubs, \pi)$ of substitutions $\usubs$ and formula-program expressions $\pi$ defined by $(\usubs', \pi') \lexorder (\usubs, \pi)$ if $\usubs' \lexorder \usubs$ or ($\usubs' = \usubs$ but $\pi' \lexorder \pi)$, 
	where~$\lexorder$ on substitutions and formula-program expressions, respectively, denotes the (mutual) structural order,
	simultaneously for all $\jointTaboo$, $\pstate{v}$, and $\varioOrigin$.

	First, consider the formula cases.
	Therefore, let $\pstate{v}$ be any $\jointTaboo$-variation of~$\varioOrigin$.

	\begin{enumerate}
		\item $\lstate{v} \vDash \usTabooOp{}{\expr_1 \sim \expr_2}$ iff $\lstate{v} \vDash \usTabooOp{}{\expr_1} \sim \usTabooOp{}{\expr_2}$ iff $\sem{\usTabooOp{}{\expr_1}}{\lstate{v}} \sim \sem{\usTabooOp{}{\expr_2}}{\lstate{v}}$, which is by \rref{lem:unisubs_term} equivalent to $\sem{\expr_1}{\lstateadj{v}} \sim \sem{\expr_2}{\lstateadj{v}}$, iff $\lstateadj{v} \vDash \expr_1 \sim \expr_2$
		
		\item In case $\psymb(\cset, \expr)$, let $d = \sem{\usTabooOp{}{\expr \downarrow \cset}}{\lstate{v}}$. 
		Then $\lstate{v} \vDash \usTabooOp{}{\psymb(\expr \downarrow \cset)}$ 
		iff $\lstate{v} \vDash \usAuxOp{\expr}{\usubs \psymb(\usarg)}$,
		by IH, 
		iff $\inter \subs{\usarg}{d} \pstate{v} \vDash \usubs \psymb(\usarg)$ because $\usarg$ has arity $0$ and $\psymb$ has arity $1$ such that $\usAux{\expr} \lexorder \usubs$.
		Since $\usTabooOp{}{\psymb(\expr \downarrow \cset)}$ is defined, $\SFV(\usubs \psymb(\usarg)) \cap \jointTaboo = \emptyset$ and $\SCN(\usubs \psymb(\usarg)) \cap \jointTaboo = \emptyset$.
		By premise, $\pstate{v}$ is a $\jointTaboo$-variation of~$\varioOrigin$, 
		\iest $\pstate{v} \downarrow (\jointTaboo^\complement \cap \Chan) = \varioOrigin \downarrow (\jointTaboo^\complement \cap \Chan)$ on $\jointTaboo^\complement \cap \V$.
		Thus, $\pstate{v} \downarrow \SCN(\usubs \psymb(\usarg)) = \varioOrigin \downarrow \SCN(\usubs \psymb(\usarg))$ on $\SFV(\usubs \psymb(\usarg))$ such that 
		($\inter \subs{\usarg}{d} \pstate{v} \vDash \usubs \psymb(\usarg)$ iff $\inter \subs{\usarg}{d} \varioOrigin \vDash \usubs \psymb(\usarg)$) by coincidence (see \rref{lem:termCoincidence}).
		Finally, $\inter \subs{\usarg}{d} \varioOrigin \vDash \usubs \psymb(\usarg)$,
		by \rref{def:adjoint}, 
		iff $(d) \in \interadj[w] (\psymb)$,
		by IH, 
		iff $(\sem{\expr \downarrow \cset}{\lstateadj[w]{v}}) \in \interadj[w] (\psymb)$
		iff $\lstateadj[w]{v} \vDash \psymb(\expr \downarrow \cset)$ because $\expr \downarrow \cset \lexorder \psymb(\expr \downarrow \cset)$.

		\item $\lstate{v} \vDash \usTabooOp{}{\neg \phi}$ iff $\lstate{v} \vDash \neg \usTabooOp{}{\phi}$ iff $\lstate{v} \nvDash \usTabooOp{}{\phi}$, by IH, iff $\lstateadj{v} \nvDash \phi$ iff $\lstateadj{v} \vDash \neg \phi$
	
		\item $\lstate{v} \vDash \usTabooOp{}{\varphi \wedge \psi}$ iff $\lstate{v} \vDash \usTabooOp{}{\varphi} \wedge \usTabooOp{}{\psi}$ iff ($\lstate{v} \vDash \usTabooOp{}{\varphi}$ and $\lstate{v} \vDash \usTabooOp{}{\psi}$), by IH, iff ($\lstateadj{v} \vDash \varphi$ and $\lstateadj{v} \vDash \psi$) iff $\lstateadj{v} \vDash \varphi \wedge \psi$
				
		\item Observe that $\pstate{v} \subs{\arbitraryVar}{\semConst}$ is a $(\jointTaboo\cup\{\arbitraryVar\})$-variation of $\varioOrigin$ for all $\semConst \in \type(\arbitraryVar)$ since $\pstate{v}$ is a $\jointTaboo$-variation of~$\varioOrigin$.
		Now, $\lstate{v} \vDash \usTabooOp{}{\fa{\arbitraryVar} \varphi}$ 
		iff $\lstate{v} \vDash \fa{\arbitraryVar} \usTabooOp{\jointTaboo\cup\{\arbitraryVar\}}{\varphi}$ 
		iff $\lstate{v} \subs{\arbitraryVar}{\semConst} \vDash \usTabooOp{\jointTaboo\cup\{\arbitraryVar\}}{\varphi}$ for all $\semConst \in \type(\arbitraryVar)$,
		by IH, iff $\lstateadj{v} \subs{\arbitraryVar}{\semConst} \vDash \varphi$ for all $\semConst \in \type(\arbitraryVar)$
		iff $\lstateadj{v} \vDash \fa{\arbitraryVar} \varphi$.

		\item Let $\lstate{v} \vDash \usTabooOp{}{[ \alpha ] \psi}$.
		Then $\lstate{v} \vDash [ \usInOutOp{\jointTaboo, \emptyset}{\jointOut}{\alpha} ] \usTabooOp{\jointOut}{\psi}$.
		To prove $\lstateadj{v} \vDash [ \alpha ] \psi$,
		let $\varioComp \in \sem{\alpha}{\interadj}$ with $\varioFin \neq \bot$.
		By the mutual IH,
		$\varioComp \in \sem{\usInOutOp{\jointTaboo, \emptyset}{}{\alpha}}{\inter}$ because $\pstate{v}$ is a $(\jointTaboo\cup\emptyset)$-variation of $\varioOrigin$.
		Hence, $\stconcat{\lvarioFin}{\trace} \vDash \usTabooOp{\jointOut}{\psi}$ by the premise.
		Finally, by IH, $\stconcat{\lstateadj{\varioFinName}}{\trace} \vDash \psi$ because $\stconcat{\varioFin}{\trace}$ is a $\jointOut$-variation of $\varioOrigin$ by \rref{lem:outputVariation}.

		Conversely, let $\lstateadj{v} \vDash [ \alpha ] \psi$.
		To prove $\lstate{v} \vDash \usTabooOp{}{[ \alpha ] \psi}$,
		recall $\usTabooOp{}{[ \alpha ] \psi} = [ \usInOutOp{\jointTaboo, \emptyset}{}{\alpha} ] \usTabooOp{\jointOut}{\psi}$ and let $\varioComp \in \sem{\usInOutOp{\jointTaboo, \emptyset}{}{\alpha}}{\inter}$ with $\varioFin \neq \bot$.
		By the mutual IH,
		$\varioComp \in \sem{\alpha}{\interadj}$ because $\pstate{v}$ is a $(\jointTaboo\cup\emptyset)$-variation of $\varioOrigin$.
		Therefore, $\interadj \stconcat{\varioFin}{\trace} \vDash \psi$ by the premise.
		Finally, by IH, $\stconcat{\lvarioFin}{\trace} \vDash \usTabooOp{\jointOut}{\psi}$ because $\stconcat{\varioFin}{\trace}$ is a $\jointOut$-variation of $\varioOrigin$ by \rref{lem:outputVariation}.
		
		\item Let $\lstate{v} \vDash \usTabooOp{}{[ \alpha ] \ac \psi}$.
		Then $\lstate{v} \vDash [ \usInOutOp{\jointTaboo, \emptyset}{}{\alpha} ] \acpair{\usTabooOp{\jointOut}{\A}, \usTabooOp{\jointOut}{\C}} \usTabooOp{\jointOut}{\psi}$.
		To prove $\lstateadj{v} \vDash [ \alpha ] \ac \psi$,
		let $\varioComp \in \sem{\alpha}{\interadj}$.
		By the mutual IH,
		$\varioComp \in \sem{\usInOutOp{\jointTaboo, \emptyset}{}{\alpha}}{\inter}$ because $\pstate{v}$ is a $(\jointTaboo\cup\emptyset)$-variation of $\varioOrigin$.
		Since variation is monotone in the variation set,
		$\pstate{v}$ is a $(\jointTaboo\cup\SBV(\usInOutOp{\jointTaboo, \emptyset}{}{\alpha})\cup\SCN(\usInOutOp{\jointTaboo, \emptyset}{}{\alpha}))$-variation of~$\varioOrigin$.
		Moreover, by the bound effect property (\rref{lem:boundEffect}), we have $\trace \downarrow \SCN(\usInOutOp{\jointTaboo, \emptyset}{}{\alpha})^\complement = \epsilon$,
		which implies that $\stconcat{\pstate{v}}{\trace[pre]}$ is a $(\jointTaboo\cup\SBV(\usInOutOp{\jointTaboo, \emptyset}{}{\alpha})\cup\SCN(\usInOutOp{\jointTaboo, \emptyset}{}{\alpha}))$-variation of $\varioOrigin$ for all $\trace[pre] \preceq \trace$.
		Hence, by \rref{lem:unisubs_correct_bound}, $\stconcat{\pstate{v}}{\trace}[pre]$ is a $\jointOut$-variation of $\varioOrigin$ for all $\trace[pre] \preceq \trace$.
		For \acCommit, 
		assume $\assCommit{\lstateadj{v}}{\trace} \vDash \A$.
		By IH, $\assCommit{\lstate{v}}{\trace} \vDash \usTabooOp{\jointOut}{\A}$.
		Hence, $\lstate{v} \cdot \trace \vDash \usTabooOp{\jointOut}{\C}$ by premise, 
		which implies $\stconcat{\lstateadj{v}}{\trace} \vDash \C$ by IH again.
		For \acPost,
		assume $\varioFin \neq \bot$ and $\assPost{\lstateadj{v}}{\trace} \vDash \A$.
		By IH, $\assPost{\lstate{v}}{\trace} \vDash \usTabooOp{\jointOut}{\A}$.
		Hence, $\stconcat{\lstate{o}}{\trace} \vDash \usTabooOp{\jointOut}{\psi}$ by premise,
		which implies $\stconcat{\lstateadj{\varioFinName}}{\trace} \vDash \psi$ by IH again because $\stconcat{\varioFin}{\trace}$ is a $\jointOut$-variation of $\varioOrigin$ by \rref{lem:outputVariation}.

		Conversely, let $\lstateadj{v} \vDash [ \alpha ] \ac \psi$.
		To prove $\lstate{v} \vDash \usTabooOp{}{[ \alpha ] \ac \psi}$,
		recall that $\usTabooOp{}{[ \alpha ] \ac \psi} = [ \usInOutOp{\jointTaboo, \emptyset}{}{\alpha} ] \acpair{\usTabooOp{\jointOut}{\A}, \usTabooOp{\jointOut}{\C}} \usTabooOp{\jointOut}{\psi}$ and
		let $\computation \in \sem{\usInOutOp{\jointTaboo, \emptyset}{}{\alpha}}{\inter}$.
		By the mutual IH, $\computation \in \sem{\alpha}{\interadj}$ because $\pstate{v}$ is a $(\jointTaboo\cup\emptyset)$-variation of $\varioOrigin$.
		As conversely, $\stconcat{\pstate{v}}{\trace}[pre]$ is a $\jointOut$-variation of $\varioOrigin$ for all $\trace[pre] \preceq \trace$.
		For \acCommit, assume $\assCommit{\lstate{v}}{\trace} \vDash \usTabooOp{\jointOut}{\A}$.
		By IH $\assCommit{\lstateadj{v}}{\trace} \vDash \A$. Hence, $\stconcat{\lstateadj{v}}{\trace} \vDash \C$ by premise,
		which implies $\stconcat{\lstate{v}}{\trace} \vDash \usTabooOp{\jointOut}{\C}$ by IH again.
		For \acPost, assume $\varioFin \neq \bot$ and $\assPost{\lstate{v}}{\A} \vDash \usTabooOp{\jointOut}{\A}$.
		By IH $\assPost{\lstateadj{v}}{\trace} \vDash \A$.
		Hence, $\lstateadj{\varioFinName} \cdot \trace \vDash \psi$ by premise,
		which implies $\stconcat{\lvarioFin}{\trace} \vDash \usTabooOp{\jointOut}{\psi}$ by IH again because $\stconcat{\varioFin}{\trace}$ is a $\jointOut$-variation of $\varioOrigin$ by \rref{lem:outputVariation}.
	\end{enumerate}

	Secondly, consider the program cases.
	Therefore, let $\pstate{v}$ be any $(\jointTaboo\cup\parallelCtx)$-variation of $\varioOrigin$.
	\Wlossg $\trace = \epsilon$ and $\varioFin \neq \bot$ in the cases of non-communicating atomic programs $\alpha$ below
	because $\trace = \epsilon$ for all $\varioComp \in \sem{\alpha}{\inter}$ if $\alpha$ is an atomic program but not a communication primitive and ($\leastComputation \in \sem{\usInOutOp{}{}{\alpha}}{\inter}$ iff $\leastComputation \in \sem{\alpha}{\interadj}$) by totality.

	\newcommand{\nonComComp}{(\pstate{v}, \epsilon, \varioFin)}

	\begin{enumerate}
		\item 
		Since $\usInOutOp{}{}{\rpconst{}{}}$ is defined, $\SBV(\usubs \pconst) \subseteq \varvec$ and $\SCN(\usubs \pconst) \subseteq \cset$.
		Hence, the equivalence $\varioComp \in \sem{\usInOutOp{}{}{\rpconst{}{}}}{\inter}$
		iff $\varioComp \in \sem{\usubs \pconst}{\inter}$ is defined,
		where the latter is by \rref{def:adjoint} equivalent to $\varioComp \in \interadj(\rpconst{}{})$
		iff $\varioComp \in \sem{\rpconst{}{}}{\interadj}$.

		\item $\nonComComp \in \sem{\usInOutOp{}{}{x \ceq \rp}}{\inter}$
		iff $\nonComComp \in \sem{x \ceq \usTabooOp{\jointTaboo\cup\parallelCtx}{\rp}}{\inter}$ 
		iff $\varioFin = \pstate{v} \subs{x}{\semConst}$, where $\semConst = \sem{\usTabooOp{\jointTaboo\cup\parallelCtx}{\rp}}{\lstate{v}}$, 
		by \rref{lem:unisubs_term}, 
		iff $\varioFin = \pstate{v} \subs{x}{\semConst}$, where $\semConst = \sem{\rp}{\lstateadj{v}}$
		iff $\nonComComp \in \sem{x \ceq \rp}{\interadj}$
		
		\item $\nonComComp \in \sem{\usTabooOp{}{x \ceq *}}{\inter}$
		iff $\varioFin = \pstate{v} \subs{x}{\semConst}$, where $\semConst \in \reals$,
		iff $\nonComComp \in \sem{x \ceq *}{\interadj}$

		\item $\nonComComp \in \sem{\usInOutOp{}{}{\evolution{}{}}}{\inter}$
		iff $\nonComComp \in \sem{\evolution{x' = \usTabooOp{\jointTaboo\cup\parallelCtx}{\rp}}{\usTabooOp{\jointTaboo\cup\parallelCtx}{\chi}}}{\inter}$
		iff a solution $\odeSolution : [0, \duration] \rightarrow \states$ from $\pstate{v}$ to $\varioFin$ exists with $\lsolution(\zeta) \vDash \globalTime' = 1 \wedge x' = \usTabooOp{\jointTaboo\cup\parallelCtx}{\rp} \wedge \usTabooOp{\jointTaboo\cup\parallelCtx}{\chi}$ and $\pstate{v} = \odeSolution(\zeta)$ on $\odeBoundVars^\complement$ for all $\zeta \in [0, \duration]$ and if $\odeSolution(\zeta)(\arbitraryVar') = \solutionDerivative{\odeSolution}{\arbitraryVar}(\arbitraryVar)$ for $\arbitraryVar \in \{\globalTime, x\}$.
		Since $\odeSolution(\zeta)$ is a $(\jointTaboo\cup\odeBoundVars)$-variation of $\varioOrigin$,
		we obtain $\lsolutionadj(\zeta) \vDash \globalTime' = 1 \wedge x' = \rp$ by \rref{lem:unisubs_term} and $\lsolutionadj(\zeta) \vDash \chi$ by the mutual IH.
		Thus, $\nonComComp \in \sem{\usInOutOp{}{}{\evolution{}{}}}{\inter}$ iff $\varioComp \in \sem{\evolution{}{}}{\interadj}$.

		\item $\nonComComp \in \sem{\usInOutOp{}{}{\test{}}}{\inter}$ 
		iff $\nonComComp \in \sem{\test{\usTabooOp{\jointTaboo\cup\parallelCtx}{\chi}}}{\inter}$
		iff $\varioFin = \pstate{v}$ and $\lstate{v} \vDash \usTabooOp{\jointTaboo\cup\parallelCtx}{\chi}$,
		by the mutual IH,
		iff $\varioFin = \pstate{v}$ and $\lstateadj{v} \vDash \chi$
		iff $\nonComComp \in \sem{\test{\chi}}{\interadj}$

		\item $\varioComp \in \sem{\usInOutOp{}{}{\send{}{}{}}}{\inter}$ 
		iff $\varioComp \in \sem{\send{}{}{\usTabooOp{\jointTaboo\cup\parallelCtx}{\rp}}}{\inter}$
		iff $\varioObservable \preceq \semCom{\semConst}{\pstate{v}}$ with $\semConst = \sem{\usTabooOp{\jointTaboo\cup\parallelCtx}{\rp}}{\lstate{v}}$,
		which by \rref{lem:unisubs_term},
		is equivalent to $\varioObservable \preceq \semCom{\semConst^*}{\pstate{v}}$ with $\semConst^* = \sem{\rp}{\lstateadj{v}}$
		iff $\varioComp \in \sem{\send{}{}{}}{\interadj}$

		\item $\varioComp \in \sem{\usInOutOp{}{}{\receive{}{}{}}}{\inter}$
		iff $\varioComp \in \sem{\receive{}{}{}}{\inter}$,
		which is equivalent to $\varioObservable \preceq \semCom{\semConst}{\pstate{v} \subs{x}{\semConst}}$ with $\semConst \in \reals$
		iff $\varioComp \in \sem{\receive{}{}{}}{\interadj}$

		\item $\varioComp \in \sem{\usInOutOp{}{}{\alpha\cup\beta}}{\inter}$ iff $\varioComp \in \sem{\usInOutOp{}{\jointOut_1}{\alpha} \cup \usInOutOp{}{\jointOut_2}{\beta}}{\inter}$, where $\jointOut = \jointOut_1 \cup \jointOut_2$,
		iff $\varioComp \in \sem{\usInOutOp{}{\jointOut_1}{\alpha}}{\inter}$ or $\varioComp \in \sem{\usInOutOp{}{\jointOut_2}{\beta}}{\inter}$,
		which by IH,
		is equivalent to $\varioComp \in \sem{\alpha}{\interadj}$ or $\varioComp \in \sem{\beta}{\interadj}$
		iff $\varioComp \in \sem{\alpha\cup\beta}{\interadj}$

		\item Let $\varioComp \in \sem{\usInOutOp{}{}{\alpha \seq \beta}}{\inter}$.
		Since $\usInOutOp{}{}{\alpha \seq \beta} \equiv \usInOutOp{}{\jointOut_0}{\alpha} \seq \usInOutOp{\jointOut_0, \parallelCtx}{}{\beta}$,
		there is $\varioComp \in \botop{(\sem{\usInOutOp{}{\jointOut_0}{\alpha}}{\inter})}$ or $\varioComp \in \sem{\usInOutOp{}{\jointOut_0}{\alpha}}{\inter} \continuation \sem{\usInOutOp{\jointOut_0, \parallelCtx}{}{\beta}}{\inter}$.
		In case $\varioComp \in \botop{(\sem{\usInOutOp{}{\jointOut_0}{\alpha}}{\inter})}$,
		there is $\varioComp \in \botop{(\sem{\alpha}{\interadj})} \subseteq \sem{\alpha \seq \beta}{\interadj}$ by IH.
		Otherwise, $\varioComp \in \sem{\usInOutOp{}{\jointOut_0}{\alpha}}{\inter} \continuation \sem{\usInOutOp{\jointOut_0, \parallelCtx}{}{\beta}}{\inter}$, 
		so computations $(\pstate{v}, \trace_1, \varioIntermediate) \in \sem{\usInOutOp{}{\jointOut_0}{\alpha}}{\inter} $ and $(\varioIntermediate, \trace_2, \varioFin) \in \sem{\usInOutOp{\jointOut_0, \parallelCtx}{}{\beta}}{\inter}$ exist with $\trace = \trace_1 \cdot \trace_2$.
		By IH, $(\pstate{v}, \trace_1, \varioIntermediate) \in \sem{\alpha}{\interadj}$.
		Since $\pstate{v}$ is a $(\jointTaboo\cup\parallelCtx)$-variation of $\varioOrigin$
		and $\pstate{v} = \varioIntermediate$ on $\SBV(\usInOutOp{}{\jointOut_0}{\alpha})^\complement$ by the bound effect property,
		we obtain $\varioOrigin = \pstate{v} = \varioIntermediate$ on $(\jointTaboo\cup\parallelCtx)^\complement \cap \SBV(\usInOutOp{}{\jointOut_0}{\alpha})^\complement = \parallelCtx^\complement \cap (\jointTaboo \cup \SBV(\usInOutOp{}{\jointOut_0}{\alpha}))^\complement$.
		By \rref{lem:unisubs_correct_bound},
		$\jointOut_0 \supseteq \jointTaboo \cup \SBV(\usInOutOp{}{\jointOut_0}{\alpha})$ such that $\varioOrigin = \varioIntermediate$ on $\parallelCtx^\complement \cap \jointOut_0^\complement = (\parallelCtx\cup\jointOut_0)^\complement$,
		\iest $\varioIntermediate$ is a $(\jointOut_0\cup\parallelCtx)$-variation of $\varioOrigin$ such that IH is applicable on $(\varioIntermediate, \trace_2, \varioFin) \in \sem{\usTabooOp{\jointOut_0\cup\parallelCtx}{\beta}}{\inter}$.
		Therefore, by IH, $(\varioIntermediate, \trace_2, \varioFin) \in \sem{\beta}{\interadj}$.
		Finally, $\varioComp \in \sem{\alpha}{\interadj} \continuation \sem{\beta}{\interadj} \subseteq \sem{\alpha \seq \beta}{\interadj}$.

		Conversely, let $\varioComp \in \sem{\alpha \seq \beta}{\interadj}$.
		Then $\varioComp \in \botop{(\sem{\alpha}{\interadj})}$ or $\varioComp \in \sem{\alpha}{\interadj} \continuation \sem{\beta}{\interadj}$.
		If $\varioComp \in \botop{(\sem{\alpha}{\interadj})}$,
		there is $\varioComp \in \botop{(\sem{\usInOutOp{}{\jointOut_0}{\alpha}}{\inter})} \subseteq \sem{\usInOutOp{}{\jointOut_0}{\alpha} \seq \usInOutOp{\jointOut_0, \parallelCtx}{}{\beta}}{\inter} = \sem{\usInOutOp{}{}{\alpha \seq \beta}}{\inter}$ by IH.
		Otherwise, $\varioComp \in \sem{\alpha}{\interadj} \continuation \sem{\beta}{\interadj}$,
		so computations $(\pstate{v}, \trace_1, \varioIntermediate) \in \sem{\alpha}{\interadj}$ and $(\varioIntermediate, \trace_2, \varioFin) \in \sem{\beta}{\interadj}$ exist such that $\trace = \trace_1 \cdot \trace_2$.
		By IH, $(\pstate{v}, \trace_1, \varioIntermediate) \in \sem{\usInOutOp{}{\jointOut_0}{\alpha}}{\inter}$.
		As conversely, $\varioIntermediate$ is a $(\jointOut_0\cup\parallelCtx)$-variation of $\varioOrigin$ such that IH is applicable on $(\varioIntermediate, \trace_2, \varioFin) \in \sem{\beta}{\interadj}$.
		Therefore, by IH, $(\varioIntermediate, \trace_2, \varioFin) \in \sem{\usInOutOp{\jointOut_0,\parallelCtx}{}{\beta}}{\inter}$.
		Finally, $\varioComp \in \sem{\usInOutOp{}{\jointOut_0}{\alpha}}{\inter} \continuation \sem{\usInOutOp{\jointOut_0, \parallelCtx}{}{\beta}}{\inter} \subseteq \sem{\usInOutOp{}{\jointOut_0}{\alpha} \seq \usInOutOp{\jointOut_0, \parallelCtx}{}{\beta}}{\inter} = \sem{\usInOutOp{}{}{\alpha \seq \beta}}{\inter}$.

		\item In case $\usInOutOp{}{}{\repetition{\alpha}}$, the output taboo $\jointOut \supseteq \jointTaboo$ of $\usInOutOp{}{}{\alpha}$ is added to $\jointTaboo$ during an additional first pass over $\alpha$ as indicated by the fixpoint notation $\usInOutOp{\jointOut, \parallelCtx}{}{\alpha}$.
		Since $\pstate{v}$ is a $(\jointTaboo\cup\parallelCtx)$-variation of $\pstate{w}$ and variation is monotone in the variation set, $\pstate{v}$ is a $(\jointOut\cup\parallelCtx)$-variation of $\pstate{w}$.
		Now, $\varioComp \in \sem{\usInOutOp{}{}{\repetition{\alpha}}}{\inter}$
		iff $\varioComp \in \sem{\repetition{\usInOutOp{\jointOut, \parallelCtx}{}{\alpha}}}{\inter}$
		iff $\varioComp \in \sem{(\usInOutOp{\jointOut, \parallelCtx}{}{\alpha})^n}{\inter}$ for some $n \in \naturals$
		iff $\varioComp \in \sem{\usInOutOp{\jointOut, \parallelCtx}{}{\alpha^n}}{\inter}$ for some $n \in \naturals$,
		by IH, iff $\varioComp \in \sem{\alpha^n}{\interadj}$ for some $n \in \naturals$
		iff $\varioComp \in \sem{\repetition{\alpha}}{\interadj}$.
 
		\item

		In case $\alpha \parOp \beta$, we have
		$\usInOutOp{}{}{\alpha \parOp \beta} \equiv \shortSubs{\alpha} \parOp \shortSubs{\beta}$,
		where $\shortSubs{\alpha} \equiv \usInOutOp{\jointTaboo, \parCtxOp{\beta}}{\jointOut_1}{\alpha}$,
		and $\shortSubs{\beta} \equiv \usInOutOp{\jointTaboo, \parCtxOp{\alpha}}{\jointOut_2}{\beta}$,
		and $\parCtxOp{\gamma} \equiv \parallelCtx \cup (\SBV(\usInOutOp{}{{}}{\gamma}) \setminus (\{\globalTime, \globalTime'\} \cup \TVar))$ for any program~$\gamma$,
		and $\jointOut = \jointOut_1 \cup \jointOut_2$.
		Since $\pstate{v}$ is a $(\jointTaboo\cup\parallelCtx)$-variation of $\varioOrigin$,
		by monotony in the variation set,
		$\pstate{v}$ is a $(\jointTaboo\cup\parCtxOp{\gamma})$-variation of $\varioOrigin$ for any $\gamma$.
		We define explicit merging $\varioFin_1 \merge_\gamma \varioFin_2$ to be $\varioFin_1$ on $\SBV(\gamma)$ and  $\varioFin_2$ elsewhere.
		
		Now, $\varioComp \in \sem{\usInOutOp{}{}{\alpha \parOp \beta}}{\inter}$
		iff $(\pstate{v}, \trace \downarrow \shortSubs{\gamma}, \varioFin_\gamma) \in \sem{\shortSubs{\gamma}}{\inter}$ for $\gamma\in\{\alpha,\beta\}$, 
		and $\varioFin = \varioFin_\alpha \merge_{\shortSubs{\alpha}} \varioFin_\beta$,
		and $\varioFin_\alpha = \varioFin_\beta$ on $\{ \globalTime, \globalTime' \}$, and
		$\trace = \trace \downarrow (\shortSubs{\alpha} \parOp \shortSubs{\beta})$,
		iff, by \rref{lem:substitutionPreservesSynchronization} and the argument about merging $\merge$ below, $(\pstate{v}, \trace \downarrow \gamma, \varioFin_\gamma) \in \sem{\shortSubs{\gamma}}{\interadj}$ for $\gamma\in\{\alpha,\beta\}$, 
		and $\varioFin = \varioFin_\alpha \merge_\alpha \varioFin_\beta$,
		and $\varioFin_\alpha = \varioFin_\beta$ on $\{ \globalTime, \globalTime' \}$, and
		$\trace = \trace \downarrow (\alpha\parOp\beta)$
		iff, by IH, $(\pstate{v}, \trace \downarrow \gamma, \varioFin_\gamma) \in \sem{\gamma}{\interadj}$ for $\gamma\in\{\alpha,\beta\}$, 
		and $\varioFin = \varioFin_\alpha \merge_\alpha \varioFin_\beta$,
		and $\varioFin_\alpha = \varioFin_\beta$ on $\{ \globalTime, \globalTime' \}$, and
		$\trace = \trace \downarrow (\alpha\parOp\beta)$
		iff $\varioComp \in \sem{\alpha\parOp\beta}{\interadj}$.

		Finally, $\varioFin_\alpha \merge_\alpha \varioFin_\beta = \varioFin_\alpha \merge_{\shortSubs{\alpha}} \varioFin_\beta$ has been left open above.
		On $\SBV(\shortSubs{\alpha})$, we have $\varioFin_\alpha \merge_\alpha \varioFin_\beta = \varioFin_\alpha = \varioFin_\alpha \merge_{\shortSubs{\alpha}} \varioFin_\beta$ since $\SBV(\shortSubs{\alpha}) \subseteq \SBV(\alpha)$.
		On $\SBV(\shortSubs{\alpha})^\complement$, first consider $\SBV(\shortSubs{\alpha})^\complement \cap \SBV(\alpha)^\complement$,
		where $\varioFin_\alpha \merge_\alpha \varioFin_\beta = \varioFin_\beta = \varioFin_\alpha \merge_{\shortSubs{\alpha}} \varioFin_\beta$.
		Now, consider $\varset = \SBV(\shortSubs{\alpha})^\complement \cap \SBV(\alpha)$.
		Further, let $\varset_\text{obs} = \{ \arbitraryVar \mid \pstate{v}(\arbitraryVar) \neq \varioFin_\alpha(\arbitraryVar) \vee \pstate{v}(\arbitraryVar) \neq \varioFin_\beta(\arbitraryVar) \}$.
		On $\varset \cap \varset_\text{obs}^\complement$,
		we have $\varioFin_\alpha \merge_{\shortSubs{\alpha}} \varioFin_\beta = \pstate{v} = \varioFin_\alpha \merge_\alpha \varioFin_\beta$.
		Otherwise, if $\arbitraryVar \in \varset \cap \varset_\text{obs}$, then $\arbitraryVar \in \SBV(\alpha)$ but $\arbitraryVar \not\in \TVar$ by \rref{lem:boundEffect}.
		Then either $\arbitraryVar \in \RVar \setminus \{ \globalTime, \globalTime' \}$ such that $\arbitraryVar \not\in \SBV(\beta)$ by well-formedness of $\alpha \parOp \beta$ or $\arbitraryVar \in \{ \globalTime, \globalTime' \}$ such that the programs must agree upon the value of $\arbitraryVar$ in their final states.
		If $\arbitraryVar \not\in \SBV(\beta)$, then $\pstate{v}(\arbitraryVar) \neq \varioFin_\alpha(\arbitraryVar)$ since $\arbitraryVar \in \varset_\text{obs}$ such that $\arbitraryVar \in \SBV(\shortSubs{\alpha})$ in contradiction to $\arbitraryVar \in \varset$.
		Therefore, $\arbitraryVar \not\in \varset \cap \varset_\text{obs}$ such that $\varset \cap \varset_\text{obs} = \emptyset$.
		Otherwise, if $\arbitraryVar \in \{ \globalTime, \globalTime' \}$, the final states agree upon the value of $\arbitraryVar$.
		Thus, $\varioFin_\alpha \merge_\alpha \varioFin_\beta = \varioFin_\alpha \merge_{\shortSubs{\alpha}} \varioFin_\beta$ on $\{ \globalTime, \globalTime' \}$.
		\qedhere
	\end{enumerate}
\end{proof}

Once the uniform substitution lemmas from \rref{sec:unisubsLemmas} are proven, 
the soundness proof of uniform substitution (\rref{thm:ussound}) is easy:

\begin{proof}[of \rref{thm:ussound}]
	Let the premise $\phi$ be valid, \iest $\lstate[alt]{v} \vDash \phi$ for all pairs $\lstate[alt]{v}$ of interpretation and state.
	For proving the conclusion, let $\lstate{v}$ be any pair of interpretation and state.
	By validity of the premise, thus $\lstateadj[v]{v} \vDash \phi$.
	Since $\pstate{v}$ is a $\emptyset$-variation of $\pstate{v}$,
	\rref{lem:unisubs_fml_prog} implies $\lstate{v} \vDash \usTabooOp{\emptyset}{\phi}$.
	\qedhere
\end{proof}

Besides the instantiation of axioms by \RuleName{US}, 
\rref{thm:usrulesound} even allows instantiation of axiomatic proof rules using uniform substitution.
The rule must be locally sound,
\iest validity of the premises in any interpretation implies validity of the conclusion under this interpretation.

\begin{theorem}[Sound uniform substitution for rules] 
	\label{thm:usrulesound}
	If the inference \textnormal{INF} is locally sound, 
	so is the inference \textnormal{US-INF}:
	\begin{equation*}
		\frac{\phi_1 \;\ldots\; \phi_k}{\psi} \;\;(\text{\textnormal{INF}})
		\qquad\qquad
		\frac{\usTabooOp{\V\cup\Chan}{\phi_1} \;\ldots\; \usTabooOp{\V\cup\Chan}{\phi_k}}{\usTabooOp{\V\cup\Chan}{\psi}} \;\;(\text{\textnormal{US-INF}})
	\end{equation*}
\end{theorem}
\begin{proof}
	Assume that the inference INF is locally sound.
	To prove that the inference US-INF is locally sound,
	let $\inter$ be any interpretation such that $\inter \vDash \usTabooOp{\V\cup\Chan}{\phi_j}$ for $1 \le j \le k$,
	\iest $\lstate{v} \vDash \usTabooOp{\V\cup\Chan}{\phi_j}$ for all $\pstate{v}$ and $j$.
	Since $\pstate{v}$ is a $(\V\cup\Chan)$-variation of any $\pstate{w}$,
	by \rref{lem:unisubs_fml_prog}, $\interadj \vDash \phi_j$ for all $j$.
	Thus, $\interadj \vDash \psi$ by local soundness of INF.
	Finally, $\inter \vDash \usTabooOp{\V\cup\Chan}{\psi}$ by \rref{lem:unisubs_fml_prog} again.
	\qedhere
\end{proof}

\rref{lem:unisubs_parallel_ctx} prepares the proof of \rref{prop:unisubs_well_formed} that uniform substitution preserves well-formedness proving that it respects the parallel context.
As well-formedness is defined in terms of the static semantics, the proofs of \rref{lem:unisubs_parallel_ctx} and \rref{prop:unisubs_well_formed} make use of the uniform substitution lemmas via \rref{lem:freeAndBoundSubs}.

\begin{lemma} \label{lem:freeAndBoundSubs}
	For formula $\phi$ and program $\alpha$, we have
	\begin{enumerate}
		\item $\SFV(\usTabooOp{}{\phi}) \subseteq \SFV(\phi) \cup \jointTaboo^\complement$, and \label{item:freeTermSubs}
		\item $\SFV(\usInOutOp{}{}{\alpha}) \subseteq \SFV(\alpha) \cup (\jointTaboo\cup\parallelCtx)^\complement$, and \label{itm:freeProgSubs}
		\item $\SBV(\usInOutOp{}{}{\alpha}) \subseteq \SBV(\alpha)$ \label{itm:boundSubs}
	\end{enumerate}
	if the substitutions $\usTabooOp{}{\phi}$ and $\usInOutOp{}{}{\alpha}$ are defined.
\end{lemma}
\begin{proof}
	For \rref{item:freeTermSubs}, let $\arbitraryVar \in \SFV(\usTabooOp{}{\phi})$ but assume $\arbitraryVar \not\in \jointTaboo^\complement$.
	By definition of $\SFV(\phi)$,
	there are $\inter$, $\pstate{v}$, and $\pstate[alt]{v}$ with $\pstate{v} = \pstate[alt]{v}$ on $\{\arbitraryVar\}^\complement$ and $\sem{\usTabooOp{}{\phi}}{\lstate{v}} \neq \sem{\usTabooOp{}{\phi}}{\lstate[opt]{v}}$.
	Since $\arbitraryVar \in \jointTaboo$,
	state $\pstate[alt]{v}$ is a $\jointTaboo$-variation of $\pstate{v}$.
	By \rref{lem:unisubs_fml_prog}, $\sem{\phi}{\lstateadj[v]{v}} \neq \sem{\phi}{\interadj[v] \pstate[alt]{v}}$.
	Therefore, $\arbitraryVar \in \SFV(\phi)$.

	For \rref{itm:freeProgSubs}, let $\arbitraryVar \in \SFV(\usInOutOp{}{}{\alpha})$ but assume $\arbitraryVar \not\in (\jointTaboo\cup\parallelCtx)^\complement$.
	By definition of $\SFV(\alpha)$, there are $\inter, \pstate{v}, \pstate[alt]{v}, \trace, \pstate{w}$ such that $\pstate{v} = \pstate[alt]{v}$ on $\{\arbitraryVar\}^\complement$ and $\computation \in \sem{\usInOutOp{}{}{\alpha}}{\inter}$ but there is no $\computation[alt] \in \sem{\usInOutOp{}{}{\alpha}}{\inter}$ such that $\trace[alt] = \trace$ and $\pstate[alt]{w} = \pstate{w}$ on $\{\arbitraryVar\}^\complement$.
	By \rref{lem:unisubs_fml_prog},  $\computation \in \sem{\alpha}{\interadj[v]}$.
	Since $x \in \jointTaboo\cup\parallelCtx$, 
	state $\pstate[alt]{v}$ is a $(\jointTaboo\cup\parallelCtx)$-variation of $\pstate{v}$.
	Hence, there is no $\computation[alt] \in \sem{\alpha}{\interadj[v]}$ such that $\trace[alt] = \trace$ and $\pstate[alt]{w} = \pstate{w}$ on $\{\arbitraryVar\}^\complement$ because otherwise, $\computation[alt] \in \sem{\usInOutOp{}{}{\alpha}}{\inter}$ by \rref{lem:unisubs_fml_prog}.
	In summary, $\arbitraryVar \in \SFV(\alpha)$.

	For \rref{itm:boundSubs}, let $\arbitraryVar \in \SBV(\usInOutOp{}{}{\alpha})$.
	Then $\inter$ and $\computation \in \sem{\usInOutOp{}{}{\alpha}}{\inter}$ exist such that $(\stconcat{\pstate{w}}{\trace})(\arbitraryVar) \neq \pstate{v}(\arbitraryVar)$.
	By \rref{lem:unisubs_fml_prog}, $\computation \in \sem{\alpha}{\interadj[v]}$.
	Hence, $\arbitraryVar \in \SBV(\alpha)$.
	\qedhere
\end{proof}

\begin{lemma} \label{lem:unisubs_parallel_ctx}
	If a program $\alpha$ respects the parallel context $\parallelCtx$,
	\iest $\SV(\alpha) \cap \parallelCtx \subseteq \{\globalTime,\globalTime'\} \cup \TVar$,
	the result of substitution $\usInOutOp{}{}{\alpha}$, if defined, respects the context~$\parallelCtx$, too,
	\iest $\SV(\usInOutOp{}{}{\alpha}) \cap \parallelCtx \subseteq \{\globalTime, \globalTime'\} \cup \TVar$.
\end{lemma}
\begin{proof}
	Let $\SV(\alpha) \cap \parallelCtx \subseteq \{\globalTime, \globalTime'\} \cup \TVar$.
	Then $\SV(\usInOutOp{}{}{\alpha}) \cap \parallelCtx = (\SFV(\usInOutOp{}{}{\alpha}) \cup \SBV(\usInOutOp{}{}{\alpha})) \cap \parallelCtx$, which is by \rref{lem:freeAndBoundSubs} smaller or equal to $(\SFV(\alpha) \cup (\jointTaboo\cup\parallelCtx)^\complement \cup \SBV(\alpha)) \cap \parallelCtx = (\SFV(\alpha) \cup \SBV(\alpha)) \cap \parallelCtx = \SV(\alpha) \cap \parallelCtx$,
	which is by the premise smaller or equal to $\{\globalTime, \globalTime'\} \cup \TVar$.
	\qedhere
\end{proof}

The proof of \rref{prop:unisubs_well_formed} shows that uniform substitution (see \rref{fig:unisubs}) already preserves the stronger well-formedness condition $(\SV(\alpha) \cap \SBV(\beta)) \cup (\SV(\beta) \cap \SBV(\alpha)) \subseteq \{\globalTime, \globalTime'\} \cup \TVar$ for programs $\alpha, \beta$ from previous work \citeDLCHP.
Importantly, note that the weaker well-formedness condition $\SBV(\alpha) \cap \SBV(\beta) \subseteq \{\globalTime, \globalTime'\} \cup \TVar$ as imposed on programs in this paper,
is still sufficient to preserve itself and that the proofs of \rref{lem:unisubs_parallel_ctx} and \rref{prop:unisubs_well_formed} can be adjusted accordingly.

\begin{proof}[of \rref{prop:unisubs_well_formed}]
	The proof is by simultaneous induction on the structure of programs and formulas.
	It uses the abbreviation $\Gshared = \{\globalTime, \globalTime'\} \cup \TVar$.
	
	The only program with context-sensitive syntax is parallel composition, thus atomic programs are trivially well-formed and other compound programs are well-formed since by IH, their subprograms are well-formed.
	
	Now, let $\alpha \parOp \beta$ be well-formed.
	Then $(\SV(\alpha) \cap \SBV(\beta)) \cup (\SV(\beta) \cap \SBV(\alpha)) \subseteq \Gshared$.
	Moreover, let $\usInOutOp{}{}{\alpha \parOp \beta}$ be defined.
	Thus, 
	$\usInOutOp{\jointTaboo, \parCtxOp{\beta}}{\jointOut_1}{\alpha}$ and 
	$\usInOutOp{\jointTaboo, \parCtxOp{\alpha}}{\jointOut_2}{\beta}$, 
	where $\parCtxOp{\alpha} = \parCtxOpExpanded[\Gshared]{\alpha}$
	and $\parCtxOp{\beta} = \parCtxOpExpanded[\Gshared]{\beta}$,
	are defined, and by IH, they are well-formed.
	By \rref{lem:unisubs_parallel_ctx},
	$\Gshared \supseteq \SV(\usInOutOp{\jointTaboo, \parCtxOp{\beta}}{\jointOut_1}{\alpha}) \cap \parCtxOp{\beta} 
	\supseteq \SV(\usInOutOp{\jointTaboo, \parCtxOp{\beta}}{\jointOut_1}{\alpha}) \cap (\SBV(\usInOutOp{}{{}}{\beta}) \setminus \Gshared)$,
	which equals $\SV(\usInOutOp{\jointTaboo, \parCtxOp{\beta}}{\jointOut_1}{\alpha}) \cap (\SBV(\usInOutOp{\jointTaboo, \parCtxOp{\alpha}}{{}}{\beta}) \setminus \Gshared)$ since the parallel context does not influence the substitution result if it is defined.
	Thus, $\Gshared \supseteq \SV(\usInOutOp{\jointTaboo, \parCtxOp{\beta}}{\jointOut_1}{\alpha}) \cap \SBV(\usInOutOp{\jointTaboo, \parCtxOp{\alpha}}{{}}{\beta})$.
	Accordingly, $\Gshared \supseteq \SV(\usInOutOp{\jointTaboo, \parCtxOp{\alpha}}{\jointOut_2}{\beta}) \cap \SBV(\usInOutOp{\jointTaboo, \parCtxOp{\beta}}{{}}{\alpha})$ by \rref{lem:unisubs_parallel_ctx}.
	Since the context-sensitive side conditions are respected, $\usInOutOp{}{}{\alpha \parOp \beta}$ is well-formed as parallel composition of well-formed programs.

	For formulas, the ac-box is the only interesting case.
	Let $[ \alpha ] \ac \psi$ be well-formed.
	Then $(\SFV(\A) \cup \SFV(\C)) \cap \SBV(\alpha) \subseteq \TVar$.
	Moreover, let $\usTabooOp{}{[ \alpha ] \ac \psi}$ be defined.
	For $\chi \in \{\A, \C\}$,
	we obtain by \rref{lem:freeAndBoundSubs} that $\SFV(\usTabooOp{\jointOut}{\chi}) \cap \SBV(\usInOutOp{\jointTaboo, \emptyset}{}{\alpha}) \subseteq (\SFV(\chi) \cup \jointOut^\complement) \cap \SBV(\usInOutOp{\jointTaboo, \emptyset}{}{\alpha})$,
	which is by \rref{lem:unisubs_correct_bound} smaller or equal to
	\begin{equation*}
		\Big( \SFV(\chi) \cup \big( \jointTaboo \cup \SBV(\usInOutOp{\jointTaboo, \emptyset}{}{\alpha}) \cup \SCN(\usInOutOp{\jointTaboo, \emptyset}{}{\alpha}) \big)^\complement \Big) \cap \SBV(\usInOutOp{\jointTaboo, \emptyset}{}{\alpha}) \text{,}
	\end{equation*}
	which equals $\SFV(\chi) \cap \SBV(\usInOutOp{\jointTaboo, \emptyset}{}{\alpha})$,
	which is by \rref{lem:freeAndBoundSubs} smaller or equal to $\SFV(\chi) \cap \SBV(\alpha)$,
	which is smaller or equal to $\TVar$ by the premise.
	\qedhere
\end{proof}

\section{Details of the Axiomatic Calculus} \label{app:calculus}

This appendix reports a soundness proof for \dLCHP's axiomatization (see \rref{fig:calculus}).
Moreover, \rref{cor:derivations} gives derivations of ac-monotonicity \RuleName{acMono} and distribution of boxes over conjuncts \RuleName{acBoxesDist}.
Finally, algebraic laws for reasoning about trace terms \citeDLCHP are lifted to uniform substitution.

\begin{proof}[of \rref{thm:soundness}]
	Since the axioms of \rref{fig:calculus} are instances of their schematic versions (except for \RuleName{acModalMP} and \RuleName{assumptionWeak}),
	they are sound as the schematic axioms are sound \citeDLCHP.
	In particular, note that axioms \RuleName{acNoCom} and \RuleName{acDropComp} internalize the side conditions of the schematic calculus correctly.
	For the newly added axiom \RuleName{acModalMP} and the changed axiom \RuleName{assumptionWeak}, 
	soundness proofs are provided below.
	Recall that $\Asymb_j \equiv \asymb_j(\cset, \hvarvec)$, and $\Csymb_j \equiv \csymb_j(\cset, \hvarvec)$, and $\Psymb_j \equiv \psymb_j(\cset, \varvec)$.%

	\begin{itemize}[leftmargin=1em, itemindent=-1em]
		\itemsep.5em
		\item[] \RuleName{acModalMP}:
		Let $\lstate{v} \vDash [ \pconst ] \acpair{\Asymb, \Csymb_1 \rightarrow \Csymb_2}(\Psymb_1 \rightarrow \Psymb_2)$,
		and $\lstate{v} \vDash [ \pconst ] \acpair{\Asymb, \Csymb_1} \Psymb_1$,
		and let $\computation \in \sem{\pconst}{\inter}$.
		For \acCommit, assume $\assCommit{\lstate{v}}{\trace} \vDash \Asymb$.
		Then $\stconcat{\lstate{v}}{\trace} \vDash \Csymb_1 \rightarrow \Csymb_2$ and $\stconcat{\lstate{v}}{\trace} \vDash \Csymb_1$, 
		thus $\stconcat{\lstate{v}}{\trace} \vDash \Csymb_2$.
		For \acPost, assume $\pstate{w} \neq \bot$ and $\assPost{\lstate{v}}{\trace} \vDash \Asymb$.
		Then $\stconcat{\lstate{v}}{\trace} \vDash \Psymb_1 \rightarrow \Psymb_2$ and $\stconcat{\lstate{v}}{\trace} \vDash \Psymb_1$ such that $\stconcat{\lstate{v}}{\trace} \vDash \Psymb_2$.

		\item[] \RuleName{assumptionWeak}:
		Let $\lstate{v} \vDash [ \pconst ] \acpair{\true, \compCondition} \true$, where $\compCondition \equiv \compConditionExpanded$,
		and let $\lstate{v} \vDash \Phi$ for $\Phi \equiv [ \pconst ] \acpair{\Asymb_1 \wedge \Asymb_2, \Csymb_1 \wedge \Csymb_2} \Psymb$ and $\computation \in \sem{\pconst}{\inter}$.
		By induction on the length of~$\trace$,
		we simultaneously prove $\lstate{v} \vDash [ \pconst ] \acpair{\Asymb, \Csymb_1 \wedge \Csymb_2} \Psymb$, and that $\assCommit{\lstate{v}}{\trace} \vDash \Asymb$ implies $\assCommit{\lstate{v}}{\trace} \vDash \Asymb_1 \wedge \Asymb_2$:
		
		\begin{enumerate}
			\item $\semLen{\trace} = 0$, 
			then $\lstate{v} \vDash \Phi$ implies $\lstate{v} \vDash \Csymb_1 \wedge \Csymb_2$ by axiom \RuleName{acWeak}.
			Hence, \acCommit holds because $\trace = \epsilon$.
			For \acPost, assume $\pstate{w} \neq \bot$ and $\assPost{\lstate{v}}{\trace} \vDash \Asymb$. 
			By $\lstate{v} \vDash [ \pconst ] \acpair{\true, \compCondition} \true$ and \RuleName{acWeak} again, we obtain $\lstate{v} \vDash \compCondition$.
			Then $\assPost{\lstate{v}}{\trace} \vDash \Asymb_1 \wedge \Asymb_2$ because $\lstate{v} \vDash \compCondition$, and $\lstate{v} \vDash \Asymb$, and $\lstate{v} \vDash \Csymb_1 \wedge \Csymb_2$.
			Thus, $\stconcat{\lstate{w}}{\trace} \vDash \Psymb$ by the premise $\lstate{v} \vDash \Phi$. 
			Note that $\assCommit{\lstate{v}}{\trace} \vDash \Asymb_1 \wedge \Asymb_2$ is trivially fulfilled since $\assCommit{\lstate{v}}{\trace} = \emptyset$.
			
			\item $\semLen{\trace} > 0$,
			then $\trace = \trace_0 \cdot \rawtrace$ with $\semLen{\rawtrace} = 1$.
			For \acCommit, assume $\assCommit{\lstate{v}}{\trace} \vDash \Asymb$.
			Then $\assPost{\lstate{v}}{\trace_0} \vDash \Asymb$, which implies $\assCommit{\lstate{v}}{\trace_0} \vDash \Asymb_1 \wedge \Asymb_2$ by IH.
			Since $(v, \trace_0, \bot) \in \sem{\pconst}{\inter}$ by prefix-closedness of the program semantics,
			we obtain $\stconcat{\lstate{v}}{\trace_0} \vDash \Csymb_1 \wedge \Csymb_2$ by $\lstate{v} \vDash \Phi$.
			Moreover, $\stconcat{\lstate{v}}{\trace_0} \vDash \compCondition$ by $\lstate{v} \vDash [ \pconst ] \acpair{\true, \compCondition} \true$.
			Hence, $\stconcat{\lstate{v}}{\trace_0} \vDash \Asymb_1 \wedge \Asymb_2$ since $\stconcat{\lstate{v}}{\trace_0} \vDash \compCondition$, and $\stconcat{\lstate{v}}{\trace_0} \vDash \Asymb$,
			and $\stconcat{\lstate{v}}{\trace_0} \vDash \Csymb_1 \wedge \Csymb_2$.
			Thus, $\assCommit{\lstate{v}}{\trace} \vDash \Asymb_1 \wedge \Asymb_2$. 
			Finally, $\stconcat{\lstate{v}}{\trace} \vDash \Csymb_1 \wedge \Csymb_2$ by $\lstate{v} \vDash \Phi$ again. 
			For \acPost, assume $\pstate{w} \neq \bot$ and $\assPost{\lstate{v}}{\trace} \vDash \Asymb$. 
			Then $\assCommit{\lstate{v}}{\trace} \vDash \Asymb$, which implies $\assCommit{\lstate{v}}{\trace} \vDash \Asymb_1 \wedge \Asymb_2$ and $\stconcat{\lstate{v}}{\trace} \vDash \Csymb_1 \wedge \Csymb_2$ as in case \acCommit. 
			By $\computation \in \sem{\pconst}{\inter}$ and $\lstate{v} \vDash [ \pconst ] \acpair{\true, \compCondition} \true$, we obtain $\stconcat{\lstate{v}}{\trace} \vDash \compCondition$, which implies using $\stconcat{\lstate{v}}{\trace} \vDash \Asymb$ and $\stconcat{\lstate{v}}{\trace} \vDash \Csymb_1 \wedge \Csymb_2$ that $\stconcat{\lstate{v}}{\trace} \vDash \Asymb_1 \wedge \Asymb_2$. 
			In summary, $\assPost{\lstate{v}}{\trace} \vDash \Asymb_1 \wedge \Asymb_2$, which implies $\stconcat{\lstate{w}}{\trace} \vDash \Psymb$ by $\lstate{v} \vDash \Phi$.
		\end{enumerate}
	\end{itemize}
\end{proof}

\begin{corollary}
	\label{cor:derivations}
	The proof rule \RuleName{acMono} of ac-monotonicity and the axiom \RuleName{acBoxesDist} of distribution of conjuncts in commitments and postconditions over boxes \citeDLCHP can be derived.
	Let $\Asymb_j \equiv \asymb_j(\cset, \hvarvec)$, and $\Csymb_j \equiv \csymb_j(\cset, \hvarvec)$, and $\Psymb_j \equiv \psymb_j(\cset, \varvec)$.%
	\vspace*{-1em}%
	\begin{prooftree}
		\Axiom{$\Asymb_2 \rightarrow \Asymb_1$}

		\Axiom{$\Csymb_1 \rightarrow \Csymb_2$}

		\Axiom{$\Psymb_1 \rightarrow \Psymb_2$}

		\RuleNameRight{acMono}
		\TrinaryInf{$[ \pconst ] \acpair{\Asymb_1, \Csymb_1} \Psymb_1 \rightarrow [ \pconst ] \acpair{\Asymb_2, \Csymb_2} \Psymb_2$}
	\end{prooftree}

	\begin{equation*}
		[ \pconst ] \acpair{\Asymb, \Csymb_1 \wedge \Csymb_2} (\Psymb_1 \wedge \Psymb_2) \leftrightarrow [ \pconst ] \acpair{\Asymb, \Csymb_1} \Psymb_1 \wedge [ \pconst ] \acpair{\Asymb, \Csymb_2} \Psymb_2 \quad\text{\RuleName{acBoxesDist}}
	\end{equation*}
\end{corollary}
\begin{proof}
	The proof is by derivation in the calculus.
	The sequent-style deduction is justified since sequent-style rules can be derived in a Hoare-style calculus.

	\begin{itemize}[leftmargin=1em, itemindent=-1em]
		\itemsep.5em
		\item[] \RuleName{acMono}:
		Let $\psi \equiv (\Asymb_2 \wedge \Csymb_2 \rightarrow \true) \wedge (\Asymb_2 \wedge \true \rightarrow \Asymb_1)$ in the following,
		where \RuleName{prop} marks propositional reasoning.%
		\vspace*{-1em}%
		\begin{prooftree}
			\Axiom{$\Asymb_2 \rightarrow \Asymb_1$}
			
			\RuleNameLeft{prop}
			\UnaryInf{$\true \wedge \psi$}
		
			\RuleNameLeft{acG}
			\UnaryInf{$[ \pconst ] \acpair{\true, \psi} \true$}
		
			\RuleNameLeft{WL}
			\UnaryInf{$[ \pconst ] \acpair{\Asymb_1, \Csymb_1} \Psymb_1 \rightarrow [ \pconst ] \acpair{\true, \psi} \true$}
		
			\Axiom{$\Csymb_1 \rightarrow \Csymb_2$}
		
			\Axiom{$\Psymb_1 \rightarrow \Psymb_2$}
		
			\RuleNameRight{andR}
			\BinaryInf{$(\Csymb_1 \rightarrow \Csymb_2) \wedge (\Psymb_1 \rightarrow \Psymb_2)$}
		
			\RuleNameRight{acG}
			\UnaryInf{$[ \pconst ] \acpair{\Asymb_1, \Csymb_1 \rightarrow \Csymb_2} (\Psymb_1 \rightarrow \Psymb_2)$}
		
			\RuleNameRight{acModalMP}
			\UnaryInf{$[ \pconst ] \acpair{\Asymb_1, \Csymb_1} \Psymb_1 \rightarrow [ \pconst ] \acpair{\Asymb_1, \Csymb_2} \Psymb_2$}
		
			\RuleNameRight{trueN}
			\UnaryInf{$[ \pconst ] \acpair{\Asymb_1, \Csymb_1} \Psymb_1 \rightarrow [ \pconst ] \acpair{\Asymb_1 \wedge \true, \Csymb_2 \wedge \true} \Psymb_2$}
		
			\RuleNameRight{andR}
			\BinaryInf{$[ \pconst ] \acpair{\Asymb_1, \Csymb_1} \Psymb_1 \rightarrow [ \pconst ] \acpair{\true, \psi} \true \wedge [ \pconst ] \acpair{\Asymb_1 \wedge \true, \Csymb_2 \wedge \true} \Psymb_2$}
		
			\RuleNameRight{assumptionWeak}
			\UnaryInf{$[ \pconst ] \acpair{\Asymb_1, \Csymb_1} \Psymb_1 \rightarrow [ \pconst ] \acpair{\Asymb_2, \Csymb_2 \wedge \true} \Psymb_2$}
		
			\RuleNameRight{trueN}
			\UnaryInf{$[ \pconst ] \acpair{\Asymb_1, \Csymb_1} \Psymb_1 \rightarrow [ \pconst ] \acpair{\Asymb_2, \Csymb_2} \Psymb_2$}
		\end{prooftree}
		
		\item[] \RuleName{acBoxesDist}:
		The implication ($\rightarrow$) can be easily derived using rule \RuleName{acMono}.
		The other direction is derived below,
		where \RuleName{prop} marks propositional reasoning.
		The proof uses currying \RuleName{curry}, which can can be easily derived by \RuleName{prop}.%
		\vspace*{-1em}%
		\begin{prooftree}
			\Axiom{$*$}

			\RuleNameLeft{prop}
			\UnaryInf{$\Csymb_1 \rightarrow (\Csymb_2 \rightarrow \Csymb_1 \wedge \Csymb_2)$}

			\Axiom{$*$}

			\RuleNameRight{prop}
			\UnaryInf{$\Psymb_1 \rightarrow (\Psymb_2 \rightarrow \Psymb_1 \wedge \Psymb_2)$}

			\RuleNameRight{acMono}
			\BinaryInf{$[ \pconst ] \acpair{\Asymb, \Csymb_1} \Psymb_1 \rightarrow [ \pconst ] \acpair{\Asymb, \Csymb_2 \rightarrow \Csymb_1 \wedge \Csymb_2} (\Psymb_2 \rightarrow \Psymb_1 \wedge \Psymb_2)$}

			\RuleNameRight{acModalMP}
			\UnaryInf{$[ \pconst ] \acpair{\Asymb, \Csymb_1} \Psymb_1 \rightarrow ( [ \pconst ] \acpair{\Asymb, \Csymb_2} \Psymb_2 \rightarrow [ \pconst ] \acpair{\Asymb, \Csymb_1 \wedge \Csymb_2} (\Psymb_1 \wedge \Psymb_2) )$}

			\RuleNameRight{curry}
			\UnaryInf{$[ \pconst ] \acpair{\Asymb, \Csymb_1} \Psymb_1 \wedge [ \pconst ] \acpair{\Asymb, \Csymb_2} \Psymb_2 \rightarrow [ \pconst ] \acpair{\Asymb, \Csymb_1 \wedge \Csymb_2} (\Psymb_1 \wedge \Psymb_2)$}
		\end{prooftree}
		\qedhere
	\end{itemize}
\end{proof}

\subsubsection{Algebra of Traces} 

\rref{fig:traceAlgebra} gives simple algebraic laws for step-wise simplification of trace terms.
In contrast to the schematic algebra of traces in our previous report \citeDLCHP,
the laws in \rref{fig:traceAlgebra} are flat axioms without side conditions.
The axioms use that a symbolic representation of (co)finite sets can be given and finitely axiomatized, especially in axiom \RuleName{projCut}, axiom \RuleName{projIn}, and axiom \RuleName{projNotIn}.

\renewcommand{\rulesep}{\\[.4em]}
\begin{figure}[h]
	\begin{minipage}{\textwidth}
		\begin{calculus}
			\startAxiom{concatDist}
				$(\fsymb[trace] \cdot \fsymb[trace, alt]) \downarrow \cset = \fsymb[trace] \downarrow \cset \cdot \fsymb[trace, alt] \downarrow \cset$
			\stopAxiom
			\startAxiom{projCut}
				$(\fsymb[trace] \downarrow \cset') \downarrow \cset = \fsymb[trace] \downarrow (\cset' \cap \cset)$
			\stopAxiom
			\startAxiom{projNeutral}
				$\epsilon \downarrow \cset = \epsilon$
			\stopAxiom
			\startAxiom{val}
				$\val{\comItem{\ch{}, \fsymb[real], \fsymb[real, alt]}} = \fsymb[real]$
			\stopAxiom
			\startAxiom{time}
				$\stamp{\comItem{\ch{}, \fsymb[real], \fsymb[real, alt]}} = \fsymb[real, alt]$
			\stopAxiom
			\startAxiom{chan}
				$\chan{\comItem{\ch{}, \fsymb[real], \fsymb[real, alt]}} = \ch{}$
			\stopAxiom
		\end{calculus}
		\begin{calculus}
			\startAxiom{concatAssoc}
				$(\fsymb[trace]_1 \cdot \fsymb[trace]_2) \cdot \fsymb[trace]_3 = \fsymb[trace]_1 \cdot (\fsymb[trace]_2 \cdot \fsymb[trace]_3)$
			\stopAxiom
			\startAxiom{concatNeutral}
				$\fsymb[trace] \cdot \epsilon = \fsymb[trace] = \epsilon \cdot \fsymb[trace]$
			\stopAxiom
			\startAxiom{projIn}
				$\ch{} \in \cset \rightarrow \comItem{\ch{}, \fsymb[real], \fsymb[real, alt]} \downarrow \cset = \comItem{\ch{}, \fsymb[real], \fsymb[real, alt]}$
			\stopAxiom
			\startAxiom{projNotIn}
				$ch{} \not\in \cset \rightarrow \comItem{\ch{}, \fsymb[real], \fsymb[real, alt]} \downarrow \cset = \epsilon$
			\stopAxiom
			\startAxiom{nonNegative}
				$\len{\fsymb[trace]} \ge 0$
			\stopAxiom
			\startAxiom{unroll}
				$\len{\fsymb[trace] \cdot \comItem{\ch{}, \fsymb[real], \fsymb[real, alt]}} = \len{\fsymb[trace]} + 1$
			\stopAxiom
		\end{calculus}
		
		\vspace{.5em}
		\begin{calculus}
			\startAxiom{accessBase}
				$\len{\fsymb[trace]} = \fsymb[int] \rightarrow \at{(\fsymb[trace] \cdot \comItem{\ch{}, \fsymb[real], \fsymb[real, alt]})}{\fsymb[int]} = \comItem{\ch{}, \fsymb[real], \fsymb[real, alt]}$
			\stopAxiom
			\startAxiom{accessInd}
				$\len{\fsymb[trace]} > \fsymb[int] \rightarrow \at{(\fsymb[trace] \cdot \comItem{\ch{}, \fsymb[real], \fsymb[real, alt]})}{\fsymb[int]} = \at{\fsymb[trace]}{\fsymb[int]}$
			\stopAxiom
		\end{calculus}
	\end{minipage}
	\caption{Axiomatic algebra of traces}
	\vspace*{-2em}
	\label{fig:traceAlgebra}
\end{figure}

\subsubsection{Axiomatization of (Co)finite Sets}

Strictly speaking, the calculus in \rref{fig:calculus} still has schematic occurrences of (co)finite sets.
As suggested by \rref{rem:finite},
this is easily fixed using symbolic (co)finite sets together with a non-schematic axiomatization.
The class $C(\setAtoms_1, \ldots, \setAtoms_n)$ of (co)finite sets over the (co)finite sets $\setAtoms_1, \ldots, \setAtoms_n$ of atoms,
has the following syntax
\begin{equation*}
	\setsymb_1, \setsymb_2 \cceq \{ e_i \} \mid \botset \mid \topset^1 \mid \ldots \mid \topset^n \mid \setsymb_1 \cap \setsymb_2 \mid \setsymb_1 \setminus \setsymb_2 \text{,}
\end{equation*}
where $\expr_i \in \setAtoms_i$ is any atom for any $1 \le i \le n$,
symbol $\botset$ represents the empty set,
$\topset^i$ represents all atoms of $\setAtoms_i$ for $1 \le i \le n$, 
and~$\cap$ and $\setminus$ are set intersection and set difference, respectively.
Other operators like union $\cup$ can be defined.
The class occurring in \dLCHP is $C(\Chan, \RVar, \NVar, \TVar)$.

For $C(\setAtoms_1, \ldots, \setAtoms_n)$, 
any $C(\subAtoms_1, \ldots, \subAtoms_k)$ with $\{ \subAtoms_1, \ldots, \subAtoms_k \} \subseteq \{ \setAtoms_1, \ldots \setAtoms_n \}$
forms a boolean algebra with binary operator $\cap$,
neutral element $\topset = \topset^1 \cup \ldots \cup \topset^k$ \wrt $\cap$,
where $\topset^i$ represents $\subAtoms_i$,
and the unary operation $\topset \setminus \usarg$ of parameter~$\usarg$.
Laws for the boolean algebra can be adopted to axiomatize $C(\subAtoms_1, \ldots, \subAtoms_k)$.

In formulas, (co)finite sets can be compared $\setsymb_1 = \setsymb_2$, and we include the element relation $\expr \in \cset$.
Axioms for the element relation over (co)finite sets unroll the relation into a finite conjunction as follows:
\begin{align*}
 	&\neg \expr \in \botset 
		&& \expr \in (\{ \expr[alt] \} \cap \setsymb) \leftrightarrow \expr = \expr[alt] \wedge \expr \in \setsymb \\
	&\expr_i \in \topset^i
		&& \expr \in ( \setsymb_1 \setminus \setsymb_2 ) \leftrightarrow \expr \in \setsymb_1 \wedge \neg \expr \in \setsymb_2 \\
	&\expr_i \not\in \topset^j \text{ for } i \neq j
\end{align*}

Equality is axiomatized in terms of the extensionality principle as usual:
\begin{equation*}
	\setsymb_1 = \setsymb_2 \leftrightarrow \fa{e} ( e \in \setsymb_1 \leftrightarrow e \in \setsymb_2)
\end{equation*}

\renewcommand{\doi}[1]{doi: \href{https://doi.org/#1}{\nolinkurl{#1}}}
\bibliographystyle{splncs04}
\bibliography{platzer,literature}

\else\fi % if report

\end{document}